\documentclass[12pt]{article}

\usepackage{hyperref}
\usepackage{epsfig}
\usepackage{epsf}
\usepackage{verbatim}
\usepackage{setspace}
\usepackage{multirow}
\usepackage{float}
\usepackage{times}
\usepackage{bm}
\usepackage{graphicx}
\usepackage{amsfonts}
\usepackage{amsmath}
\usepackage{amssymb}
\usepackage{lipsum}
\usepackage{enumerate}
\usepackage{filecontents}
\usepackage[flushleft]{threeparttable}
\usepackage[affil-it]{authblk}
\usepackage[utf8]{inputenc}
\usepackage[english]{babel}
\usepackage[section]{placeins}
\usepackage{soul,xcolor}
\usepackage{amsthm}
\usepackage{subcaption}
\usepackage{diagbox}
\makeatletter
\newcommand{\vast}{\bBigg@{4}}
\newcommand{\Vast}{\bBigg@{5}}
\makeatother

\newtheorem{theorem}{Theorem}

\newtheorem{proposition}{Proposition}

\newtheorem{definition}{Definition}

\clearpage\pagebreak\newpage 
\allowdisplaybreaks

\usepackage{mathtools}

\usepackage[round]{natbib}
\usepackage{tabularray}
\usepackage{adjustbox}

\begin{document}

\baselineskip=28pt %\vskip 5mm
\begin{center} {\LARGE{\bf{A Multivariate Skew-Normal-Tukey-$h$ Distribution}}}
\end{center}

\baselineskip=12pt 
\begin{center}\large
Sagnik Mondal\footnote[1]{
\baselineskip=10pt 
Statistics Program, King Abdullah University of Science and Technology, Thuwal 23955-6900, Saudi Arabia.\\
E-mail: sagnik.mondal@kaust.edu.sa, marc.genton@kaust.edu.sa} and Marc G.~Genton$^1$
\end{center}

\baselineskip=12pt %\vskip 10mm 
\centerline{\today} \vskip 2mm
{\bf Abstract:} We introduce a new family of multivariate distributions by taking the component-wise Tukey-$h$ transformation of a random vector following a skew-normal distribution. The proposed distribution is named the skew-normal-Tukey-$h$ distribution and is an extension of the skew-normal distribution for handling heavy-tailed data. We compare this proposed distribution to the skew-$t$ distribution, which is another extension of the skew-normal distribution for modeling tail-thickness, and demonstrate that when there are substantial differences in marginal kurtosis, the proposed distribution is more appropriate. Moreover, we derive many appealing stochastic properties of the proposed distribution and provide a methodology for the estimation of the parameters in which the computational requirement increases linearly with the dimension. Using simulations, as well as a wine and a wind speed data application, we illustrate how to draw inferences based on the multivariate skew-normal-Tukey-$h$ distribution.

\baselineskip=12pt
%\par\vfill\noindent
\vspace{.3cm}

\baselineskip=13pt

\section{Introduction}

In recent decades, there has been a growing interest in developing parametric multivariate distributions flexible enough to handle skewness and tail-thickness for various statistical applications. In a multivariate setup, two of the most popular methods to introduce both skewness and tail-thickness are:
\begin{enumerate}
    \item  {\bf Perturbation of symmetry} of an elliptically contoured distribution which is capable of capturing tail-thickness. Examples of such distributions include the multivariate skew-$t$ distribution \citep{2003.A.A.A.C.JRSSB} and the multivariate extended skew-$t$ distribution \citep{2010.R.B.A.V.M.G.G.M}. 
    \item {\bf Transformation} of a random vector following some elliptically contoured distribution for imposing skewness and tail-thickness. Examples of such transformations are the Tukey $g$-and-$h$ transformation \citep{2006.C.F.M.G.G.T} and the Sinh-Arcsinh transformation \citep{2009.C.M.J.A.P.B} in the multivariate case, and the Lambert's-$W$ transformation \citep{2011.G.M.G.TAAS} in the univariate case.
    % \item Location-scale mixtures of Gaussian distribution proposed by \cite{2014.F.F.D.W.SC} and \cite{2015.D.W.F.F.CSDA}. 
\end{enumerate}

 The primary parametric model obtained by perturbing the symmetry of an elliptically contoured distribution, which instigated the research in this area, is the multivariate skew-normal distribution introduced by \cite{1996.A.A.A.D.V.Biometrika}. Many distributions such as the multivariate skew-$t$ distribution, the multivariate extended skew-normal distribution, and the multivariate extended skew-$t$ distribution were built upon the foundation of the skew-normal distribution. These distributions can be viewed as special cases of the multivariate unified skew-elliptical distribution studied by \cite{2010.R.B.A.V.M.G.CJS}. For more on these types of distributions, readers are referred to the books by \cite{2004.M.G.CRCPress} and \cite{2013.A.A.A.C.CUP}, and to a recent review by \cite{2022.A.A.JMVA}. Since the skew-normal distribution is obtained by perturbing the symmetry of the Gaussian distribution and the skew-$t$ distribution is obtained by perturbing the symmetry of the Student's-$t$ distribution, the skew-normal distribution is not capable of handling tail-thickness while the skew-$t$ distribution is more apt for modeling heavy-tailed data. However, one shortcoming of the skew-$t$ distribution is that it cannot handle different tail-thickness for different marginals, since the tail-thickness is controlled only by one parameter. There has been a proposal by \cite{1968.K.S.M.TAMS} to introduce a   multivariate Student's-$t$ distribution with different tail-thickness parameters for different marginals. However, the probability density function (pdf) of the proposed distribution involves complicated hypergeometric functions that make inference with such a distribution very challenging. 

The second approach above for introducing skewed and heavy-tailed distribution is to use some non-linear transformation on a light-tailed elliptically symmetric random variable. The Lambert's-$W$ transformation, proposed by \cite{2011.G.M.G.TAAS} in the univariate case, can impose both skewness and tail-thickness on a Gaussian random variable using a single parameter. However, as this transformation is not one-to-one, the pdf of its multivariate extension becomes almost impossible to track down, especially for higher dimensions. \cite{2015.G.M.G.TSWJ} solved this issue by slightly changing the Lambert's-$W$ transformation and made it one-to-one. This modified transformation is a generalized version of the Tukey-$h$ transformation. Although \cite{2015.G.M.G.TSWJ} proposed this new distribution in the univariate setting, he only briefly mentioned how it can be extended to the multivariate setting by applying this transformation component-wise. Other examples include the Sinh-Arcsinh (SAS) transformation and the Tukey $g$-and-$h$ transformation which are monotonic and control skewness and tail-thickness with separate parameters. \cite{2006.C.F.M.G.G.T} presented a multivariate $g$-and-$h$ distribution which is based on the component-wise Tukey's $g$-and-$h$ transformation of a random vector following a Gaussian distribution. As a result, it permits different kurtosis for different marginals. However, one drawback of this distribution is drawing inferences. Since the inverse of Tukey's $g$-and-$h$ transformation does not have a closed form, the likelihood function cannot be readily calculated. Moreover, for parameter estimation, some definitions of multivariate quantiles are needed. This can be computationally challenging when the dimension is high because the number of directions in which the quantiles have to be computed grows exponentially with dimension. \cite{2009.C.M.J.A.P.B} discussed mainly the univariate SAS distribution and its various stochastic and inferential properties. The idea of the multivariate expansion of this family has also been discussed by \cite{2009.C.M.J.A.P.B}. It consists in using the transformation on the marginals of a standardized but correlated multivariate Gaussian distribution. A similar approach has been taken by \cite{2016.F.R.E.O.J.H.BJPS} who proposed a distribution that is capable of modeling higher skewness than the original SAS distribution by applying the two-piece transformation to the symmetric SAS distribution. \cite{2020.Y.Y.J.J.M.G.JJSDS} used the SAS distribution in the context of a bivariate random field for wind data and discussed how to draw inference based on it. However, inference in the multivariate scenario is yet to be thoroughly explored.

In this article, we propose a new multivariate distribution by combining these two techniques, the perturbation of symmetry for skewness and the transformation for tail-thickness. We introduce the skew-normal-Tukey-$h$ distribution by taking the Tukey-$h$ transformation on the components of a skew-normal random vector to introduce tail-thickness on each component. Moreover, by changing the marginal kurtosis parameter, we can have different kurtosis for different marginals. We study some basic statistical properties of the skew-normal-Tukey-$h$ distribution. Furthermore, we discuss how to draw inferences based on this distribution. We compare the proposed distribution with the skew-$t$ distribution since both of them are extensions of the skew-normal distribution for handling heavy-tailed data. Finally, we justify in which scenarios the skew-normal-Tukey-$h$ distribution is more appropriate compared to the skew-$t$ distribution using a simulation study and two data applications. 

It should be pointed out that the aforementioned two methods for constructing skewed and heavy-tailed distributions are not exhaustive. There exists a variety of proposals in the statistics literature. For example, distributions studied by \cite{2001.M.B.D.D.JMVA} and \cite{2004.J.W.J.B.M.G.G.SS} are very similar to the definition of the skew-normal distribution. \cite{2005.M.G.G.N.L.A.O.T.I.O.S.M} proposed a  definition of generalized skew-elliptical distributions which bring such different skewed distributions defined by perturbation of symmetry under one umbrella. Another avenue for the introduction of skewness and tail-thickness was explored by \cite{2014.F.F.D.W.SC} and further generalized by \cite{2015.D.W.F.F.CSDA} under the name of location-scale mixtures of Gaussian distributions. Various other non-Gaussian distributions for modeling skewed and heavy-tailed data can also be obtained using the theory of copulas \citep{1959.M.S.PISUP}. We refer interested readers to the books by \cite{1997.H.J.CRC} and \cite{2007.R.B.N.SSBM}, and the references therein, for more details on copulas. These are some other examples of parametric families proposed for modeling various skewed and heavy-tailed or light-tailed data.

The rest of the article is organized as follows. In Section 2, we formally define the skew-normal-Tukey-$h$ distribution, whereas various of its stochastic properties are discussed in Section 3. In Section 4, we illustrate how to draw inferences based on the skew-normal Tukey-$h$ distribution. In Sections 5 and 6, we present simulation studies and two applications to wine data and to wind speed data showing when the skew-normal-Tukey-$h$ distribution is more appropriate compared to the skew-$t$ distribution. Finally, in Section 7, we conclude our article and discuss some avenues for future research work. 
\vspace{-.4cm}

\section{Multivariate Skew-Normal-Tukey-$h$ Distribution}

In this section, we define the multivariate skew-normal-Tukey-$h$ distribution. We start by defining an alternative parameterization of the multivariate skew-normal distribution. 

\subsection{Skew-Normal Distribution}\label{SNdistrib}

The multivariate skew-normal distribution was introduced by \cite{1996.A.A.A.D.V.Biometrika} and later studied in \cite{1999.A.A.A.C.JRSSB}. A random vector $\bm Y \in \mathbb{R}^p$ is said to have a multivariate skew-normal distribution with location parameter $\bm \xi \in \mathbb{R}^p$, symmetric positive definite scale parameter $\bm \Omega \in \mathbb{R}^{p\times p}$, and skewness parameter $\bm \alpha \in \mathbb{R}^p$, if its pdf is
\begin{equation}\label{eq:ASN_mpdf}
f_{\bm Y}(\bm y) = 2 \phi_p\left(\bm y;\bm \xi,\bm \Omega\right) \Phi  \{\bm \alpha^\top \bm \omega^{-1} (\bm y - \bm \xi)\}, \quad \bm y \in \mathbb{R}^p,
\end{equation}
where $\phi_p (\cdot;\bm \mu,\bm \Sigma)$ is the pdf of a $p$-dimensional normal distribution with mean $\bm \mu \in \mathbb{R}^p$ and positive definite covariance matrix $\bm \Sigma \in \mathbb{R}^{p \times p}$, and $\bm \omega = \textup{diag}(\bm \Omega)^{1/2}$. Here, and from now on, we call this distribution with the parameterization in Equation (\ref{eq:ASN_mpdf}) the Azzalini skew-normal ($\mathcal{ASN}$) distribution and we denote it by $\bm Y \sim \mathcal{ASN}_p(\bm \xi,\bm \Omega,\bm \alpha)$. 

As used in \cite{2023.S.M.R.B.A.V.M.G.G.SP}, the $\mathcal{ASN}_p(\bm \xi,\bm \Omega,\bm \alpha)$ distribution can be reparameterized by means of the relations $\bm \Omega = \bm \Psi + \bm \eta \bm \eta^\top$ and $\bm \alpha = (1 + \bm \eta^\top \bm \Psi^{-1} \bm \eta)^{-1/2}  \bm \omega \bm \Psi^{-1} \bm \eta,$ where $\bm \Psi \in \mathbb{R}^{p \times p}$ is a symmetric positive definite matrix, $\bm \eta \in \mathbb{R}^p$ and
$\bm \omega=\textup{diag}(\sqrt{\Psi_{11}+\eta_1^2},\ldots,\sqrt{\Psi_{pp}+\eta_p^2})$, with $\Psi_{ii}$ and $\eta_i$ being the $i$th diagonal element of $\bm \Psi$ and $\bm \eta$, respectively, for $i = 1,\ldots,p$. Conversely, by letting $\bm \omega = \textup{diag}(\bm \Omega)^{1/2}$, $\bar{\bm \Omega}=\bm \omega^{-1}\bm \Omega \bm \omega^{-1}$ and  $\bm \delta =({{1+\bm \alpha^\top \bar{\bm \Omega} \bm \alpha}})^{-1/2}{\bar{\bm \Omega} \bm \alpha}$, we have
$\bm \Psi = \bm \omega(\bar{\bm \Omega}^{-1}+\bm \alpha \bm \alpha^\top)^{-1} \bm \omega = \bm \omega(\bar{\bm \Omega}-\bm \delta \bm \delta^\top) \bm \omega$ and $\bm \eta = \bm \omega \bm \delta.$ With this alternative parameterization, the pdf of $\bm Y$ from Equation \eqref{eq:ASN_mpdf} is
\begin{equation}\label{eq:SN_mpdf}
f_{\bm Y}(\bm y) = 2 \phi_p\left(\bm y;\bm \xi,\bm \Psi + \bm \eta \bm \eta^\top\right) \Phi \Bigg \{\dfrac{\bm \eta^\top\bm \Psi^{-1}(\bm y - \bm \xi)}{\sqrt{1+\bm \eta^\top \bm \Psi^{-1} \bm \eta}} \Bigg\}, \quad \bm y \in \mathbb{R}^p.
\end{equation}
\cite{1996.A.A.A.D.V.Biometrika} used this parameterization up to minor differences. Moreover, \cite{2001.C.J.A.K.S.Financialmodelling}, \cite{2004.C.J.A.S.E.D}, and \cite{2005.C.J.A.T.E.J.F} have also used the same parameterization. With this parameterization, a $p$-variate random vector $\bm Y$ is said to have a skew-normal $(\mathcal{SN})$ distribution with location parameter $\bm \xi \in \mathbb{R}^p$, symmetric positive definite scale matrix $\bm \Psi \in \mathbb{R}^{p \times p}$, and skewness parameter $\bm \eta \in \mathbb{R}^p$ if its pdf is given by Equation \eqref{eq:SN_mpdf}. We denote it by $\bm Y \sim \mathcal{SN}_p(\bm \xi, \bm \Psi,\bm \eta)$. 

Many interesting properties of the $\mathcal{SN}$ distribution with the parameterization in Equation \eqref{eq:SN_mpdf} have been derived in \cite{2023.S.M.R.B.A.V.M.G.G.SP}. The following results are given here as they will be useful later on, while their proofs can be found in \cite{2023.S.M.R.B.A.V.M.G.G.SP}:
\begin{itemize}
    \item \textit{Stochastic representation of $\mathcal{SN}$ distribution:} If $\bm Y \sim \mathcal{SN}_p(\bm \xi, \bm \Psi,\bm \eta)$, then  $\bm Y = \bm \xi +U \bm \eta+ \bm W$, where $U$ and $\bm W$ are independently distributed, with half-normal $U$ denoted by $U \sim \mathcal{HN}(0,1)$, and $\bm W \sim \mathcal{N}_p(\bm 0,\bm \Psi)$.
    \item \textit{Affine transformation of the $\mathcal{SN}$ distribution:} If $\bm Y \sim \mathcal{SN}_p(\bm \xi, \bm \Psi,\bm \eta)$, then for any fixed vector $\bm a \in \mathbb{R}^q$ and any fixed matrix $\bm B \in \mathbb{R}^{q \times p}$ of full row rank and $q \leq p$: $\bm a + \bm B \bm Y \sim \mathcal{SN}_q(\bm a + \bm B \bm \xi, \bm B \bm \Psi \bm B^\top, \bm B \bm \eta)$.
    \item \textit{Marginal distributions of the $\mathcal{SN}$ distribution:} Let $\bm Y \sim \mathcal{SN}_p(\bm \xi,\bm \Psi,\bm \eta)$ and consider the partition of $\bm Y = (\bm Y_1^\top, \bm Y_2^\top)^\top$ with $\bm Y_i$ of size $p_i$ ($i=1,2$) and such that $p_1+p_2 = p$, with corresponding partitions of the parameters in blocks of matching sizes. Then $\bm Y_i \sim \mathcal{SN}_{p_i} (\bm \xi_i,\bm \Psi_{ii}, \bm \eta_i), i= 1,2$.
\end{itemize}
The $\mathcal{ASN}$ and $\mathcal{SN}$ parameterizations describe the same distribution but the simplicity of the marginal distributions in the $\mathcal{SN}$ parameterization (see above) will prove useful for inferential purposes later on.

\subsection{Skew-Normal-Tukey-$h$ Distribution}

We introduce tail-thickness in the skew-normal distribution by taking the Tukey-$h$ transformation of all the components of a random vector following a $\mathcal{SN}$ distribution. The Tukey-$h$ transformation is 
\begin{equation}\label{eq:Tukey-h_transformation}
    \tau_h(x) = x \exp(hx^2/2), \quad x \in \mathbb{R}, \quad h \geq 0.
\end{equation}
Moreover, for $\bm x = (x_1,\ldots,x_p)^\top  \in \mathbb{R}^p$, we define
\begin{equation}\label{eq:multivariate_Tukey-h_transformation}
    \bm \tau_{\bm h}(\bm x) =  \{\tau_{h_1}(x_1),\ldots,\tau_{h_p}(x_p)\}^\top, \quad \bm h = (h_1,\ldots,h_p)^\top,\, h_i\geq 0, i = 1,\ldots,p.
\end{equation}

\theoremstyle{definition}
\begin{definition}[Skew-normal-Tukey-$h$ distribution]
A random vector $\bm Y \in \mathbb{R}^p$ with the stochastic representation $\bm Y = \bm \xi + \bm \omega \bm \tau_{\bm h}(\bm Z)$, where $\bm Z \sim \mathcal{SN}_p(\bm 0,\bar{\bm \Psi},\bm \eta)$ and $\bar{\bm \Psi}$ is a $p \times p$ correlation matrix, is said to have a multivariate skew-normal-Tukey-$h$ distribution.
Here $\bm \xi \in \mathbb{R}^p$ is the location parameter, $\bm \omega = \text{diag}(\omega_{11},\ldots,\omega_{pp})$ is a $p \times p$ diagonal scale matrix such that $\omega_{ii} > 0$,  $i=1,\ldots,p$, $\bm \eta \in \mathbb{R}^p$ is the skewness parameter, and $\bm h$ is the tail-thickness parameter vector such that $\bm h = (h_1,\ldots,h_p)^\top \in \mathbb{R}^p$, $h_i\geq 0$, $i = 1,\ldots,p$. We denote $\bm Y \sim \mathcal{SNTH}_p(\bm \xi, \bm \omega,\bar{\bm \Psi},\bm \eta,\bm h)$.
\end{definition}
We define the $\mathcal{SNTH}$ distribution with a correlation matrix $\bar{\bm \Psi}$ and a diagonal scale matrix $\bm \omega$. The $\bar{\bm \Psi}$ parameter governs the dependence structure in the model and $\bm \omega$ is a diagonal matrix consisting of the marginal scale parameters. To make all the parameters identifiable we restrict $\bar{\bm \Psi}$ to be a correlation matrix. It is immediate from the definition of the $\mathcal{SNTH}$ distribution that when $\bm h = \bm 0$ the $\mathcal{SNTH}$ distribution reduces to the $\mathcal{SN}$ distribution. The Tukey-$h$ transformation applied on the marginals of the skew-normal distribution imposes tail-thickness in the distribution. Moreover, since we can vary the components of the $\bm h$ parameter over the marginals, the resulting distribution can have different kurtosis for different marginals. In this way, we propose an extension of the skew-normal distribution, capable of handling different marginal tail-thickness. In that sense, the $\mathcal{SNTH}$ distribution is different from the skew-$t$ distribution. The skew-$t$ distribution can also be thought as an extension of the skew-normal distribution for modeling tail-thickness in the data, but it is incapable of capturing different kurtosis for different marginals. 

It should be pointed out that the proposed $\mathcal{SNTH}$ distribution belongs to the Lambert-$W$ $\times$ $F$ family of distributions \citep{2015.G.M.G.TSWJ}, where $F$ represents the cumulative distribution function of the skew-normal distribution. The main difference is that \cite{2015.G.M.G.TSWJ} proposed the location-scale Lambert-$W$ $\times$ $F$ distribution with $\mu_X = \mathbb{E}(X)$ as the location parameter and $\sigma_X = \sqrt{\mathbb{V} \text{ar} (X)}$ as the scale parameter, $X \sim F$, and the transformation is applied on $(X-\mu_X)/\sigma_X$. For defining the $\mathcal{SNTH}$ distribution, we start with a ``standard'' skew-normal distribution and apply the Tukey-$h$ transformation on it, and then we use a location-scale transformation on the transformed random variable.

\section{Properties of the $\mathcal{SNTH}$ Distribution}

We outline some basic probabilistic properties of the $\mathcal{SNTH}$ distribution such as its pdf,  cumulative distribution function (cdf), moments,  marginal and conditional distributions, and canonical form. Due to the $\mathcal{SNTH}$ definition using the $\mathcal{SN}$ distribution, many of the $\mathcal{SN}$ appealing properties get transferred to the $\mathcal{SNTH}$ distribution. This is one of the reasons we defined the $\mathcal{SNTH}$ with the $\mathcal{SN}$ distribution parameterized in Equation~(\ref{eq:SN_mpdf}).

\subsection{Probability Density Function of $\mathcal{SNTH}$}

In the next proposition we present the pdf of the $\mathcal{SNTH}$ distribution. The univariate $\mathcal{SNTH}$ pdf can be found using Theorem 1 of \cite{2015.G.M.G.TSWJ} using $F$ as the skew-normal distribution. We extend Theorem 1 of \cite{2015.G.M.G.TSWJ} with $F$ as the skew-normal distribution to the multivariate setup in the next proposition.
\begin{proposition}\label{prop:density_SNTH}
The pdf of $\bm Y \sim \mathcal{SNTH}_p(\bm \xi, \bm \omega,\bar{\bm \Psi},\bm \eta,\bm h)$ is, for $\bm y \in \mathbb{R}^p$:
\begin{equation}\label{eq:snth-density}
\hspace{-.3cm}\begin{split}
f_{\bm Y} (\bm y) &= 2 \phi_p \{ \bm g(\bm y) ; \bm 0, (\bar{ \bm \Psi} + \bm \eta \bm \eta^\top) \} \Phi \Bigg \{ \dfrac{\bm \eta^\top \bar{ \bm \Psi}^{-1} \bm g(\bm y)}{\sqrt{1+\bm \eta^\top \bar{ \bm \Psi}^{-1} \bm \eta}} \Bigg \} \prod_{i = 1}^p \Bigg \{ \dfrac{1}{\omega_{ii}} \left( \dfrac{\exp[\frac{1}{2} W_0 \{ h_i (\frac{y_i - \xi_i}{\omega_{ii}})^2 \}]}{h_i (\frac{y_i - \xi_i}{\omega_{ii}})^2 + \exp[ W_0 \{ h_i (\frac{y_i - \xi_i}{\omega_{ii}})^2 \}]} \right) \Bigg \},  
\end{split}
\end{equation}
where $\bm g(\bm y) = \{ g_1(y_1) ,\ldots, g_p(y_p) \}^\top$, $g_i(y_i) = (\frac{y_i - \xi_i}{\omega_{ii}}) \exp[-\frac{1}{2} W_0 \{ h_i (\frac{y_i - \xi_i}{\omega_{ii}})^2 \}]$, $i = 1,\ldots,p$, and $W_0(\cdot)$ is the principal branch of the Lambert's-$W$ function.
\end{proposition}
\begin{proof}
Consider the transformation $z = x \exp (hx^2 /2)$. Then $hz^2 = h x^2 \exp(h x^2) \Rightarrow hx^2 = W_0 (hz^2) \Rightarrow x = z \exp \{ -W_0 (hz^2)/2\}$,
where $W_0(\cdot)$ is the principal branch of the Lambert's-$W$ function \citep{1996.C.W}. This essentially means that $W_0(\cdot)$ is the inverse function of the function $f(x)= x\exp(x)$, $x\in \mathbb{R}$. Although the inverse of $f(x)$ is not unique when $x<0$, it is unique when $x>0$. For us the argument of $W_0(\cdot)$ is $hz^2 \geq 0$, which makes the inverse of the Tukey-$h$ transformation unique (see also Lemma 5 in \cite{2015.G.M.G.TSWJ}). Hence, the inverse of the Tukey-$h$ transformation (\ref{eq:Tukey-h_transformation}) is 
\begin{equation}\label{eq:inv_tukey_h}
  \tau_h^{-1}(z) = z \exp \{ -W_0 (hz^2)/2\},  
\end{equation}
and it is a one-to-one function as it should be since $\tau_h(z)$ is one-to-one for $h\geq 0$. Moreover, 
\begin{align*}
    \dfrac{\partial}{\partial z} \tau^{-1}_h(z) = \dfrac{\exp\{W_0 (hz^2)/2\}}{hz^2 + \exp \{W_0(hz^2)\}},
\end{align*}
and is obtained using the fact that $W_0 ^ \prime(z) = 1/[z+\exp \{W_0(z)\}]$. With the form of $\tau^{-1}_h(z)$ and $\frac{\partial \tau^{-1}_h(z)}{\partial z}$ it is straightforward to show that the pdf of $\bm Y$ reduces to Equation \eqref{eq:snth-density}.\qed
\end{proof}

The pdf of the $\mathcal{SNTH}$ distribution is given in closed form in Proposition \ref{prop:density_SNTH} and it involves the principal branch $W_0(\cdot)$ of the Lambert's-$W$ function. Although $W_0(\cdot)$ does not have a closed form, it is a well studied function and the function has been already implemented in many softwares, including in \cite{R-Core-Team:2022aa} in the \texttt{LambertW} package by \cite{2011.G.M.G.TAAS}. This is an advantage of the $\mathcal{SNTH}$ distribution over the multivariate Tukey $g$-and-$h$ distribution in the sense that the inverse of the Tukey $g$-and-$h$ transformation is not in a closed form. As a result, the computation of the probability density function and the log-likelihood of the $\mathcal{SNTH}$ distribution is somewhat simpler compared to that of the multivariate Tukey $g$-and-$h$ distribution. 

\begin{figure}[b!]
\begin{center}
\centering
\begin{subfigure}{0.3\textwidth}
  \centering
  \includegraphics[width=\linewidth]{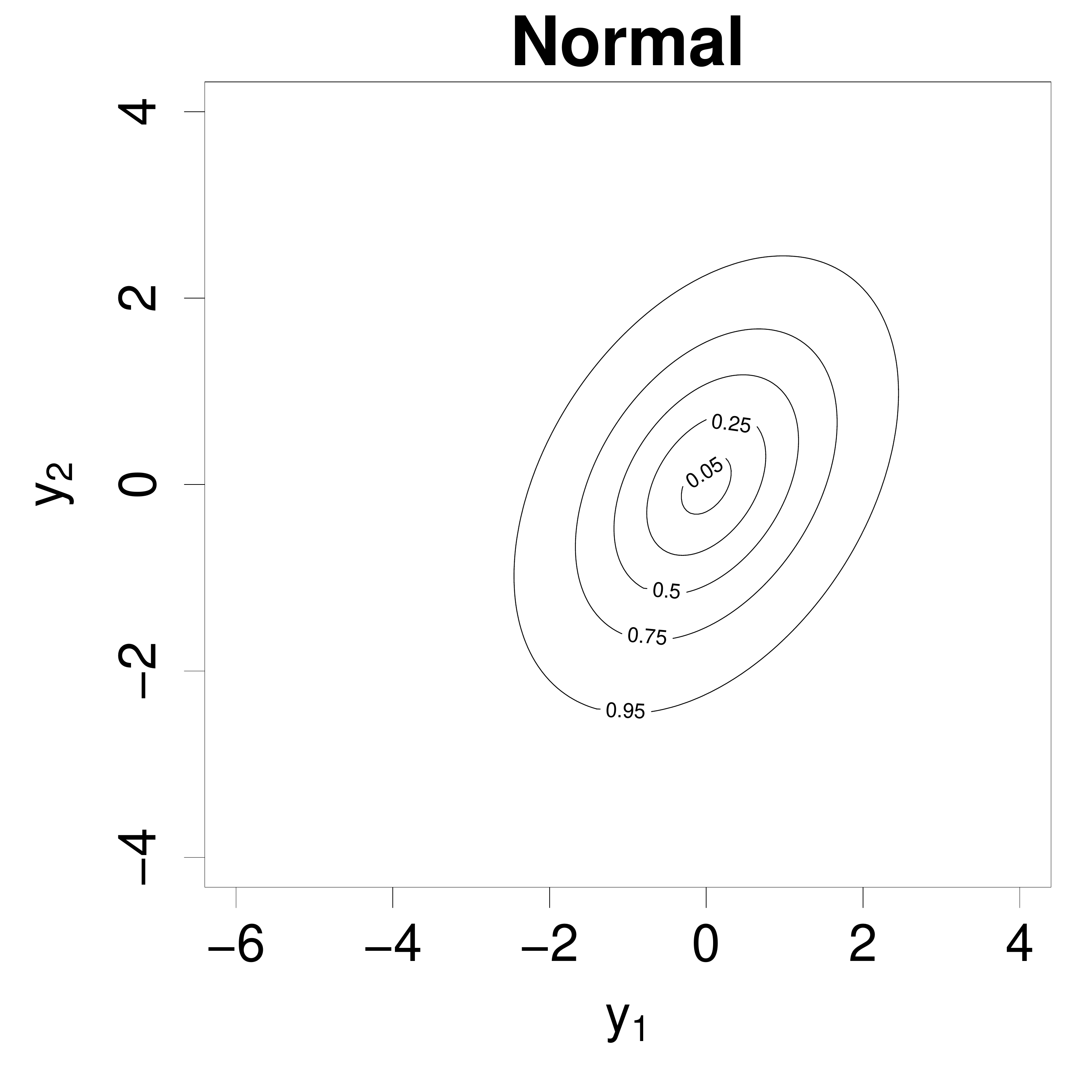}
  %\vspace{-0.35in}
  
 % \label{fig:boxplot_2.1}
\end{subfigure}%
\begin{subfigure}{0.3\textwidth}
  \centering
  \includegraphics[width=\linewidth]{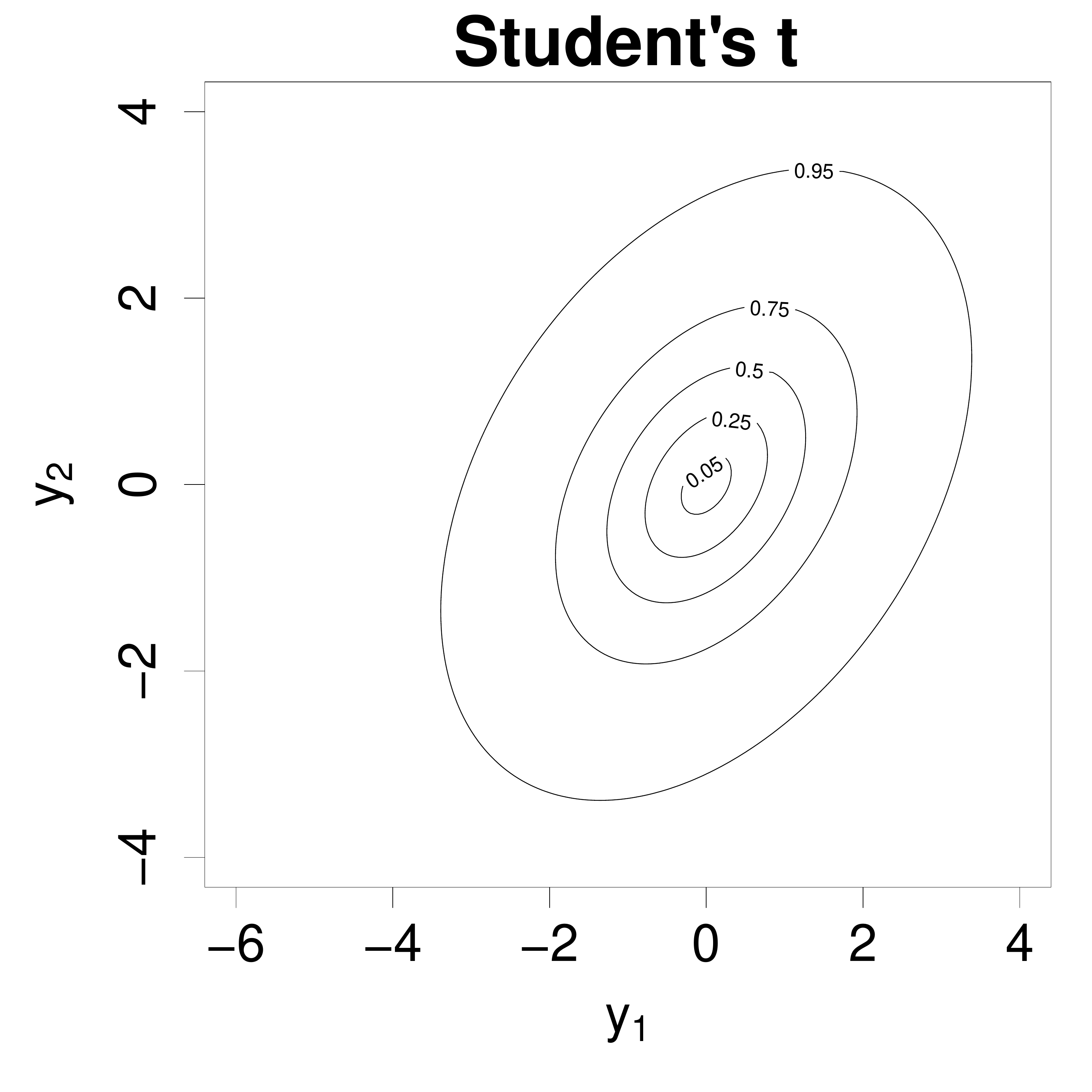}
  %\vspace{-0.35in}
  
 % \label{fig:boxplot_2.2}
\end{subfigure}
\begin{subfigure}{0.3\textwidth}
  \centering
  \includegraphics[width=\linewidth]{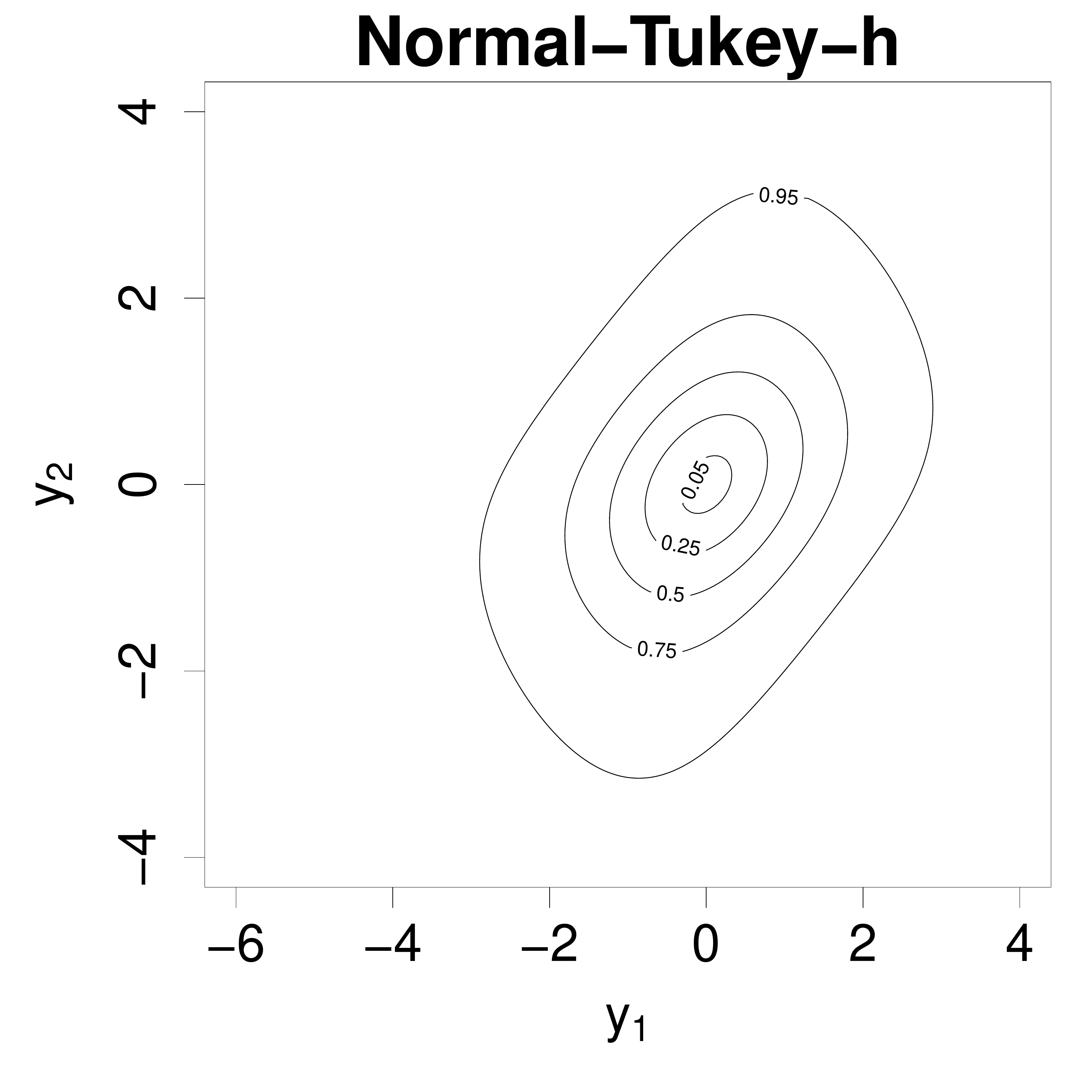}
  %\vspace{-0.35in}
  
 % \label{fig:boxplot_2.2}
\end{subfigure}
\begin{subfigure}{0.3\textwidth}
  \centering
  \includegraphics[width=\linewidth]{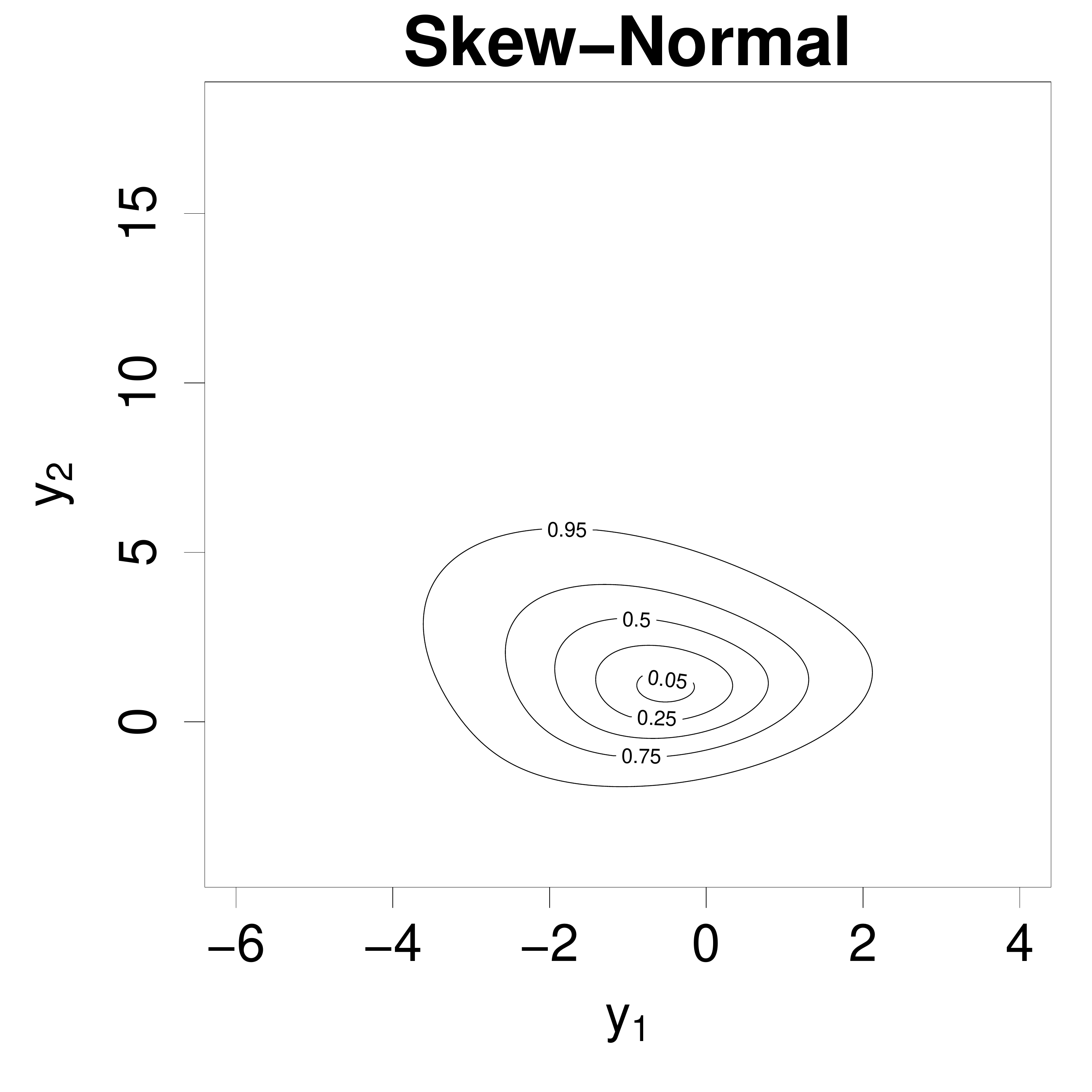}
  \vspace{-0.35in}
  
 % \label{fig:boxplot_2.1}
\end{subfigure}%
\begin{subfigure}{0.3\textwidth}
  \centering
  \includegraphics[width=\linewidth]{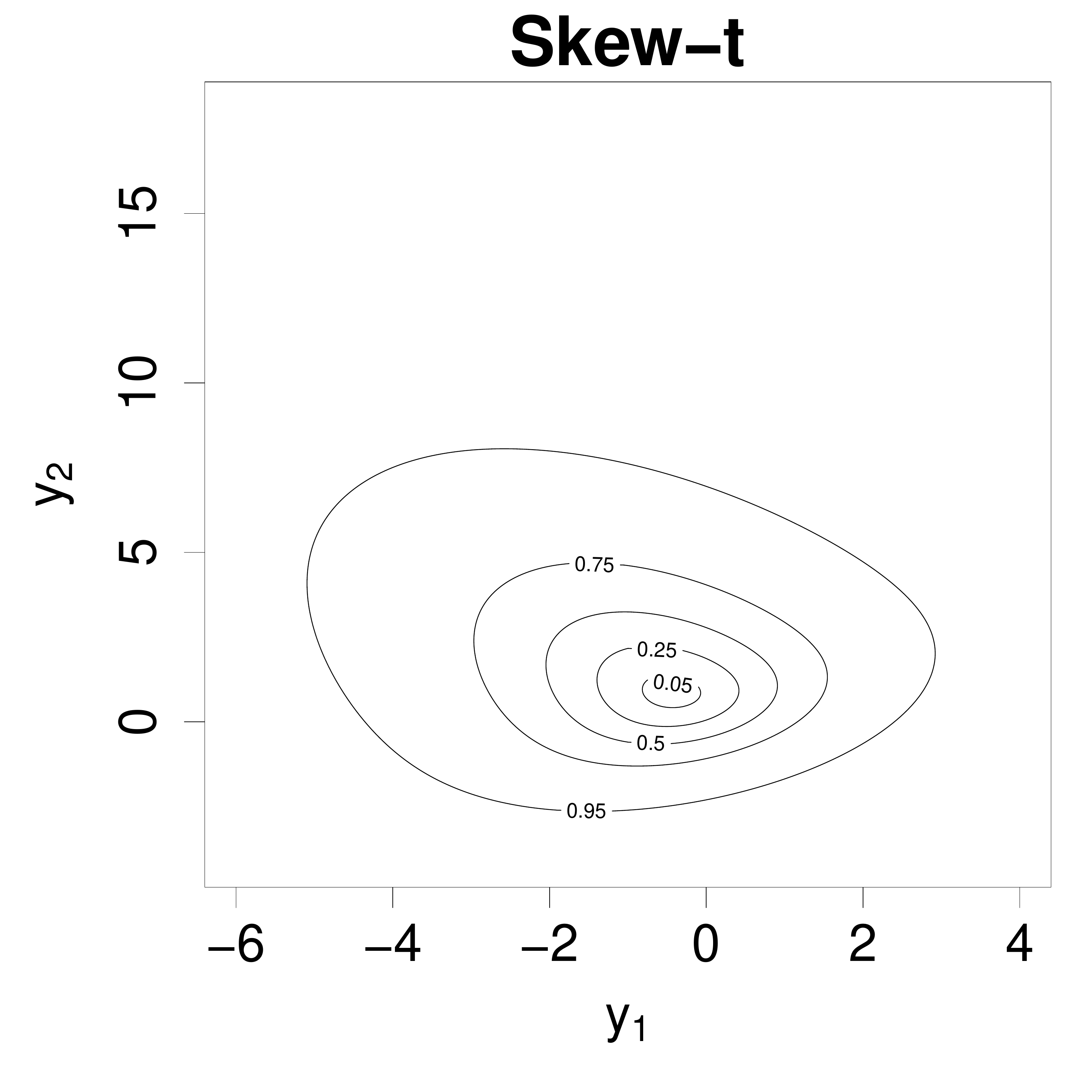}
  \vspace{-0.35in}
  
 % \label{fig:boxplot_2.2}
\end{subfigure}
\begin{subfigure}{0.3\textwidth}
  \centering
  \includegraphics[width=\linewidth]{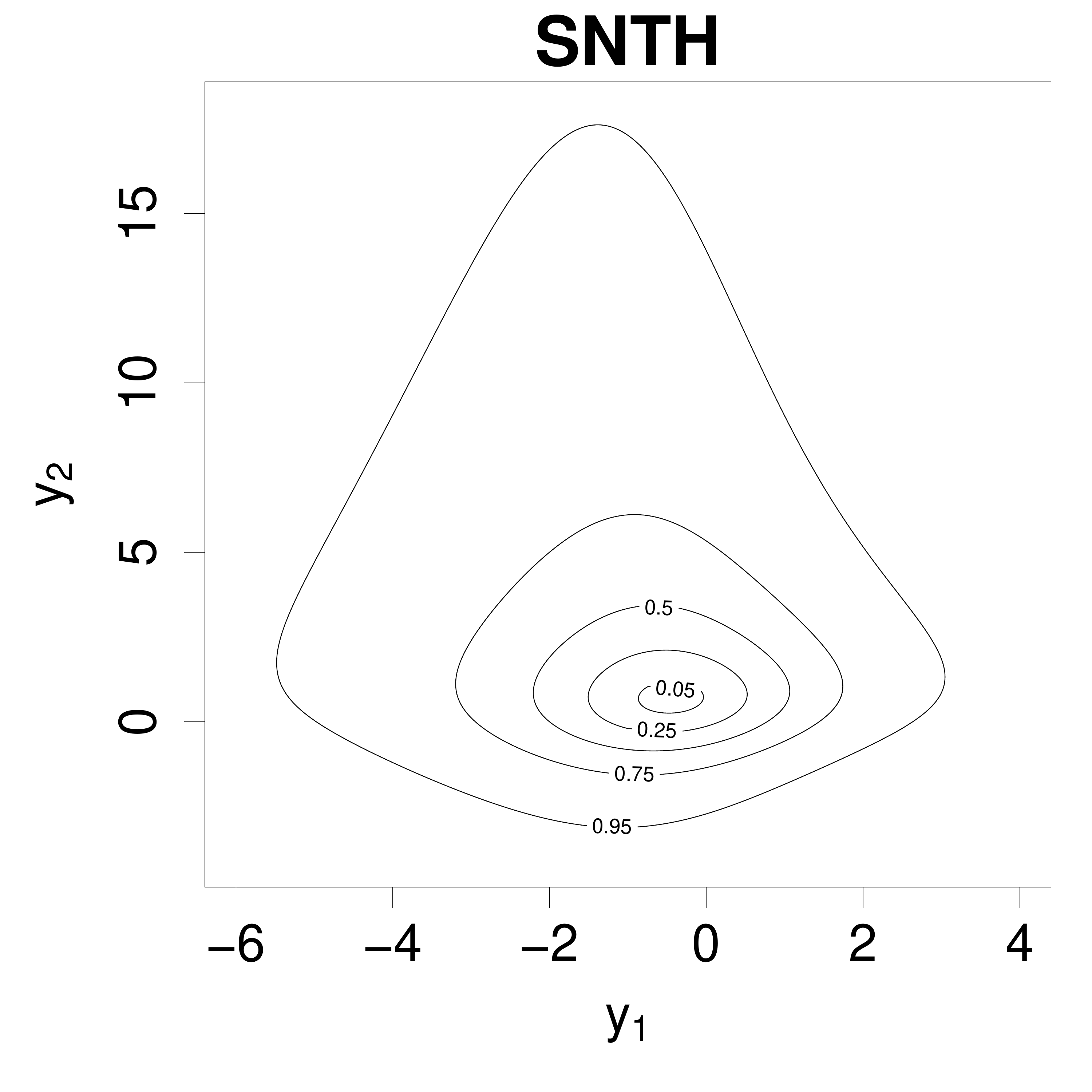}
  \vspace{-0.35in}
  
 % \label{fig:boxplot_2.2}
\end{subfigure}
\caption{Bivariate probability density contours of various distributions. Contours are given so that their coverage probabilities are approximately 0.05, 0.25, 0.5, 0.75, and 0.95.}
\label{fig:differnt_contours}
\end{center}
\end{figure}

To illustrate the effects of the skewness and the tail-thickness parameters of the $\mathcal{SNTH}$ distribution, we present the contour plots of $\mathcal{SNTH}_2(\bm 0,\text{diag}(1,1),\bar{\bm \Psi},\bm \eta,\bm h)$ probability densities with $\bar{\bm \Psi} = \begin{psmallmatrix} 1 & 0.4 \\ 0.4 & 1 \end{psmallmatrix}$ for four different pairs of $\bm \eta$ and $\bm h$: $\bm \eta = (0,0)^\top$ and $\bm h = (0,0)^\top$ corresponding to a normal density; $\bm \eta = (0,0)^\top$ and $\bm h = (0.05,0.1)^\top$ corresponding to a Normal-Tukey-$h$ density; $\bm \eta = (-1,2)^\top$ and $\bm h = (0,0)^\top$ corresponding to a $\mathcal{SN}$ density; and $\bm \eta = (-1,2)^\top$ and $\bm h = (0.05,0.1)^\top$ corresponding to a $\mathcal{SNTH}$ density. For comparison we also plot the density contours of a skew-$t$ distribution with $\bm \xi = (0,0)^\top$, $\bm\Omega=\begin{psmallmatrix} 2 & -1.6\\ -1.6 & 5 \end{psmallmatrix}$, $\bm \alpha = (-1.02,2.15)^\top$, and $\nu = 5$, and a Student's $t$ distribution with these same parameters (i.e., the same skew-$t$ with $\bm \alpha = \bm 0$). The $\bm \Omega$ and $\bm \alpha$ parameters are obtained so that they correspond to $\bar{\bm \Psi} = \begin{psmallmatrix} 1 & 0.4 \\ 0.4 & 1 \end{psmallmatrix}$ and $\bm \eta = (-1,2)^\top$ using the relationship between the parameters of the $\mathcal{ASN}$ and the $\mathcal{SN}$ parameterizations. All the density contours are plotted in Figure \ref{fig:differnt_contours}. The contours are drawn for the levels with approximate coverage probabilities $0.05$, $0.25$, $0.5$, $0.75$, and $0.95$. The density contour plots in the first row correspond to the density contours of the second row when the corresponding skewness parameters are set to zero. Although the contours in the first row are all symmetric, their symmetry differs from each other. More precisely, in Figure \ref{fig:differnt_contours}, the normal and the Student's $t$ probability contours are centrally symmetric whereas the normal-Tukey-$h$ probability contours are sign-invariant symmetric, which is a special case of central symmetry. It can be concluded from Figure \ref{fig:differnt_contours} that the shapes of the Student's $t$ and skew-$t$ density contours are similar to that of the normal and the skew-normal densities, respectively, with more spacing in-between the different levels for the formers due to thicker tails. The contours of the normal-Tukey-$h$ density and the $\mathcal{SNTH}$ density look similar to the normal and the skew-normal density contours, respectively, but the former have been stretched along the two axes. Since the extent of this stretching can be different along the two axes, the $\mathcal{SNTH}$ density contours can represent a  variety of shapes with changes in the skewness and the tail-thickness parameters. 

\subsection{Cumulative Distribution Function of $\mathcal{SNTH}$}

The cdf of the $\mathcal{SNTH}$ distribution can be obtained in closed form involving the principal branch $W_0(\cdot)$ of the Lambert's-$W$ function as shown next.
\begin{proposition}
    The cdf of $\bm Y \sim \mathcal{SNTH}_p (\bm \xi,\bm \omega,\bar{\bm \Psi}, \bm \eta, \bm h)$ is $F_{\bm Y}(\bm y) = 2 \Phi_{p+1}(\bm y_{**};\bm 0,\bm \Omega_{**})$
where $\Phi_{p+1}$ is the multivariate Gaussian cdf of dimension $p+1$,  $\bm y_{**} = \left\{ \tau_{h_1}^{-1} \left( \dfrac{y_1 - \xi_1}{\omega_{11}} \right),\ldots, \tau_{h_p}^{-1} \left( \dfrac{y_p - \xi_p}{\omega_{pp}} \right),0 \right\}^\top$ and $\bm \Omega_{**} = \begin{pmatrix} \bar{\bm \Psi} + \bm \eta \bm \eta^\top & -\bm \eta\\ -\bm \eta^\top & 1 \end{pmatrix}$.
\end{proposition}
\begin{proof} Let 
     $\bm Y= \bm \xi + \bm \omega \bm \tau_{\bm h}(\bm Z)$, where $\bm Z \sim \mathcal{SN}_p (\bm 0, \bar{\bm \Psi}, \bm \eta)$. Then the cdf of  $\bm Y$ is 
\begin{align*}
    F_{\bm Y}(\bm y) &= \mathbb{P}(Y_1 \leq y_1,\ldots,Y_p \leq y_p) 
    = \mathbb{P} \left[ Z_1 \leq \tau_{h_1}^{-1} \left( \dfrac{y_1 - \xi_1}{\omega_{11}} \right),\ldots,Z_p \leq \tau_{h_p}^{-1} \left( \dfrac{y_p - \xi_p}{\omega_{pp}} \right) \right]\\
    &= F_{\bm Z} \left\{ \tau_{h_1}^{-1} \left( \dfrac{y_1 - \xi_1}{\omega_{11}} \right),\ldots, \tau_{h_p}^{-1} \left( \dfrac{y_p - \xi_p}{\omega_{pp}} \right) \right\}
    = 2 \Phi_{p+1}(\bm y_{**};\bm 0,\bm \Omega_{**}), \quad \bm y \in \mathbb{R}^p,
\end{align*}
where $\tau_{h}^{-1}(z)$ is given in Equation (\ref{eq:inv_tukey_h}). The cdf of $\bm Z$, $F_{\bm Z}(\cdot)$, is obtained using Proposition 12 of \cite{2023.S.M.R.B.A.V.M.G.G.SP}.\qed
\end{proof}

\subsection{Marginal Distributions of $\mathcal{SNTH}$}

Similar to the $\mathcal{SN}$ distribution, the marginals of the $\mathcal{SNTH}$ distribution are also from the same family, as shown in the next proposition. 
\begin{proposition}\label{prop:marginal_SNTH}
Let $\bm Y \sim \mathcal{SNTH}_p (\bm \xi, \bm \omega,\bar{\bm \Psi},\bm \eta,\bm h)$ and consider the partition $\bm Y = (\bm Y_1^\top, \bm Y_2^\top)^\top$ with $\bm Y_i$ of size $p_i$ ($i=1,2$) and such that $p_1+p_2 = p$, with corresponding partitions of the parameters in blocks of matching sizes, as follows:
\begin{equation*}\label{eq:partition}
\bm \xi = \begin{pmatrix}
\bm \xi_1\\
\bm \xi_2
\end{pmatrix},
\bm \omega = \begin{pmatrix}
\bm \omega_{11} & \bm 0 \\
\bm 0 & \bm \omega_{22}
\end{pmatrix},
\bar{\bm \Psi} = \begin{pmatrix}
\bar{\bm \Psi}_{11} & \bar{\bm \Psi}_{12} \\
\bar{\bm \Psi}_{21} & \bar{\bm \Psi}_{22}
\end{pmatrix},
\bm \eta = \begin{pmatrix}
\bm \eta_1 \\ \bm \eta_2
\end{pmatrix}, \bm h  = \begin{pmatrix}
\bm h_1 \\ \bm h_2
\end{pmatrix}.
\end{equation*}
Then $\bm Y_i \sim \mathcal{SNTH}_{p_i} (\bm \xi_i, \bm \omega_{ii},\bar{\bm \Psi}_{ii},\bm \eta_i,\bm h_i)$, $i=1,2.$
\end{proposition}
\begin{proof}
Since, $\bm Y \sim \mathcal{SNTH}_p (\bm \xi , \bm \omega,\bar{\bm \Psi},\bm \eta, \bm h)$, then by definition there exists a random vector $\bm Z \sim \mathcal{SN}_p (\bm 0, \bar{\bm \Psi}, \bm \eta)$ such that $\bm Y = \bm \xi + \bm \omega \bm \tau_{\bm h} (\bm Z)$. Consider the partition $\bm Z = (\bm Z_1^\top, \bm Z_2^\top)^\top$, similar to $\bm Y$. Then, $\bm Z_i \sim \mathcal{SN}_{p_i} (\bm 0, \bar{\bm \Psi}_{ii},\bm \eta_i),~i=1,2$, and $\bm Y_i = \bm \xi_i + \bm \omega_{ii} \bm \tau_{\bm h_i}(\bm Z_i)$. Hence, $\bm Y_i \sim \mathcal{SNTH}_{p_i} (\bm \xi_i, \bm \omega_{ii},\bar{\bm \Psi}_{ii},\bm \eta_i,\bm h_i)$, $i=1,2.$ \qed
\end{proof}

Although the marginals of the $\mathcal{SNTH}$ remain in the same family, the same cannot be said for any general affine transformation of the $\mathcal{SNTH}$ distribution. The distribution of an arbitrary affine transformation of a $\mathcal{SNTH}$ random vector is not of a known type.

\subsection{Mean and Variance-Covariance of $\mathcal{SNTH}$}
The mean vector and the variance-covariance matrix of the $\mathcal{SNTH}$ distribution can be obtained in closed form. The next proposition presents these results. 
\begin{proposition}\label{prop:exp_var_cov}
Let $\bm Y \sim \mathcal{SNTH}_p(\bm \xi,\bm \omega, \bar{\bm \Psi}, \bm \eta, \bm h)$. The mean vector $\bm \mu=\mathbb{E}(\bm Y)$ and variance-covariance matrix $\bm\Sigma=(\sigma_{ij})=\mathbb{V} \text{ar} (\bm Y)$ are defined by:
\begin{align}
\mu_i&=\xi_i+   \omega_{ii} \sqrt{\dfrac{2}{\pi}} \dfrac{\eta_i}{\sqrt{1 - h_i}\{1 - h_i(1+\eta_i^2)\}},\quad \text{~if~} h_i < \dfrac{1}{1+\eta_i^2},\, \nonumber \\
    \sigma_{ii}&= \omega_{ii}^2 \left[ \dfrac{1+\eta_i^2}{\{ 1-  2h_i (1 +\eta_i^2)\}^{3/2}}-\dfrac{2}{\pi} \dfrac{\eta_i^2}{(1-h_i)\{1 - h_i(1+\eta_i^2)\}^2} \right], \quad \text{~if~} h_i < \dfrac{1}{2(1+\eta_i^2)}, 
 \nonumber \\
\sigma_{ij}&= \omega_{i}\omega_{j} \Bigg [\dfrac{\sqrt{\det(\bm A^{(ij)})}}{\sqrt{\det(\bar{\bm \Psi}_{i,j} + \bm \eta _{i,j} \bm \eta _{i,j} ^\top)}} a_{12}^{(ij)}-{\dfrac{2}{\pi}} \dfrac{\eta_i\eta_j}{\sqrt{(1 - h_i)(1 - h_j)} \{1 - h_i(1+\eta_i^2)\} \{1 - h_j(1+\eta_j^2)\}} \Bigg ],\nonumber \\
&~~~~~~~~~~~~~~~~~~~~~~~~~~~~~~~~~~~~~~~~~~~~~~~~~~~~~~~~~~~~~~~~~~~~~~~~~~~~~~~~~~~~~~~~~~~~~~~~~~~~~~~~~~~\text{~if~} \bm A^{(ij)} \text{~is~positive~definite}, \nonumber
\end{align}
where $\bm  \eta_{i,j} = (\eta_i, \eta_j)^\top,~ \bar{\bm \Psi}_{i,j} = \begin{psmallmatrix}
 1 & \bar{\Psi}_{ij} \\ \bar{\Psi}_{ij} & 1
 \end{psmallmatrix} $, $ \bm A^{(ij)} = \{(\bar{\bm \Psi}_{i,j} + \bm \eta_{i,j} \bm \eta_{i,j} ^\top)^{-1} - \bm H_{i,j} \} ^{-1}=\begin{pmatrix}
a_{11}^{(ij)} & a_{12}^{(ij)}\\ a_{12}^{(ij)} & a_{22}^{(ij)}
\end{pmatrix}$, $\bm H_{i,j} = \text{diag} (h_i,h_j)$, $i \neq j$, and $i,j = 1,\ldots,p$.
\end{proposition}
\begin{proof}
Since $\bm Y \sim \mathcal{SNTH}_p(\bm \xi,\bm \omega, \bar{\bm \Psi}, \bm \eta, \bm h)$, then $\bm Y$ can be written as $\bm Y= \bm \xi + \bm \omega \bm \tau_{\bm h} (\bm Z)$, where $\bm Z \sim \mathcal{SN}_p(\bm 0,\bar{\bm \Psi},\bm \eta)$. Then using the fact $Z_i \sim \mathcal{SN}(0,1,\eta_i)$:
\begin{equation*}
\begin{split}
\mathbb{E}\{\tau_{h_i}(Z_i)\} &= \int_{\mathbb{R}} x \exp(h_i x^2 /2) 2 \phi(x;0,1+\eta_i ^2 ) \Phi \left( \dfrac{\eta_i x}{\sqrt{1+\eta_i^2}} \right) \text{d} x\\
& =  \int_{\mathbb{R}} \dfrac{\sqrt{1+\eta_i^2}}{\sqrt{1-h_i (1+\eta_i^2)}} t \dfrac{2}{\sqrt{2 \pi} \sqrt{1+\eta_i^2}} \exp(-t^2/2) \Phi \left( \dfrac{\eta_i t}{\sqrt{1 - h_i (1+\eta_i^2)}} \right) \dfrac{\sqrt{1+\eta_i^2}}{\sqrt{1 - h_i (1+ \eta_i^2)}} \text{d} t \\
&~~~~ \left(\text{ using the change of variable }t = \dfrac{\sqrt{1 - h_i (1 +\eta_i^2)}}{\sqrt{1+\eta_i^2}} x \right)\\
&= \dfrac{\sqrt{1+\eta_i^2}}{1 - h_i(1+\eta_i^2)} \mathbb{E}(X_i) \quad \text{with } X_i \sim \mathcal{ASN} \left(0,1,\frac{\eta_i}{\sqrt{1-h_i(1+\eta_i^2)}} \right)\\
&= \sqrt{\dfrac{2}{\pi}} \dfrac{\eta_i}{\sqrt{1 - h_i}\{1 - h_i(1+\eta_i^2)\}}, \quad h_i(1+\eta_i^2)<1, ~ i = 1,\ldots,p;
\end{split}
\end{equation*}
% Thus, $\mathbb{E}(X_j) = \mathbb{E}\{\xi_j+\omega_{jj} \tau_{h_j}(Z_j)\} =\xi_j+   \omega_{jj} \sqrt{\dfrac{2}{\pi}} \dfrac{\eta_j}{\sqrt{1 - h_j}\{1 - h_j(1+\eta_j^2)\}}$,  if $h_j < \dfrac{1}{1+\eta_j^2}$, $j = 1,\ldots,p$.
\begin{equation*}
\begin{split}
\mathbb{E}[\{\tau_{h_i}(Z_i)\}^2] &= \int_{\mathbb{R}} x^2 \exp(h_i x^2 ) 2 \phi(x;0,1+\eta_i ^2 ) \Phi \left( \dfrac{\eta_i x}{\sqrt{1+\eta_i^2}} \right) \text{d} x \hspace{7.7cm}
\end{split}
\end{equation*}
\begin{equation*}
\begin{split}
&= \int_{\mathbb{R}} \dfrac{1+\eta_i^2}{1 - 2 h_i (1 +\eta_i^2)} t^2 \dfrac{2}{\sqrt{2 \pi} \sqrt{1+\eta_i^2}} \exp(-t^2/2) \Phi \left( \dfrac{\eta_i t}{\sqrt{1 -2 h_i (1+\eta_i^2)}} \right) \dfrac{\sqrt{1+\eta_i^2}}{\sqrt{1 - 2h_i (1+ \eta_i^2)}} \text{d} t \\
&~~~~ \left(\text{ using the change of variable }t = \dfrac{\sqrt{1 - 2h_i (1 +\eta_i^2)}}{\sqrt{1+\eta_i^2}} x \right)\\
&= \dfrac{1+\eta_i^2}{\{ 1-  2h_i (1 +\eta_i^2)\}^{3/2}} \mathbb{E}(X_i^2) \quad \text{with } X_i \sim \mathcal{ASN} \left(0,1,\frac{\eta_i}{\sqrt{1-2h_i(1+\eta_i^2)}} \right)\\
&= \dfrac{1+\eta_i^2}{\{ 1-  2h_i (1 +\eta_i^2)\}^{3/2}}, \quad 2h_i(1+\eta_i^2) <1, ~ i = 1,\ldots,p.
\end{split}
\end{equation*}
Hence:
\begin{equation*}
    \mathbb{V} \text{ar} \{ \tau_{h_i}(Z_i) \}  = \dfrac{1+\eta_i^2}{\{ 1-  2h_i (1 +\eta_i^2)\}^{3/2}}-\dfrac{2}{\pi} \dfrac{\eta_i^2}{(1-h_i)\{1 - h_i(1+\eta_i^2)\}^2}, ~ h_i < \dfrac{1}{2(1+\eta_i^2)}, ~ i = 1,\ldots,p.
\end{equation*}
\begin{equation*}
\begin{split}
\mathbb{E} \{\tau_{h_i}(Z_i) \tau_{h_j}(Z_j)\} &= \int_{\mathbb{R}^2} x_1 x_2 \exp\{ (h_i x_1^2 + h_j x_2^2)/2 \} 2 \phi_2(\bm x; \bm 0,  \bar{\bm \Psi}_{i,j} + \bm \eta_{i,j} \bm \eta_{i,j}^\top )\Phi \left( \dfrac{\bm \eta_{i,j} ^\top \bar{\bm \Psi}_{i,j}^{-1} \bm x}{\sqrt{1+ \bm \eta_{i,j}^\top \bar{\bm \Psi}_{i,j}^{-1} \bm \eta_{i,j}}} \right) \text{d} \bm x\\
&= \int_{\mathbb{R}^2} x_1 x_2 \dfrac{\sqrt{\det(\bm A^{(ij)})}}{\sqrt{\det(\bar{\bm \Psi}_{i,j} + \bm \eta _{i,j} \bm \eta _{i,j} ^\top)}} 2 \phi_2 (\bm x; \bm 0, \bm A^{(ij)}) \Phi \left( \dfrac{\bm \eta_{i,j} ^\top \bar{\bm \Psi}_{i,j}^{-1} \bm \omega_{\bm A^{(ij)}} \bm \omega_{\bm A^{(ij)}}^{-1} \bm x}{\sqrt{1+ \bm \eta_{i,j}^\top \bar{\bm \Psi}_{i,j}^{-1} \bm \eta_{i,j}}} \right) \text{d} \bm x\\
& = \dfrac{\sqrt{\det(\bm A^{(ij)})}}{\sqrt{\det(\bar{\bm \Psi}_{i,j} + \bm \eta _{i,j} \bm \eta _{i,j} ^\top)}} \mathbb{E}(X_i X_j)\\
& = \dfrac{\sqrt{\det(\bm A^{(ij)})}}{\sqrt{\det(\bar{\bm \Psi}_{i,j} + \bm \eta _{i,j} \bm \eta _{i,j} ^\top)}} \mathbb{E}(X_i X_j) a_{12}^{(ij)}, \quad\text{if}\bm A^{(ij)}~\text{is}~\text{positive definite}, ~i \neq j, ~ i,j = 1,\ldots,p,
\end{split}
\end{equation*}
where $(X_i,X_j)^\top \sim \mathcal{ASN}_2 \left(\bm 0,\bm A^{(ij)}, \frac{ \bm \omega_{\bm A^{(ij)}} \bar{\bm \Psi}_{i,j}^{-1}  \bm \eta_{i,j} }{\sqrt{1+ \bm \eta_{i,j}^\top \bar{\bm \Psi}_{i,j}^{-1} \bm \eta_{i,j}}} \right)$ and 
$\bm \omega_{\bm A^{(ij)}} = \{\text{diag}(\bm A^{(ij)})\}^{1/2}$.
The moments related to the $\mathcal{ASN}$ distribution are obtained from Chapter 2 (univariate) and Chapter 5 (multivariate) of \cite{2013.A.A.A.C.CUP}. The rest of the proof is straightforward and hence omitted.\qed
\end{proof}

To this point, we have closed-form expressions of the mean vector and the variance-covariance matrix for the $\mathcal{SNTH}$ distribution. However, we cannot have a closed-form expression for its moment generating function or characteristic function. This is because the distribution of any general affine transformation of the $\mathcal{SNTH}$ distribution is not known.

\subsection{Marginal Skewness and Kurtosis of $\mathcal{SNTH}$}

Here we discuss some results related to the skewness and kurtosis of the $\mathcal{SNTH}$ distribution. The Mardia's measures of multivariate skewness and kurtosis \citep{1970.V.K.M.Biometrika} for the $\mathcal{SNTH}$ distribution cannot be derived in closed form. However, their univariate counterparts can be derived. Similar to the skew-$t$ distribution, the Pearson's measures of skewness and excess-kurtosis are also unbounded for the univariate $\mathcal{SNTH}$ distribution, suggesting that it is also the case in the multivariate setting.
\begin{proposition}\label{prop:gamma_1_gamma_2}
The Pearson's measures of skewness and excess-kurtosis of $Y \sim \mathcal{SNTH}_1(0,1,1,\eta,h)$ are
$\gamma_1 = \mu_3/\mu_2^{3/2}$ and $\gamma_2 =\mu_4/\mu_2^{2} - 3$,
where $\mu_2 = \mathbb{V} \text{ar}(Y)$, $\mu_3 = \mathbb{E}\{Y- \mathbb{E}(Y)\}^3 = \mathbb{E}(Y^3) - 3 \mathbb{E}(Y^2) \mathbb{E}(Y) + 2\mathbb{E}(Y)^2$, $\mu_4 = \mathbb{E}\{Y- \mathbb{E}(Y)\}^4 = \mathbb{E}(Y^4) - 4 \mathbb{E}(Y^3) \mathbb{E}(Y) + 6 \mathbb{E}(Y^2) \mathbb{E}(Y)^2 - 3\mathbb{E}(Y)^4$ with:
\begin{align*}
    \mathbb{E}(Y^3) &= \sqrt{\dfrac{2}{\pi}} \dfrac{(1+\eta^2)^{3/2}}{\{1 - 3 h (1+\eta^2)\}^2} \Bigg [ \dfrac{2 \eta^3 + 3 \eta \{1 - 3h (1+\eta^2)\}}{\{ (1+\eta^2)(1- 3 h)\}^{3/2}} \Bigg ], \quad h< \dfrac{1}{3(1+\eta^2)},\\
    \mathbb{E}(Y^4) &= \dfrac{3(1+\eta^2)}{\{1 - 4 h (1+\eta^2)\}^{5/2}}, \quad h< \dfrac{1}{4(1+\eta^2)}.
\end{align*}
\end{proposition}
\begin{proof}
The expressions of $\mathbb{E}(Y)$, $\mathbb{E}(Y^2)$, and $\mathbb{V}\text{ar}(Y)$ are given in Proposition \ref{prop:exp_var_cov}. Since $Y \sim \mathcal{SNTH}(0,1,1,\eta,h)$, we have $Y = \tau_h(Z)$, where $Z \sim \mathcal{SN}(0,1,\eta)$. Hence,
\begin{align*}
\mathbb{E}(Y^3) &= \int_{\mathbb{R}} x^3 \exp(3 h x^2/2 ) 2 \phi(x;0,1+\eta ^2 ) \Phi \left( \dfrac{\eta x}{\sqrt{1+\eta^2}} \right) \text{d} x\\
&= \dfrac{1}{\sqrt{1-3h (1+\eta^2)}} \int_{\mathbb{R}} x^3 2 \phi \left( x;0,\dfrac{(1+\eta^2)}{1 - 3h (1+\eta^2)} \right)\\
&~~~~~~~~~~~~~~~~~~~~~~~~~~~~~~~~~~~~~~~\times \Phi \Bigg (  \dfrac{\eta}{\{1-3h (1+\eta^2)\}^{1/2}}  \dfrac{(1+\eta^2)^{-1/2}}{\{1-3h (1+\eta^2)\}^{-1/2}} x \Bigg )\text{d} x\\
&= \dfrac{1}{\sqrt{1-3h (1+\eta^2)}} \mathbb{E}(X^3)\quad \text{with } X \sim \mathcal{ASN} \left(0,\frac{(1+\eta^2)}{1 - 3h (1+\eta^2)},\frac{\eta}{\{1-3h (1+\eta^2)\}^{1/2}} \right)\\
&=\sqrt{\dfrac{2}{\pi}} \dfrac{(1+\eta^2)^{3/2}}{\{1 - 3 h (1+\eta^2)\}^2} \Bigg [ \dfrac{2 \eta^3 + 3 \eta \{1 - 3h (1+\eta^2)\}}{\{ (1+\eta^2)(1- 3 h)\}^{3/2}} \Bigg ], \quad h< \dfrac{1}{3(1+\eta^2)},
\end{align*}
and
\begin{align*}
\mathbb{E}(Y^4) &= \int_{\mathbb{R}} x^4 \exp(2 h x^2 ) 2 \phi(x;0,1+\eta ^2 ) \Phi \left( \dfrac{\eta x}{\sqrt{1+\eta^2}} \right) \text{d} x\\
&=  \dfrac{1}{\sqrt{1-4h (1+\eta^2)}} \int_{\mathbb{R}} x^4 2 \phi \left( x;0,\dfrac{(1+\eta^2)}{1 - 4h(1+\eta^2)} \right)\\
&~~~~~~~~~~~~~~~~~~~~~~~~~~~~~~~~~~~~~~~\times \Phi \Bigg \{  \dfrac{\eta}{\{1-4h(1+\eta^2)\}^{1/2}}  \dfrac{(1+\eta^2)^{-1/2}}{\{1-4h (1+\eta^2)\}^{-1/2}} x \Bigg \} \text{d} x\\
&= \dfrac{1}{\sqrt{1-4h (1+\eta^2)}} \mathbb{E}(X^4) \quad \text{with } X \sim \mathcal{ASN} \left(0,\frac{(1+\eta^2)}{1 - 4h (1+\eta^2)},\frac{\eta}{\{1-3h (1+\eta^2)\}^{1/2}} \right)\\
&=\dfrac{3(1+\eta^2)}{\{1 - 4 h (1+\eta^2)\}^{5/2}}, \quad h< \dfrac{1}{4(1+\eta^2)}.
\end{align*}
The $3^{\text{rd}}$ and $4^{\text{th}}$ order moments of the $\mathcal{ASN}$ distribution are obtained from Chapter 2 of \cite{2013.A.A.A.C.CUP}.\qed
\end{proof}

\begin{figure}[t!]
\begin{center}
\centering
\begin{subfigure}{0.4\textwidth}
  \centering
  \includegraphics[width=\linewidth]{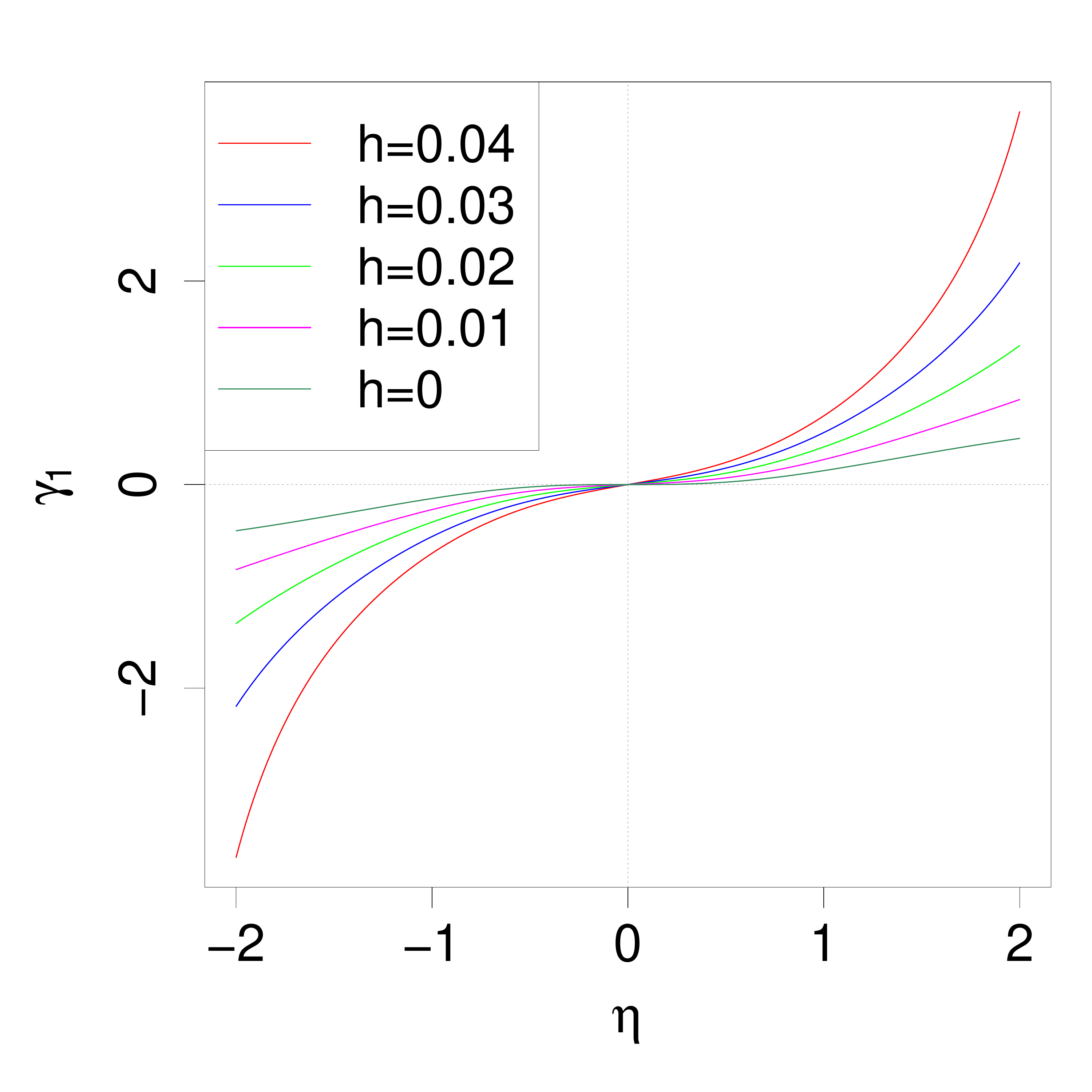}
  %\vspace{-0.35in}
  \caption{Plots of $\gamma_1$ against $\eta$ for different fixed $h$}
  
 % \label{fig:boxplot_2.1}
\end{subfigure}%
\begin{subfigure}{0.4\textwidth}
  \centering
  \includegraphics[width=\linewidth]{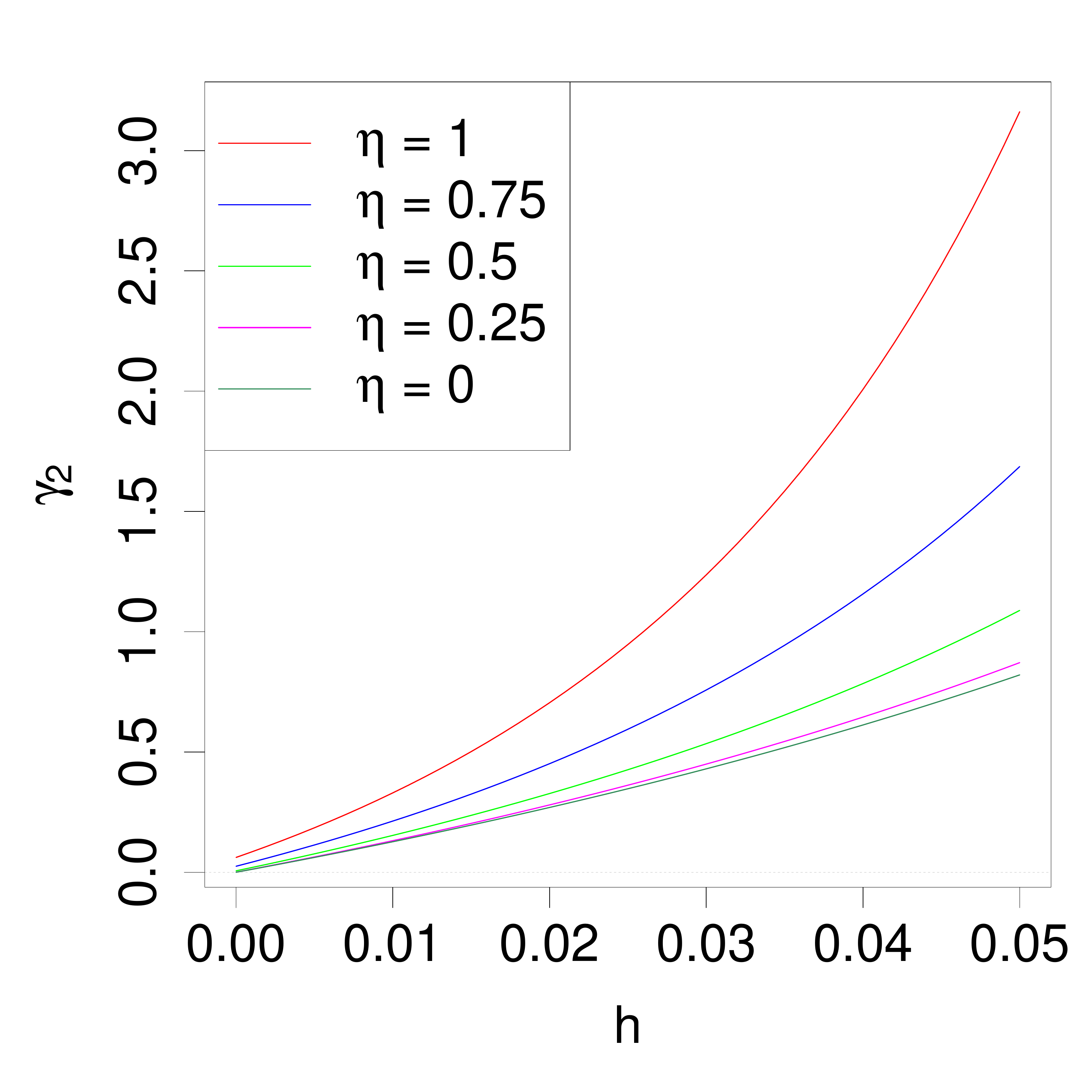}
  %\vspace{-0.35in}
  \caption{Plots of $\gamma_2$ against $h$ for different fixed $\eta$}
 % \label{fig:boxplot_2.2}
\end{subfigure}
\caption{Plots of the measures of skewness and kurtosis for the $\mathcal{SNTH}_1(0,1,1,\eta,h)$ distribution.}
\label{fig:skewness_kurtosis}
\end{center}
\end{figure}
We provide plots of the $\gamma_1$ and $\gamma_2$ measures for the $\mathcal{SNTH}_1(0,1,1,\eta,h)$ distribution against $\eta$ and $h$ for different fixed $h$ and $\eta$, respectively, in Figure \ref{fig:skewness_kurtosis}. From the plots, it is clear that the parameter $\eta$ dictates the extent of skewness in the distribution. Moreover, for a fixed $\eta$, the extent of skewness increases with increase in $h$ and vice-versa. Similarly, the extent of the tail-thickness is dictated by the parameter $h$ and for a fixed $h$, the tail-thickness increases with increase in $\eta$ and vice-versa. Here we only plot $\gamma_2$ against $h$ for positive $\eta$ as $\gamma_2$ is only a function of $\eta^2$. The plots show how the effect of $\eta$ and $h$ on skewness and kurtosis are intertwined. Nevertheless, we associate the parameter $\eta$ with the skewness and the parameter $h$ with the tail-thickness of the $\mathcal{SNTH}$ distribution. It is also worth pointing out from the plots that the $\gamma_2$ measure cannot be less than zero for the $\mathcal{SNTH}$ distribution. Hence, the $\mathcal{SNTH}$ distribution is not suitable for scenarios when tail-thickness of the data is less than that of the Gaussian distribution.

\begin{figure}[b!]
\begin{center}
\centering
\begin{subfigure}{0.4\textwidth}
  \centering
  \includegraphics[width=\linewidth]{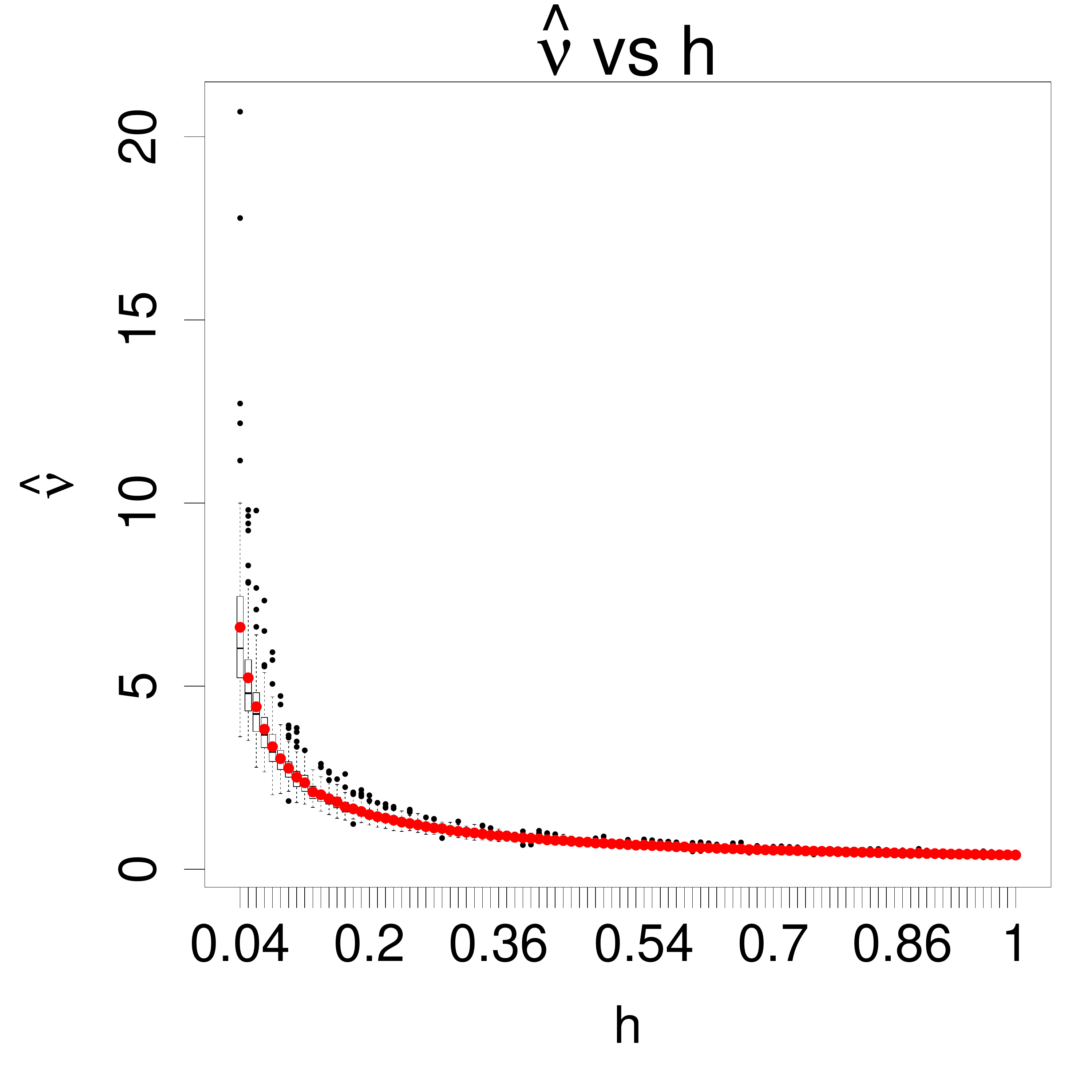}
  %\vspace{-0.35in}
  \caption{Boxplots of $\hat{\nu}$ against the true $h$ parameter for data simulated from $\mathcal{SNTH}$}
  
 \label{fig:true_h_vs_nu_est}
\end{subfigure}%
\hspace{0.25in}
\begin{subfigure}{0.4\textwidth}
  \centering
  \includegraphics[width=\linewidth]{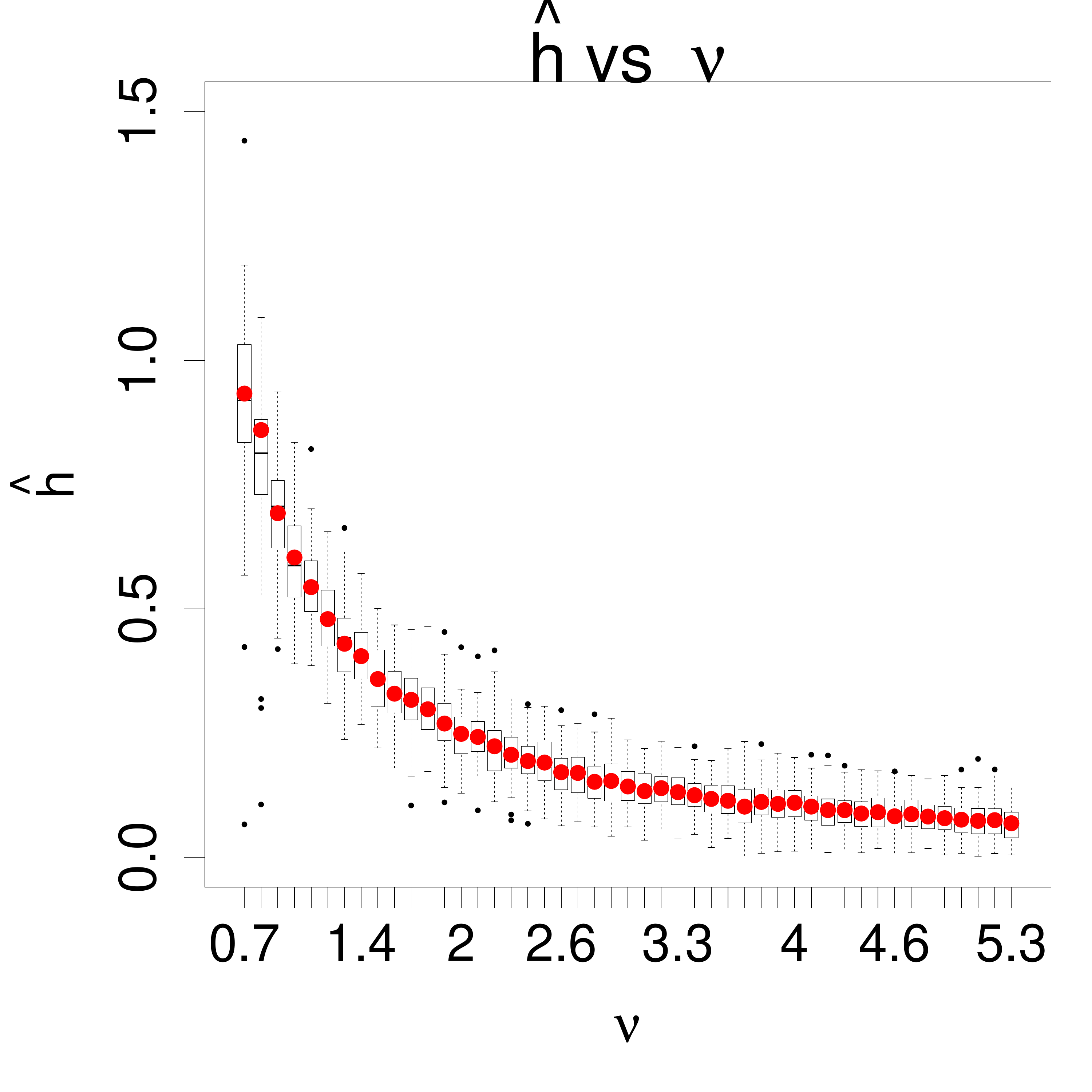}
  %\vspace{-0.35in}
  \caption{Boxplots of $\hat{h}$ against the true $\nu$ parameter for data simulated from skew-$t$}
 \label{fig:true_nu_vs_h_est}
\end{subfigure}
\caption{Boxplots of estimated $\nu$ parameter against the true $h$ parameter in (a) and estimated $h$ against the true $\nu$ parameter in (b). The red dots in each plot correspond to the means of the estimates based on $100$ replicates.}
\label{fig:relation_h_nu_1d}
\end{center}
\end{figure}

The $h$ parameter of the $\mathcal{SNTH}$ distribution is the counterpart of the $\nu$ parameter of the skew-$t$ distribution since these two parameters primarily control the tail-thickness in their respective distributions. The relationship between $h$ and $\nu$ is studied here using two simulation experiments. In the first experiment we simulate $500$ realizations from $\mathcal{SNTH}_1(0,1,1,1.5,h)$, where $h$ varies in the interval $[0.02,1]$. We fit the skew-$t$ distribution to the simulated $\mathcal{SNTH}$ data for varying $h$ with the R \citep{R-Core-Team:2022aa} package \texttt{sn} \citep{2015.AA.R} and note the estimate of $\nu$. For each $h$, we repeat this experiment $100$ times and present the boxplots of $\nu$ estimates as a function of $h$ in Figure \ref{fig:true_h_vs_nu_est}. Moreover, the estimates' means are indicated by the red dots. Similar experiment results are provided in Figure \ref{fig:true_nu_vs_h_est}, where we present the boxplots of the $100$ estimates of $h$ obtained by fitting the  $\mathcal{SNTH}$ distribution to $100$ replicates of size $500$ from the skew-$t$ distribution with location, scale, and skewness parameter as $0$, $1$, and $1.5$, with varying degrees of freedom $\nu \in [0.7,5.3]$. From the two boxplots in Figure \ref{fig:relation_h_nu_1d}, we can see  how the two tail-thickness parameters of the $\mathcal{SNTH}$ and the skew-$t$ are related. As $\nu$ in the skew-$t$ distribution increases, the kurtosis decreases and that corresponds to the decrease in $h$ in the $\mathcal{SNTH}$ distribution and vice-versa.

A similar experiment, done in the bivariate case, yields some interesting results. In this experiment, we simulate $500$ realizations from $\mathcal{SNTH}_2 \left( \begin{psmallmatrix}0\\0\end{psmallmatrix},\textbf{I}_2,\begin{psmallmatrix}1 & 0.3\\0.3 & 1\end{psmallmatrix},\begin{psmallmatrix} -1.5\\2\end{psmallmatrix},\begin{psmallmatrix} h_1 \\ h_2\end{psmallmatrix} \right)$
with varying $h_2 \in [0.01,1]$, for fixed $h_1 \in \{0.2,0.4,0.6,0.8,1\}$. We fit a bivariate skew-$t$ distribution to the $\mathcal{SNTH}$ observation and note the estimate of $\nu$. Based on $100$ replicates, we plot the median of $\hat \nu$s against $h_2$ for different $h_1$ in Figure \ref{fig:true_h2_vs_est_nu}. Moreover, we smooth the curve using local polynomial fitting. From this plot we see that a particular $\hat{\nu}$ can be obtained for different pairs of $h_1$ and $h_2$. To emphasize that, we have highlighted the line $\hat{\nu} = 1$ which cuts all the curves in the plot. From here we conclude that the skew-$t$ distribution is not suitable for scenarios when there is a great disparity between marginal kurtosis. When $h_2$ is very small and $h_1$ is large, the skew-$t$ model puts more emphasis on $h_1$ and the overall estimate of $\nu$ in that case becomes small, which corresponds to heavier tail in the fitted distribution. As $h_2$ increases, the true distribution becomes more heavy-tailed but the fitted distribution becomes less heavy-tailed.
\begin{figure}[t!]
\begin{center}
\centering
  \includegraphics[width=0.8\linewidth]{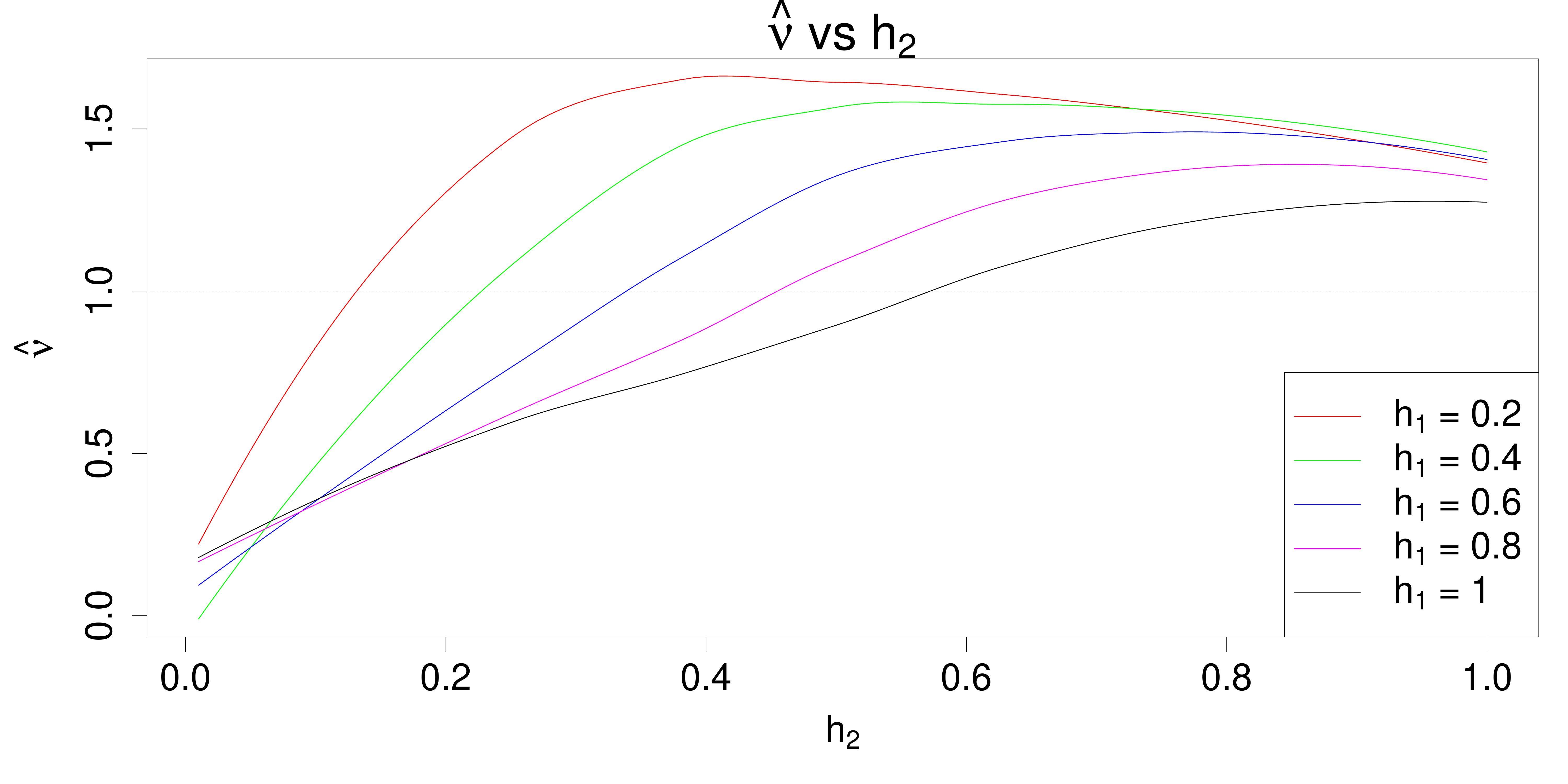}
\caption{Curves obtained by smoothing the median of $\hat{\nu}$s from fitted bivariate skew-$t$ based on $100$ replicates from $\mathcal{SNTH}_2$ as a function of true $h_2$ for different values of $h_1$.}
\label{fig:true_h2_vs_est_nu}
\end{center}
\end{figure}

\subsection{Conditional Distribution of $\mathcal{SNTH}$}

Before deriving the conditional distribution of the $\mathcal{SNTH}$ family, we first discuss the result about the conditional distribution of the $\mathcal{SN}$ distribution. To do that, we need to revisit the family of the extended skew-normal distribution \citep{2001.C.J.A.K.S.Financialmodelling, 2000.B.A.R.B.S, 2003.A.C.A.A.E.S.SJS, 2010.R.B.A.V.M.G.G.M} but with the $\bm \Psi$-$\bm \eta$ parameterization, similar to the definition of the $\mathcal{SN}$ distribution in Section~\ref{SNdistrib}. A $p$-variate random vector $\bm Y$ has an extended skew-normal distribution if its pdf is
\begin{equation}\label{eq:esn_mpdf}
    f_{\bm Y}(\bm y) = \dfrac{1}{\Phi(\tau)} \phi_p(\bm y; \bm \xi + \tau \bm \eta, \bm \Psi + \bm \eta \bm \eta^\top) \Phi \left\{ \dfrac{\tau + \bm \eta^\top \bm \Psi^{-1} (\bm y - \bm \xi)}{\sqrt{1+\bm \eta^\top \bm \Psi^{-1} \bm \eta}} \right\},\quad \bm y \in \mathbb{R}^p,
\end{equation}
where $\bm \xi \in \mathbb{R}^p$ is the location parameter, $\bm \Psi \in \mathbb{R}^{p \times p}$ is the symmetric positive definite scale matrix, $\bm \eta \in \mathbb{R}^p$ is the skewness parameter, and $\tau \in \mathbb{R}$ is the extension parameter. We denote $\bm Y \sim \mathcal{ESN}_p (\bm \xi,\bm \Psi,\bm \eta,\tau)$. From the pdf of the $\mathcal{ESN}$ distribution in Equation \eqref{eq:esn_mpdf} we have, when this extension parameter $\tau = 0$, that the $\mathcal{ESN}$ distribution reduces to the $\mathcal{SN}$ distribution. Like the $\mathcal{SN}$ distribution, a random vector $\bm Y \sim \mathcal{ESN}_p (\bm \xi,\bm \Psi,\bm \eta,\tau)$ also has a concise stochastic representation  
\begin{equation}\label{eq:esn_stochastic}
    \bm Y = \bm \xi + \tau \bm \eta + \bm \eta U + \bm W,
\end{equation}
where $U \,{\buildrel d \over =}\, (Z| Z+ \tau >0)$, $Z \sim \mathcal{N}(0,1)$, $\bm W \sim \mathcal{N}_p(\bm 0 ,\bm \Psi)$, and $Z$ and $\bm W$ are independently distributed. The last statement is directly obtained from Proposition 1 of \cite{2010.R.B.A.V.M.G.G.M} (see their Equation (10) with $\nu \rightarrow \infty$). As a consequence of this stochastic representation, the marginals of the $\mathcal{ESN}$ distribution also remain in the same family and the parameters of the marginal distribution are just the corresponding marginal parameters, similar to the $\mathcal{SN}$ distribution. We need this definition of the $\mathcal{ESN}$ distribution because the conditionals of the $\mathcal{SN}$ family belongs to the $\mathcal{ESN}$ family.

Let $\bm Y \sim \mathcal{SN}_p(\bm \xi,\bm \Psi,\bm \eta)$, and consider the partition of $\bm Y = (\bm Y_1^\top, \bm Y_2^\top)^\top$ with $\bm Y_i$ of size $p_i$ ($i=1,2$) and such that $p_1+p_2 = p$, with corresponding partitions of the parameters in blocks of matching sizes. Then the conditional distribution of $\bm Y_1$ given $\bm Y_2 = \bm y_2$, $\bm y_2 \in \mathbb{R}^{p_2}$, is 
\begin{equation}\label{eq:sn_conditional}
(\bm Y_1|\bm Y_2 = \bm y_2) \sim \mathcal{ESN}_{p_1} (\bm \xi_{1.2},\bar{\bm \Psi}_{11.2}, \bar{\bm \eta}_{1.2},\bar{\tau}_{1.2}),
\end{equation}
where $\bm \xi_{1.2} = \bm \xi_1 + \bm \Psi_{12} \bm \Psi_{22}^{-1} (\bm y_2-  \bm \xi_2)$, $\bm \Psi_{11.2} = \bm \Psi_{11} - \bm \Psi_{12} \bm \Psi_{22}^{-1} \bm \Psi_{21}$, $\bm \eta_{1.2} = \bm \eta_1 - \bm \Psi_{12} \bm \Psi_{22}^{-1} \bm \eta_2$,
$$\bar{\bm \eta}_{1.2} = \dfrac{{\bm \eta}_{1.2}}{\sqrt{1+ {\bm \eta}_{2}^\top \bm \Psi_{22}^{-1} {\bm \eta}_{2}}},\text{~and~} \bar{\tau}_{1.2} = \dfrac{\bm \eta_2^\top \bm \Psi_{22}^{-1} (\bm y_2 - \bm \xi_2)}{\sqrt{1+ {\bm \eta}_{2}^\top \bm \Psi_{22}^{-1} {\bm \eta}_{2}}}.$$
This result can be verified by the fact that the conditional distribution of the $\mathcal{ASN}$ family belongs to the extended skew-normal distribution proposed by  \cite{2010.R.B.A.V.M.G.G.M} (see Section 5.3.2 in \cite{2013.A.A.A.C.CUP}) and by reparameterizing to the $\bm \Psi$-$\bm \eta$ parameterization.

In the next proposition we derive the conditional distribution of the $\mathcal{SNTH}$ family. We show that the conditional distributions of the $\mathcal{SN}$ family and the $\mathcal{SNTH}$ family are related.
\begin{proposition}\label{prop:SNTH_conditional}
Let $\bm Y \sim \mathcal{SNTH}_p(\bm \xi,\bm \omega,\bar{\bm \Psi},\bm \eta,\bm h)$, and consider the partition of $\bm Y = (\bm Y_1^\top, \bm Y_2^\top)^\top$ with $\bm Y_i$ of size $p_i$ ($i=1,2$) and such that $p_1+p_2 = p$, with corresponding partitions of the parameters in blocks of matching sizes. Then the conditional distribution of $\bm Y_1$ given $\bm Y_2 = \bm y_2$ is 
$$(\bm Y_1|\bm Y_2 = \bm y_2) \,{\buildrel d \over =}\, \bm \tau_{\bm h_1}(\bm Y_0),\quad \bm Y_0  \sim \mathcal{ESN}_{p_1} (\bm \xi_{1.2},\bar{\bm \Psi}_{11.2}, \bar{\bm \eta}_{1.2},\bar{\tau}_{1.2}),$$
where $\bm \xi_{1.2} = \bar{\bm \Psi}_{12} \bar{\bm \Psi}_{22}^{-1} \bm g_2 (\bm y_2)$, $\bar{\bm \Psi}_{11.2} = \bar{\bm \Psi}_{11} - \bar{\bm \Psi}_{12} \bar{\bm \Psi}_{22}^{-1} \bar{\bm \Psi}_{21}$,
$\bm \tau_{\bm h_1} (\cdot)$ is the same as in Equation \eqref{eq:multivariate_Tukey-h_transformation}, $\bm g(\bm y)$ is the same as in Equation \eqref{eq:snth-density},
$\bm g(\bm y) = \{\bm g_1(\bm y_1),\bm g_2(\bm y_2)\}^\top$ with
$\bm g_1(\bm y_1) = \{g_{1}(y_{1}),\ldots,g_{p_1}(y_{p_1})\}^\top$ and $\bm g_2(\bm y_2) = \{g_{p_1+1}(y_{p_1+1}),\ldots,g_{p}(y_p)\}^\top$, and 
$$\bar{\bm \eta}_{1.2} = \dfrac{\bm \eta_1 - \bar{\bm \Psi}_{12} \bar{\bm \Psi}_{22}^{-1} \bm \eta_2}{\sqrt{1+\bm \eta_2^\top \bar{\bm \Psi}_{22}^{-1} \bm \eta_2}}, \quad \bar{\tau}_{1.2} = \dfrac{\bm \eta_2 ^\top \bar{\bm \Psi}_{22}^{-1} \bm g_2 (\bm y_2)}{\sqrt{1+\bm \eta_2^\top \bar{\bm \Psi}_{22}^{-1} \bm \eta_2}}.$$
\end{proposition}
\begin{proof}
 From Proposition \ref{prop:marginal_SNTH}, the marginal pdf of $\bm Y_2$ is
\begin{equation*}\label{eq:snth-density_Y_2}
\begin{split}
f_{\bm Y_2} (\bm y_2) &= 2 \phi_{p_2} \{ \bm g_2(\bm y_2) ; \bm 0, (\bar{ \bm \Psi}_{22} + \bm \eta_2 \bm \eta_2^\top) \} \Phi \Bigg \{ \dfrac{\bm \eta_2^\top \bar{ \bm \Psi}_{22}^{-1} \bm g_2(\bm y_2)}{\sqrt{1+\bm \eta_2^\top \bar{ \bm \Psi}_{22}^{-1} \bm \eta_2}} \Bigg \}\\ &~~~~\times \prod_{i = p_1+1}^{p} \Bigg \{ \dfrac{1}{\omega_{ii}} \left( \dfrac{\exp[\frac{1}{2} W_0 \{ h_i (\frac{y_i - \xi_i}{\omega_{ii}})^2 \}]}{h_i (\frac{y_i - \xi_i}{\omega_{ii}})^2 + \exp[ W_0 \{ h_i (\frac{y_i - \xi_i}{\omega_{ii}})^2 \}]} \right) \Bigg \}, \quad \bm y_2 \in \mathbb{R}^{p_2}.
\end{split}
\end{equation*}
Hence, the conditional pdf of $\bm Y_1|\bm Y_2 = \bm y_2$ is 
\begin{equation*}\label{eq:conditional_snth}
\begin{split}
    f_{\bm Y_1| \bm Y_2 = \bm y_2}(\bm y_1) = \dfrac{f_{\bm Y}(\bm y)}{f_{\bm Y_2}(\bm y_2)} &= \dfrac{\phi_{p} \{ \bm g (\bm y) ; \bm 0, ( \bar{\bm \Psi} + \bm \eta \bm \eta^\top) \} \Phi \Bigg \{ \dfrac{\bm \eta^\top \bar{\bm \Psi}^{-1} \bm g(\bm y)}{\sqrt{1+\bm \eta^\top \bar{\bm \Psi}^{-1} \bm \eta}} \Bigg \}}{\phi_{p_2} \{ \bm g_2 (\bm y_2) ; \bm 0, ( \bar{\bm \Psi}_{22} + \bm \eta_2 \bm \eta_2^\top) \} \Phi \Bigg \{ \dfrac{\bm \eta_2^\top \bar{\bm \Psi}_{22}^{-1} \bm g_2(\bm y_2)}{\sqrt{1+\bm \eta_2^\top \bar{\bm \Psi}_{22}^{-1} \bm \eta_2}} \Bigg \}} \\
    &~~~~ \times \prod_{i = 1}^{p_1} \Bigg \{ \dfrac{1}{\omega_{ii}} \left( \dfrac{\exp[\frac{1}{2} W_0 \{ h_i (\frac{y_i - \xi_i}{\omega_{ii}})^2 \}]}{h_i (\frac{y_i - \xi_i}{\omega_{ii}})^2 + \exp[ W_0 \{ h_i (\frac{y_i - \xi_i}{\omega_{ii}})^2 \}]} \right) \Bigg \}, \quad \bm y_1 \in \mathbb{R}^p_1.
\end{split}
\end{equation*}
From the pdf given above, we can see that it is the density function of $\bm \tau _{\bm h_1} (\bm Y_0)$, where $\bm Y_0 \overset{\text{d}}{=} [\bm Z_1 | \{\bm Z_2 = \bm g_2 (\bm y_2)\}]$ and $\bm Z = (\bm Z_1^\top,\bm Z_2^\top)^\top \sim \mathcal{SN}_p (\bm 0, \bar{\bm \Psi}, \bm \eta)$. Hence, from Equation \eqref{eq:sn_conditional}, we have $\bm Y_{0} \sim \mathcal{ESN}_{p_1} (\bm \xi_{1.2},\bar{\bm \Psi}_{11.2}, \bar{\bm \eta}_{1.2},\bar{\tau}_{1.2}).$\qed
\end{proof}

Since the conditional distribution of the $\mathcal{SNTH}$ family can be viewed as a component-wise Tukey-$h$ transformation on the $\mathcal{ESN}$, closed-form expressions of its mean vector and variance-covariance matrix can be derived. The conditional mean and the variance-covariance matrix will be helpful for using the $\mathcal{SNTH}$ model for various formal statistical purposes such as regression modeling, time-series analysis, and spatial modeling. In the next three propositions we provide the mathematical expressions of the elements of the conditional mean vector and the conditional variance-covariance matrix. The proofs of Proposition \ref{prop:var_condition_snth} and \ref{prop:cov_condition_snth} below are very similar to the proof of Proposition \ref{prop:expec_condition_snth}, hence they are omitted in the main article and are given in Sections S1 and S2 in the supplementary material.
\begin{proposition}\label{prop:expec_condition_snth}
Let $\bm Y_0$ be defined as in Proposition \ref{prop:SNTH_conditional}. The mean vector $\bm \mu = \mathbb{E}\{\bm \tau_{\bm h_1} (\bm Y_0)\}$ is:
\begin{align*}
    \mu_i &= \dfrac{1}{\sqrt{1-(\bar{ \Psi}_{11.2_{ii}}+{\bar{ \eta}_{1.2_i}}^2)h_i}}
    \exp \left \{ \dfrac{(\xi_{1.2_i}+\bar{\tau}_{1.2} {\bar{ \eta}_{1.2_i}})^2 h_i}{2(1-(\bar{ \Psi}_{11.2_{ii}}+{\bar{ \eta}_{1.2_i}}^2)h_i)} \right \}   \dfrac{\Phi(\Tilde{\tau}_i)}{\Phi(\bar{\tau}_{1.2})} \left \{ \tilde{\xi}_i + \tilde{\omega}_i \tilde{\delta}_i \dfrac{\phi(\tilde{\tau}_i)}{\Phi(\tilde{\tau}_i)} \right \}, 
\end{align*}
where $\bm \xi_{1.2} = (\xi_{1.2_1},\ldots,\xi_{1.2_{p_1}})^\top$, $\text{diag}(\bar{\bm \Psi}_{11.2}) = (\bar{ \Psi}_{11.2_{11}},\ldots,\bar{ \Psi}_{11.2_{p_1 p_1}})^\top$, $\bar{\bm \eta}_{1.2} = (\bar{ \eta}_{1.2_1},\ldots,\bar{ \eta}_{1.2_{p_1}})^\top$,
 $\Tilde{\xi}_i = \frac{\xi_{1.2_i} + \bar{\tau}_{1.2} {\bar{ \eta}_{1.2_i}}}{1-(\bar{ \Psi}_{11.2_{ii}} + {\bar{ \eta}_{1.2_i}}^2)h_i}$, 
$\Tilde{\omega}_i = \sqrt{\frac{\bar{ \Psi}_{11.2_{ii}} +{\bar{ \eta}_{1.2_i}}^2}{1-(\bar{ \Psi}_{11.2_{ii}} + {\bar{ \eta}_{1.2_i}}^2)h_i}}$, 
$\Tilde{\alpha}_i = \frac{{\bar{ \eta}_{1.2_i}}}{\sqrt{\bar{ \Psi}_{11.2_{ii}}}} \frac{1}{\sqrt{1-(\bar{ \Psi}_{11.2_{ii}} + {\bar{ \eta}_{1.2_i}}^2)h_i}}$,\\
$\Tilde{\alpha}_{0_i} = \frac{\bar{\tau}_{1.2} \sqrt{\bar{ \Psi}_{11.2_{ii}}} + \frac{{\bar{ \eta}_{1.2_i}}}{\sqrt{\bar{ \Psi}_{11.2_{ii}}}} \left\{\frac{\bar{\tau}_{1.2} {\bar{ \eta}_{1.2_i}} +  \xi_{1.2_i} (\bar{ \Psi}_{11.2_{ii}} +{\bar{ \eta}_{1.2_i}}^2)h_i}{1-(\bar{ \Psi}_{11.2_{ii}}+{\bar{ \eta}_{1.2_i}}^2)h_i}\right \}}{\sqrt{\bar{ \Psi}_{11.2_{ii}} + {\bar{ \eta}_{1.2_i}}^2}}$, $\Tilde{\delta}_i = \frac{\Tilde{\alpha}_i}{\sqrt{1+\Tilde{\alpha}_i}^2}$,  $\Tilde{\tau}_i = \frac{\Tilde{\alpha}_{0_i}}{\sqrt{1+\Tilde{\alpha}_i^2}}$, and $h_i < \frac{1}{\bar{ \Psi}_{11.2_{ii}} + {\bar{ \eta}_{1.2_i}}^2}$, $i = 1,\ldots,p_1$.
\end{proposition}
\begin{proof}
From Equation \eqref{eq:esn_stochastic} it can be established that
 $Y_{0_i} \sim \mathcal{ESN}_{1} (\xi_{1.2_i},\bar{ \Psi}_{11.2_{ii}}, {\bar{ \eta}_{1.2_i}},\bar{\tau}_{1.2}),\quad i = 1,\ldots,p_1$. 
Then, 
\begin{align*}
    \mu_i &=\mathbb{E}(Y_{0_i}) = \int_{\mathbb{R}} x \exp(h_i x^2/2) \dfrac{1}{\Phi(\bar{\tau}_{1.2})} \phi(x;\xi_{1.2_i}+\bar{\tau}_{1.2} {\bar{ \eta}_{1.2_i}}, \bar{ \Psi}_{11.2_{ii}}+{\bar{ \eta}_{1.2_i}}^2)\\
    &~~~~~~~~~~~~~~~~~~~~~~~~~~\times \Phi \left\{ \dfrac{\bar{\tau}_{1.2} + {\bar{ \eta}_{1.2_i}} (x - \xi_{1.2_i})/\bar{ \Psi}_{11.2_{ii}}}{\sqrt{1+{\bar{ \eta}_{1.2_i}}^2/\bar{ \Psi}_{11.2_{ii}}}} \right\} \text{d} x\\
    &= \exp \left [ \dfrac{(\xi_{1.2_i}+\bar{\tau}_{1.2} \bar{ \eta}_{1.2_i})^2 h_i}{2 \{1-(\bar{ \Psi}_{11.2_{ii}}+\bar{ \eta}_{1.2_i}^2)h_i\}} \right ] \dfrac{1}{\Phi(\bar{\tau}_{1.2})} \dfrac{1}{\sqrt{2\pi} \sqrt{\bar{ \Psi}_{11.2_{ii}}+\bar{ \eta}_{1.2_i}^2}} \\&~~~~~\times \int_{\mathbb{R}}  x \exp \left [ -\dfrac{1}{2} \dfrac{\left \{x - \frac{\xi_{1.2_i} +\bar{\tau}_{1.2} \bar{ \eta}_{1.2_i}}{1-(\bar{ \Psi}_{11.2_{ii}}+\bar{ \eta}_{1.2_i}^2)h_i} \right \}^2}{\left \{ \frac{\bar{ \Psi}_{11.2_{ii}}+\bar{ \eta}_{1.2_i}^2}{1-(\bar{ \Psi}_{11.2_{ii}}+\bar{ \eta}_{1.2_i}^2)h_i} \right\} }  \right ]
    \Phi \left\{ \dfrac{\bar{\tau}_{1.2} + \bar{ \eta}_{1.2_i} (x - \xi_{1.2_i})/\bar{ \Psi}_{11.2_{ii}}}{\sqrt{1+\bar{ \eta}_{1.2_i}^2/\bar{ \Psi}_{11.2_{ii}}}} \right\} \text{d} x\\
    &= \exp \left \{ \dfrac{(\xi_{1.2_i}+\bar{\tau}_{1.2} {\bar{ \eta}_{1.2_i}})^2 h_i}{2(1-(\bar{ \Psi}_{11.2_{ii}}+{\bar{ \eta}_{1.2_i}}^2)h_i)} \right \}  \dfrac{1}{\sqrt{1-(\bar{ \Psi}_{11.2_{ii}}+{\bar{ \eta}_{1.2_i}}^2)h_i}} \dfrac{\Phi(\Tilde{\tau}_i)}{\Phi(\bar{\tau}_{1.2})}\\
    &~~~~~\times \int_{\mathbb{R}} x \dfrac{1}{\Phi(\Tilde{\tau}_i)} \phi(x;\Tilde{\xi}_i,\Tilde{\omega}_i^2) \Phi \{\Tilde{\alpha}_{0_i} + \Tilde{\alpha}_i \Tilde{\omega}_i^{-1} (x -\Tilde{\xi}_i) \} \text{d}x\\
    &= \dfrac{1}{\sqrt{1-(\bar{ \Psi}_{11.2_{ii}}+{\bar{ \eta}_{1.2_i}}^2)h_i}}
    \exp \left \{ \dfrac{(\xi_{1.2_i}+\bar{\tau}_{1.2} {\bar{ \eta}_{1.2_i}})^2 h_i}{2(1-(\bar{ \Psi}_{11.2_{ii}}+{\bar{ \eta}_{1.2_i}}^2)h_i)} \right \}   \dfrac{\Phi(\Tilde{\tau}_i)}{\Phi(\bar{\tau}_{1.2})} \left \{ \tilde{\xi}_i + \tilde{\omega}_i \tilde{\delta}_i \dfrac{\phi(\tilde{\tau}_i)}{\Phi(\tilde{\tau}_i)} \right \}.
\end{align*}
The last step is obtained from the moments of the extended skew-normal distribution from \cite{2013.A.A.A.C.CUP} (see Section 5.3.4). 
\qed
\end{proof}

\begin{proposition}\label{prop:var_condition_snth}
Let $\bm Y_0$ be defined as in Proposition \ref{prop:SNTH_conditional}, and let $\bm \Sigma = (\sigma_{ij}) = \mathbb{V}\text{ar} \{\bm \tau_{\bm h_1}(\bm Y_0)\}$. Then:
\begin{align*}
    \sigma_{ii} &= \dfrac{1}{\sqrt{1-2(\bar{ \Psi}_{11.2_{ii}} +{\bar{ \eta}_{1.2_i}}^2)h_i}} \exp \left \{ \dfrac{(\xi_{1.2_i}+\tau {\bar{ \eta}_{1.2_i}})^2 h_i}{1-2(\bar{ \Psi}_{11.2_{ii}}+{\bar{ \eta}_{1.2_i}}^2)h_i} \right \} \dfrac{\Phi(\Tilde{\tau}_i)}{\Phi(\bar{\tau}_{1.2})}\\
    &~~~~~\times \left \{ \Tilde{\xi}^2_i + \Tilde{\omega}^2_i - \Tilde{\tau}_i \frac{\phi(\Tilde{\tau}_i)}{\Phi(\Tilde{\tau}_i)} \Tilde{\omega}_i^2 \Tilde{\delta}_i^2 + 2 \frac{\phi(\Tilde{\tau}_i)}{\Phi(\Tilde{\tau}_i)} \Tilde{\xi}_i \Tilde{\omega} _i\Tilde{\delta}_i \right \} - \mu_i^2, 
\end{align*}
where $\bm \xi_{1.2} = (\xi_{1.2_1},\ldots,\xi_{1.2_{p_1}})^\top$, $\text{diag}(\bar{\bm \Psi}_{11.2}) = (\bar{ \Psi}_{11.2_{11}},\ldots,\bar{ \Psi}_{11.2_{p_1 p_1}})^\top$, $\bar{\bm \eta}_{1.2} = (\bar{ \eta}_{1.2_1},\ldots,\bar{ \eta}_{1.2_{p_1}})^\top$, $\Tilde{\xi}_i = \frac{\xi_{1.2_i} + \bar{\tau}_{1.2} {\bar{ \eta}_{1.2_i}}}{1-2(\bar{ \Psi}_{11.2_{ii}} + {\bar{ \eta}_{1.2_i}}^2)h_i}$, $\Tilde{\omega}_i = \sqrt{\frac{\bar{ \Psi}_{11.2_{ii}} +{\bar{ \eta}_{1.2_i}}^2}{1-2(\bar{ \Psi}_{11.2_{ii}} + {\bar{ \eta}_{1.2_i}}^2)h_i}}$, $\Tilde{\alpha}_i = \frac{{\bar{ \eta}_{1.2_i}}}{\sqrt{\bar{ \Psi}_{11.2_{ii}}}} \frac{1}{\sqrt{1-2(\bar{ \Psi}_{11.2_{ii}} + {\bar{ \eta}_{1.2_i}}^2)h_i}}$,\\ $\Tilde{\alpha}_{0_i} = \frac{\bar{\tau}_{1.2} \sqrt{\bar{ \Psi}_{11.2_{ii}}} + \frac{{\bar{ \eta}_{1.2_i}}}{\sqrt{\bar{ \Psi}_{11.2_{ii}}}} \left\{\frac{\bar{\tau}_{1.2} {\bar{ \eta}_{1.2_i}} + 2 \xi_{1.2_i} (\bar{ \Psi}_{11.2_{ii}} +{\bar{ \eta}_{1.2_i}}^2)h_i}{1-2(\bar{ \Psi}_{11.2_{ii}}+{\bar{ \eta}_{1.2_i}}^2)h_i}\right \}}{\sqrt{\bar{ \Psi}_{11.2_{ii}} + {\bar{ \eta}_{1.2_i}}^2}}$, $\Tilde{\delta}_i = \frac{\Tilde{\alpha}_i}{\sqrt{1+\Tilde{\alpha}_i}^2}$,  $\Tilde{\tau}_i = \frac{\Tilde{\alpha}_{0_i}}{\sqrt{1+\Tilde{\alpha}_i^2}}$, $\mu_i$ is the same as in Proposition \ref{prop:expec_condition_snth}, and $h_i < \frac{1}{2(\bar{ \Psi}_{11.2_{ii}} + {\bar{ \eta}_{1.2_i}}^2)}$, $i = 1,\ldots,p_1$.
\end{proposition}

\begin{proposition}\label{prop:cov_condition_snth}
Let $\bm Y_0$ be defined as in Proposition \ref{prop:SNTH_conditional}, and let $\bm \Sigma = (\sigma_{ij}) = \mathbb{V}\text{ar} \{\bm \tau_{\bm h_1}(\bm Y_0)\}$. Then:
\begin{align*}
    \sigma_{ij} &= \dfrac{\sqrt{\det \{ (\bm \Omega_{i,j}^{-1}- \bm H_{i,j})^{-1} \}}}{\sqrt{\det (\bm \Omega_{i,j})}} \exp \left[ -\frac{1}{2}  \{ \tilde{{\bm \mu}}_{i,j} ^\top \bm \Omega_{i,j}^{-1} \tilde{{\bm \mu}}_{i,j} - \tilde{{\bm \mu}}_{i,j}^\top (\bm \Omega_{i,j} - \bm \Omega_{i,j} \bm H_{i,j} \bm \Omega_{i,j})^{-1} \tilde{{\bm \mu}}_{i,j} \} \right ]\\
    &~~~~~\times \dfrac{\Phi(\Tilde{\tau}_{i,j})}{\Phi(\bar{\tau}_{1.2})}  \Bigg \{ (\Tilde{\bm \Omega}_{{i,j}})_{12} - \Tilde{\tau}_{i,j} \frac{\phi(\Tilde{\tau}_{i,j})}{\Phi(\Tilde{\tau}_{i,j})} (\Tilde{\bm \omega}_{{i,j}})_{11} (\Tilde{\bm \omega}_{{i,j}})_{22} (\Tilde{\bm \delta}_{{i,j}})_{1} (\Tilde{\bm \delta}_{{i,j}})_{2} + \xi_{1.2_i} \xi_{1.2_j} \\&~~~~~~~~~~~~~~~~~~~~~~~~~~+ \frac{\phi(\Tilde{\tau}_{i,j})}{\Phi(\Tilde{\tau}_{i,j})}\xi_{1.2_i} (\Tilde{\bm \omega}_{{i,j}})_{22} (\Tilde{\bm \delta}_{{i,j}})_{2} +  \frac{\phi(\Tilde{\tau}_{i,j})}{\Phi(\Tilde{\tau}_{i,j})} \xi_{1.2_j} (\Tilde{\bm \omega}_{{i,j}})_{11} (\Tilde{\bm \delta}_{{i,j}})_{1} \Bigg\} - \mu_i \mu_j, 
\end{align*}
where $\bm \xi_{i,j} = (\xi_{1.2_i}, \xi_{1.2_j})^\top$, $\bm \Psi_{i,j} = \begin{psmallmatrix}\bar{ \Psi}_{11.2_{ii}} & \bar{ \Psi}_{11.2_{ij}} \\ \bar{ \Psi}_{11.2_{ij}} & \bar{ \Psi}_{11.2_{jj}} \end{psmallmatrix}$, $\bm \eta_{i,j} = (\bar{ \eta}_{1.2_i},\bar{ \eta}_{1.2_{j}})^\top$, $\bm \Omega_{i,j}= \bm \Psi_{i,j} +\bm \eta_{i,j} \bm \eta_{i,j}^\top$, $\tilde{{\bm \mu}}_{i,j} = \bm \xi_{i,j} + \bar{\tau}_{1.2} \bm \eta_{i,j}$, $\bm H_{i,j}= \begin{psmallmatrix} h_i & 0\\0& h_j  \end{psmallmatrix}$,$\Tilde{\bm \xi}_{i,j} = (\textbf{I}_2 - \bm \Omega_{i,j} \bm H_{i,j} )^{-1} \tilde{{\bm \mu}}_{i,j}$, $\Tilde{\bm \Omega}_{i,j} = (\bm \Omega_{i,j}^{-1} -\bm H_{i,j})^{-1}$,\\ $\Tilde{\alpha}_{0_{i,j}} = \frac{\bar{\tau}_{1.2} + \bm \eta_{i,j}^\top \bm \Psi_{i,j} ^{-1} (\Tilde{\bm \xi}_{i,j} -\bm \xi_{i,j})}{\sqrt{1+\bm \eta_{i,j} ^\top \bm \Psi_{i,j} ^{-1} \bm \eta_{i,j}}}$, $\Tilde{\bm \alpha}_{i,j} = \frac{\Tilde{\bm \omega}_{i,j} \bm \Psi_{i,j} ^{-1} \bm \eta_{i,j}}{\sqrt{1+\bm \eta_{i,j} ^\top \bm \Psi_{i,j} ^{-1} \bm \eta_{i,j}}}$, $\Tilde{\bm \omega}_{i,j} = \{\text{diag} (\Tilde{\bm \Omega}_{i,j})\}^{1/2}$, $\bar{\Tilde{\bm \Omega}}_{i,j} =\Tilde{\bm \omega}_{i,j}^{-1} \Tilde{\bm \Omega}_{i,j} \Tilde{\bm \omega}_{i,j}^{-1}$, $\bm \delta_{i,j} = (1+\Tilde{\bm \alpha}_{i,j}^\top \bar{\Tilde{\bm \Omega}}_{i,j} \Tilde{\bm \alpha}_{i,j})^{-1/2}  \bar{\Tilde{\bm \Omega}}_{i,j} \Tilde{\bm \alpha}_{i,j} $, $\mu_i$, $\mu_j$ are the same as in Proposition \ref{prop:expec_condition_snth}, and $h_i < \frac{1}{2(\bar{ \Psi}_{11.2_{ii}} + {\bar{ \eta}_{1.2_i}}^2)}$, $h_j < \frac{1}{2(\bar{ \Psi}_{11.2_{jj}} + {\bar{ \eta}_{1.2_j}}^2)}$, $i = 1,\ldots,p_1$. $i,j = 1,\ldots,p_1$, $i \neq j$.
\end{proposition}

\subsection{Canonical Form of the $\mathcal{SNTH}$ Distribution}

Consider a $p$-variate random vector $\bm X \sim \mathcal{ASN}_p(\bm \xi,\bm \Omega,\bm \alpha)$. It can be shown that there exists a matrix $\bm H \in \mathbb{R}^{p \times p}$ such that $\bm H (\bm X - \bm \xi) \sim \mathcal{ASN}_p (\bm 0,\textbf{I}_p,\bm \alpha^*)$, where $\bm \alpha^* = (\alpha^*,0,\ldots,0)^\top$, $\alpha^* = \sqrt{\bm \alpha^\top \bar{\bm \Omega} \bm \alpha}$, and $\bm \Omega = \bm \omega \bar{\bm \Omega} \bm \omega$. \cite{2020.A.C.S} showed that the matrix $\bm H$ is of the form $\bm H = \bm Q \bm \Omega^{-1/2}$,
where $\bm Q$ is obtained from the spectral decomposition of $\bm Q^\top \bm \Lambda \bm Q = \bm \Omega^{-1/2} \bm \Sigma \bm \Omega^{-1/2}$, $\bm \Sigma = \mathbb{V}\text{ar}(\bm X) = \bm \Omega - \frac{2}{\pi} \bm \omega \bm \delta \bm \delta^\top \bm \omega$, and $\bm \delta = (1 + \bm \alpha^\top \bar{\bm \Omega} \bm \alpha)^{-1/2} \bar{\bm \Omega} \bm \alpha$. The distribution of $\bm H (\bm X - \bm \xi)$ is defined as the canonical form of the $\mathcal{ASN}$ distribution. 

Similarly, we can define the canonical form of the $\mathcal{SN}$ distribution. Consider a random vector $\bm X \sim \mathcal{SN}_p(\bm \xi ,\bm \Psi,\bm \eta)$. Using $\bm \Omega = \bm \Psi + \bm \eta \bm \eta^\top$ and $\bm \eta = \bm \omega \bm \delta$, the relations  between the parameterizations of the $\mathcal{ASN}$ and the $\mathcal{SN}$, the distribution of $\bm H (\bm X - \bm \xi)$ is obtained as $\mathcal{SN}_p (\bm 0, \textbf{I}_p - \frac{\bm \alpha^* {\bm \alpha^*}^\top}{1 + {\bm \alpha^*}^\top {\bm \alpha^*}}, \frac{\bm \alpha^* }{\sqrt{1 + {\bm \alpha^*}^\top {\bm \alpha^*}}} )$.
Hence, the canonical form of the $\mathcal{SN}$ distribution is defined by the distribution of $\bm H^* (\bm X - \bm \xi) \sim \mathcal{SN}_p (\bm 0, \textbf{I}_p, \bm \eta^*)$, where $\bm \eta^* = (\eta^*, 0,\ldots,0)^\top$, $\eta^* = \sqrt{\bm \alpha^\top \bar{\bm \Omega} \bm \alpha}$, and  $\bm H^* = \begin{psmallmatrix} \sqrt{1+\bm \alpha^\top \bar{\bm \Omega} \bm \alpha} & \bm 0^\top \\ \bm 0 & \textbf{I}_{p-1}\end{psmallmatrix} \bm H$. 

The canonical form of the $\mathcal{ASN}$ or the $\mathcal{SN}$ distribution is useful for deriving  Mardia's measures of multivariate skewness and kurtosis \citep{1970.V.K.M.Biometrika} and the measures of multivariate skewness and kurtosis introduced by \cite{1973.J.F.M.A.A.JASA} since they are invariant under affine transformations of the variable. Moreover, using the canonical form, the unique mode of the $\mathcal{ASN}$ distribution can be derived; see Proposition 5.14 in \cite{2013.A.A.A.C.CUP}. Hence, the canonical form is used mainly to reduce the dimensionality of various problems when applicable. 

For the $\mathcal{SNTH}$ distribution, we define the canonical form by taking the component-wise Tukey-$h$ transformation of the canonical form of the latent $\mathcal{SN}$ random vector.
\begin{proposition}
Suppose $\bm Y \sim \mathcal{SNTH}_p(\bm \xi, \bm \omega,\bar{\bm \Psi},\bm \eta,\bm h)$. We define the canonical form of the $\mathcal{SNTH}$ by the distribution of $$\bm \omega^{-1} (\bm Y^* - \bm \xi) = \bm \tau_{\bm h} [ \bm H^* \bm \tau_{\bm h}^{-1} \{ \bm \omega^{-1} (\bm Y - \bm \xi) \}] \sim \mathcal{SNTH}_p(\bm 0,\textbf{I}_p,\textbf{I}_p,\bm \eta^*,\bm h),$$
where $\bm \tau_{\bm h}^{-1}(\bm z) = \{\tau_{h_1}^{-1} (z_1),\ldots, \tau_{h_p}^{-1}(z_p) \}^\top$, $\tau_{h}^{-1}(z)$ is same as in Equation \eqref{eq:inv_tukey_h}, $\bm \eta^* = (\eta^*, 0,\ldots,0)^\top$, $\eta^* = \sqrt{\bm \alpha^\top \bar{\bm \Omega} \bm \alpha}$, $\bm \Omega = \bar{\bm \Psi} + \bm \eta \bm \eta^\top$, $\bm \alpha = (1 + \bm \eta^\top \bar{\bm \Psi}^{-1} \bm \eta)^{-1/2} \{\text{diag}(\bm \Omega)\}^{1/2} \bar{\bm \Psi}^{-1} \bm \eta$, 
$\bar{\bm \Omega} = \{\text{diag}(\bm \Omega)\}^{-1/2} \bm \Omega \{\text{diag}(\bm \Omega)\}^{-1/2}$, $\bm H^* = \begin{psmallmatrix} \sqrt{1+\bm \alpha^\top \bar{\bm \Omega} \bm \alpha} & \bm 0^\top \\ \bm 0 & \textbf{I}_{p-1}\end{psmallmatrix} \bm H$, $\bm H = \bm Q \bm \Omega^{-1/2}$, 
$\bm Q$ is obtained from the spectral decomposition of $\bm Q^\top \bm \Lambda \bm Q = \bm \Omega^{-1/2} \bm \Sigma \bm \Omega^{-1/2}$, and $\bm \Sigma = \bar{\bm \Psi} + \left(1-\frac{2}{\pi}\right)\bm \eta \bm \eta^\top$.
\end{proposition}
\begin{proof}
We have $\bm Y = \bm \xi + \bm \omega \bm \tau_{\bm h} (\bm Z)$, where $\bm Z \sim \mathcal{SN}_p(\bm 0,\bar{\bm \Psi}, \bm \eta)$. Moreover, let $\bm Z^*$ be the canonical transform of $\bm Z$, and $\bm Z^* = \bm H^* \bm Z \sim \mathcal{SN}_p(\bm 0,\textbf{I}_p,\bm \eta^*)$. Here, $\bm H^* = \begin{psmallmatrix} \sqrt{1+\bm \alpha^\top \bar{\bm \Omega} \bm \alpha} & \bm 0^\top \\ \bm 0 & \textbf{I}_{p-1}\end{psmallmatrix} \bm H$, $\bm H = \bm Q \bm \Omega^{-1/2}$, 
$\bm Q$ is obtained from the spectral decomposition of $\bm Q^\top \bm \Lambda \bm Q = \bm \Omega^{-1/2} \bm \Sigma \bm \Omega^{-1/2}$, and $\bm \Sigma = \mathbb{V}\text{ar}(\bm Z) = \bar{\bm \Psi} + \left(1-\frac{2}{\pi}\right)\bm \eta \bm \eta^\top$. Hence, $\bm \omega^{-1} (\bm Y - \bm \xi) = \bm \tau_{\bm h}(\bm Z^*) \sim \mathcal{SNTH}_p(\bm 0,\textbf{I}_p,\textbf{I}_p,\bm \eta^*,\bm h). $\qed
\end{proof}

Since the canonical form of the $\mathcal{SNTH}$ distribution is not exactly an affine transformation, it cannot be used for deriving the measures of multivariate skewness and kurtosis introduced by \cite{1970.V.K.M.Biometrika} and \cite{1973.J.F.M.A.A.JASA}. However, it can be used for reducing the dimensionality of the problem, when applicable, such as simulating observations from the $\mathcal{SNTH}$ distribution.

\section{Inference for the $\mathcal{SNTH}$ Distribution}

In this section, we discuss how to estimate parameters and perform tests for the $\mathcal{SNTH}$ distribution. 

\subsection{Parameter Estimation for the $\mathcal{SNTH}$ Distribution}

To estimate the parameters of the $\mathcal{SNTH}$ distribution, we use the method of maximizing the likelihood function. Suppose $\bm Y_1, \ldots, \bm Y_n$ is a random sample of size $n$ from the $\mathcal{SNTH}_p (\bm \xi, \bm \omega,\bar{\bm \Psi},\bm \eta,\bm h)$ distribution with $\bm Y_i = (Y_{i1},\ldots,Y_{ip})^\top$, $i = 1,\ldots,n$. For an observed sample $\bm y_1, \ldots,\bm y_n$, with $\bm y_i = (y_{i1},\ldots,y_{ip})^\top$, $i = 1,\ldots,n$, the log-likelihood function based on Equation \eqref{eq:snth-density} is 
\begin{equation}\label{eq:SNTH_full_llh}
\begin{split}
    \ell(\bm \theta) &= \log(2) - \frac{np}{2} \log(2 \pi) - \frac{n}{2} \log\{\det(\bar{\bm \Psi} +\bm \eta \bm \eta^\top)\} - \frac{1}{2} \sum_{i = 1}^n  \bm g (\bm y_i)^\top (\bar{\bm \Psi} + \bm \eta \bm \eta^\top)^{-1} \bm g (\bm y_i)\\
    &~~~~~ + \sum_{i = 1}^n  \Phi \left \{ \dfrac{\bm \eta^\top \bar{\bm \Psi}^{-1} \bm g(\bm y_i)}{\sqrt{1+\bm \eta^\top \bar{\bm \Psi} \bm \eta}} \right\} - n \sum_{j = 1}^p \log(\omega_{jj})
    + \sum_{i = 1}^n \sum_{j = 1}^p   \frac{1}{2} W_0 \left \{ h_j \left( \frac{y_{ij} - \xi_j}{ \omega_{jj}} \right)^2 \right \} \\
    &~~~~~- \sum_{i = 1}^n \sum_{j = 1}^p \log \left ( h_j \left( \frac{y_{ij} - \xi_j}{ \omega_{jj}} \right)^2 + \exp \left[ W_0 \left\{h_j \left( \frac{y_{ij} - \xi_j}{ \omega_{jj}} \right)^2  \right\} \right] \right ),
\end{split}
\end{equation}
where $\bm \theta = (\bm \xi^\top, \mbox{diag}(\bm \omega)^\top, \text{vech}(\bar{\bm \Psi})^\top, \bm \eta^\top,\bm h^\top)^\top$, where $\text{vech}(\bar{\bm \Psi})^\top$ is the vector of all the upper-off-diagonal elements of $\bar{\bm \Psi}$. We estimate the parameters in $\bm \theta$ by maximizing $\ell(\bm \theta)$ with respect to $\bm \theta$. This maximization cannot be done analytically and has to be done numerically. Hence, for a $p$-dimensional problem, we need to perform a $\{4p +p(p-1)/2\}$-dimensional numerical optimization, which becomes difficult when $p$ is large. We can tackle this problem in a different way. 

Since $\bm Y_1 ,\ldots,\bm Y_n \overset{\text{i.i.d.}}{\sim} \mathcal{SNTH}_p (\bm \xi, \bm \omega,\bar{\bm \Psi},\bm \eta,\bm h)$, from Proposition \ref{prop:marginal_SNTH}
we also have that $Y_{1j},\ldots,Y_{nj} \overset{\text{i.i.d.}}{\sim} \mathcal{SNTH}_1 ( \xi_j,  \omega_{jj},1, \eta_j, h_j)$, $j = 1,\ldots,p$. Based on the $j^\text{th}$ marginal data, the marginal log-likelihood function is 
\begin{equation}\label{eq:SNTH_marginal_llh}
\begin{split}
    \ell_j (\xi_j,\omega_{jj},\eta_j,h_j) &= \log(2) - \frac{n}{2} \log(2 \pi) - \frac{n}{2} \log(1 +\eta_j^2) - \frac{1}{2} \sum_{i = 1}^n  \frac{g_j (y_{ij})^2}{1 +\eta_j^2} \\
    &~~~~~+ \sum_{i = 1}^n  \Phi \left \{ \dfrac{ \eta_j g_j ( y_{ij})}{\sqrt{1+\eta_j^2}} \right\} - n  \log(\omega_{jj})
    + \sum_{i = 1}^n   \frac{1}{2} W_0 \left \{ h_j \left( \frac{y_{ij} - \xi_j}{ \omega_{jj}} \right)^2 \right \} \\
    &~~~~~- \sum_{i = 1}^n  \log \left ( h_j \left( \frac{y_{ij} - \xi_j}{ \omega_{jj}} \right)^2 + \exp \left[ W_0 \left\{h_j \left( \frac{y_{ij} - \xi_j}{ \omega_{jj}} \right)^2  \right\} \right] \right ),
\end{split}
\end{equation}
$j = 1,\ldots,p$. We estimate $\xi_j$, $\omega_{jj}$, $\eta_j$, and $h_j$, by maximizing the log-likelihood function for the $j^\text{th}$ marginal $\ell_j (\xi_j,\omega_{jj},\eta_j,h_j)$, $j = 1,\ldots,p$. Therefore, by performing four-dimensional numerical optimization $p$ times, we obtain the marginal maximum likelihood estimates (MLEs) for $\bm \xi$, $\bm \omega$, $\bm \eta$, and $\bm h$. 

At this point, we are yet to obtain the estimate for $\bar{\bm \Psi}$. From the definition of the $\mathcal{SNTH}$ distribution, we have $\bm Y_i \overset{\text{d}}{=} \bm \xi + \bm \omega \bm \tau_{\bm h} (\bm Z_i)$, $i=1,\ldots,n$ and $\bm Z_1,\ldots,\bm Z_n \overset{\text{i.i.d.}}{\sim} \mathcal{SN}_p(\bm 0,\bar{\bm \Psi},\bm \eta)$. With the marginal MLEs $\widehat{\bm \xi}$, $\widehat{\bm \omega}$, $\widehat{\bm \eta}$, and $\widehat{\bm h}$ of $\bm \xi$, $\bm \omega$, $\bm \eta$, and $\bm h$, we can compute an estimate for the latent $\mathcal{SN}$ observations. Then, $\widehat{\bm Z}_i = \bm \tau_{\widehat{\bm h}} ^{-1} \{\widehat{\bm \omega}^{-1}( \bm Y_i - \widehat{\bm \xi}) \} $, $i = 1,\ldots,n$ are the estimates for $\bm Z_1,\ldots,\bm Z_n$. Assuming that, $\widehat{\bm Z}_1,\ldots,\widehat{\bm Z}_n \overset{\text{i.i.d.}}{\sim} \mathcal{SN}_p(\bm 0,\bar{\bm \Psi},\widehat{\bm \eta})$ we can estimate~$\bar{\bm \Psi}$. 

We use the EM algorithm for the $\mathcal{SN}$ distribution for estimating $\bar{\bm \Psi}$, keeping the location and the skewness parameter fixed at $\bm 0$ and $\widehat{\bm \eta}$. The EM algorithm does not ensure that the estimate of $\bar{\bm \Psi}$ will be a correlation matrix, but the estimate is a covariance matrix, which can be easily converted to its corresponding correlation matrix. We use this correlation matrix as an estimate for $\bar{\bm \Psi}$. In the next section, we will justify the effectiveness of the described method for estimating parameters using a simulation study. Moreover, if we use the marginal MLEs of $\bm \xi$, $\bm \omega$, $\bm \eta$ and $\bm h$ and the estimate of $\bar{\bm \Psi}$ obtained from the EM algorithm as the initial value for the numerical maximization of $\ell(\bm \theta)$ in Equation \eqref{eq:SNTH_full_llh}, we can converge to the joint MLEs of $\bm \theta$ in very few iterations. Although it does not completely tackle the problem of high-dimensional numerical maximization, this specific selection of initial values reduces the run-time of the numerical maximization greatly. Moreover, we will show in our simulation study that the initial parameter values obtained in the aforementioned way are close to the joint MLEs and can be directly used for high-dimensional problems as the computation required for estimating the initial estimates is linear in $p$. In the next subsection, we describe the EM algorithm for the $\mathcal{SN}$ distribution in details. Note that instead of computing the marginal MLEs of the parameters one can use the iterative generalized method of moments (IGMM) estimators proposed by \cite{2011.G.M.G.TAAS}. IGMM is also based on the estimates of the latent observations and from there estimating the parameters corresponding to the latent random vector. While using the IGMM estimators for the $\mathcal{SNTH}$ distribution one has to keep in mind that the location and the scale parameters used in its definition are not the mean and the marginal standard deviation of the latent random vectors, unlike the proposal of \cite{2011.G.M.G.TAAS}. The IGMM has to be adapted accordingly for getting the correct estimates of the parameters.

\subsection{EM Algorithm for the $\mathcal{SN}$ Distribution}
The EM algorithm for the skew-normal distribution is a well-researched topic. Interested readers are directed to the recent paper by \cite{2021.T.A.H.F.T.K.C.L.ES} and the references therein for more on this topic. In this section, we put forward an EM algorithm for the skew-normal distribution with $\bm \Psi$-$\bm \eta$ parameterization (see \eqref{eq:SN_mpdf}), which is new in the literature. Moreover, we are only concerned with the scenario when we need to estimate the scale parameter $\bm \Psi$ while the location $\bm \xi = \bm 0$ and the skewness parameter $\bm \eta$ is known.

Consider a random sample $\bm Z_1,\ldots,\bm Z_n \overset{\text{i.i.d.}}{\sim} \mathcal{SN}_p(\bm 0,\bm \Psi,\bm \eta_0)$, where $\bm \eta_0$ is given. The log-likelihood of an observed sample $\bm z_1,\ldots,\bm z_n$ is 
\begin{align*}
    \ell(\bm \Psi) &= -\frac{np}{2} \log(2 \pi) - \frac{n}{2} \log \{ \det(\bm \Psi + \bm \eta_0 \bm \eta_0^\top) \} - \frac{1}{2} \sum_{i= 1}^n \bm z_i^\top (\bm \Psi + \bm \eta_0 \bm \eta_0^\top)^{-1} \bm z_i + \sum_{i = 1}^n \log \left \{ 2 \Phi \left(\dfrac{\bm \eta_0^\top \bm \Psi^{-1} \bm z_i}{\sqrt{1+ \bm \eta_0 ^\top \bm \Psi^{-1} \bm \eta_0}}\right) \right\}.
\end{align*}

Using the stochastic representation of the $\mathcal{SN}$ distribution we can represent $\bm Z_1,\ldots,\bm Z_n$ as
$(\bm Z_i | U_i = u_i) \overset{\text{i.i.d.}}{\sim} \mathcal{N}_p(u_i \bm \eta_0, \bm \Psi)$, $U_i  \overset{\text{i.i.d}}{\sim} \mathcal{HN}(0,1)$,  $i = 1,\ldots,n$ and obtain
the conditional pdf of $(U_i|\bm Z_i=\bm z_i)$ as
\begin{align*}
    f_{U_i|(\bm Z_i = \bm z_i}(u) & \propto  \phi_p(\bm z_i;u_i \bm \eta_0,\bm \Psi) \phi(u;0,1)\\
    &= \phi_p(\bm z_i;\bm 0,\bm \Psi + \bm \eta_0 \bm \eta_0^\top) \phi\{u; \bm \eta_0^\top (\bm \Psi + \bm \eta_0 \bm \eta_0^\top)^{-1} \bm z_i, 1 - \bm \eta_0^\top (\bm \Psi + \bm \eta_0 \bm \eta_0^\top)^{-1} \bm \eta_0 \}\\
    &= \phi_p (\bm z_i;\bm 0,\bm \Psi + \bm \eta_0 \bm \eta_0^\top) \phi \left(u;\tau_i,\frac{1}{1+\alpha^2}\right),\quad u > 0,\quad i = 1,\ldots,n,
\end{align*}
where $\alpha^2 = \bm \eta_0^\top \bm \Psi^{-1} \bm \eta_0$ and $\tau_i = \frac{\bm \eta_0^\top \bm \Psi^{-1} \bm z_i}{1+\alpha^2}$. Hence, the conditional distribution of the latent
variables $U_i$ given the observable $\bm Z_i$ is
$$(U_i|\bm Z_i = \bm z_i) \overset{\text{i.i.d}}{\sim} \mathcal{TN} \left(0;\tau_i,\frac{1}{1+\alpha^2} \right), \quad i = 1,\ldots,n.$$
Moreover, the first and second order raw moments of $(U_i|\bm Z_i = \bm z_i)$ are
$$v_{1i} = \mathbb{E}(U_i|\bm Z_i = \bm z_i) = \dfrac{\bar{\tau}_i+ \frac{\phi(\bar{\tau}_i)}{\Phi(\bar{\tau}_i)}}{\sqrt{1+\alpha^2}},\quad v_{2i} = \mathbb{E}(U_i^2 |\bm Z_i = \bm z_i) = \dfrac{ 1 + \bar{\tau}_i^2+ \bar{\tau}_i\frac{\phi(\bar{\tau}_i)}{\Phi(\bar{\tau}_i)}}{1+\alpha^2} ,\quad i = 1,\ldots,n, $$
where $\bar{\tau}_i = \sqrt{1+\alpha^2} \tau_i = \frac{\bm \eta_0^\top \bm \Psi^{-1} \bm z_i}{\sqrt{1+\alpha^2}}$. 

From the hierarchical representation above, the complete log-likelihood for $\bm \Psi$ based on the observed data $\bm z = (\bm z_1,\ldots,\bm z_n)^\top$ and the missing data $\bm u = (u_1,\ldots,u_n)^\top$ is
\begin{align*}
    \ell_c (\bm \Psi| \bm z, \bm u) &= -\frac{np}{2} \log(2 \pi)+ \frac{n}{2} \log\{ \det(\bm \Lambda) \} - \frac{1}{2} \sum_{i=1}^n \bm z_i^\top \bm \Lambda \bm z_i + \bm \eta_0^\top \bm \Lambda \sum_{i = 1}^n u_i \bm z_i\\
    &~~~~~- \frac{1}{2} \bm \eta_0^\top \bm \Lambda \bm \eta_0 \sum_{i = 1}^n u_i^2 + \frac{n}{2} \log \left( \frac{2}{\pi} \right) - \frac{1}{2} \sum_{i = 1}^n u_i^2,
\end{align*}
where $\bm \Lambda = \bm \Psi^{-1}$.

Let $\bm Z = (\bm Z_1,\ldots,\bm Z_n)^\top$ be the observable random sample and $\bm U = (U_1,\ldots,U_n)^\top$ be the latent random sample. Then the E-Step at the $(k+1)^\text{th}$ iteration of the EM algorithm is
\begin{align*}
    Q(\bm \Psi|\bm \Psi^{(k)}) &= \mathbb{E}_{\bm \Psi^{(k)}} \{ \ell_c (\bm \Psi| \bm Z, \bm U) |\bm Z = \bm z \} \\
    &= -\frac{np}{2} \log(2 \pi)+ \frac{n}{2} \log\{ \det(\bm \Lambda ) \} - \frac{1}{2} \sum_{i=1}^n \bm z_i^\top \bm \Lambda \bm z_i + \bm \eta_0^\top \bm \Lambda \sum_{i = 1}^n v_{1i}^{(k)} \bm z_i\\
    &~~~~~- \frac{1}{2}\bm \eta_0^\top \bm \Lambda \bm \eta_0 \sum_{i = 1}^n v_{2i}^{(k)} + \frac{n}{2} \log \left( \frac{2}{\pi} \right) - \frac{1}{2} \sum_{i = 1}^n v_{2i}^{(k)},
\end{align*}
where $\bm \Psi^{(k)}$ is the estimated value of $\bm \Psi$ in the $k^{\text{th}}$ step, $\bm \Lambda^{(k)} = \{\bm \Psi^{(k)}\}^{-1}$,
$$v_{1i}^{(k)} =  \dfrac{\bar{\tau}_i^{(k)}+ \frac{\phi(\bar{\tau}_i^{(k)})}{\Phi(\bar{\tau}_i^{(k)})}}{\sqrt{1+\{\alpha^{(k)}\}^2}},\quad v_{2i}^{(k)}  = \dfrac{ 1 + \{\bar{\tau}_i^{(k)}\}^2+ \bar{\tau}_i ^{(k)} \frac{\phi(\bar{\tau}_i^{(k)})}{\Phi(\bar{\tau}_i^{(k)})}}{1+\{\alpha^{(k)}\}^2} ,$$
$\bar{\tau}_i^{(k)} = [1 + \{\alpha^{(k)}\}^2]^{-1/2} \bm \eta_0 ^\top \bm \Lambda^{(k)} \bm z_i$, $\alpha^{(k)} = \sqrt{\bm \eta_0 ^\top \bm \Lambda^{(k)} \bm \eta_0}$. To get the $(k+1)^\text{th}$ estimate of $\bm \Psi$, we maximize $Q(\bm \Psi|\bm \Psi^{(k)})$ with respect to $\bm \Psi$ and update $\bm \Psi^{(k+1)} = \text{argmax} \{Q(\bm \Psi|\bm \Psi^{(k)})\}$.

Since $\bm \Psi$ is a symmetric positive definite matrix according to our definition of the $\mathcal{SN}$ distribution, we can write $\bm \Psi^{-1} = \bm \Lambda = \bm C^\top \bm C$, where $\bm C \in \mathbb{R}^{p \times p}$ is a nonsingular matrix. Hence, 
\small{
\begin{align*}
   &~~~~~ Q(\bm \Psi|\bm \Psi^{(k)}) \propto  \frac{n}{2} \log\{ \det(\bm C^\top \bm C ) \} - \frac{1}{2} \sum_{i=1}^n \bm z_i^\top \bm C^\top \bm C \bm z_i + \bm \eta_0^\top \bm C^\top \bm C \sum_{i = 1}^n v_{1i}^{(k)} \bm z_i- \frac{1}{2} \bm \eta_0^\top  \bm C^\top \bm C \bm \eta_0 \sum_{i = 1}^n v_{2i}^{(k)} \\
   & \Rightarrow \dfrac{\partial Q(\bm \Psi|\bm \Psi^{(k)})}{\partial \bm C} = n (\bm C^\top)^{-1} - \bm C \sum_{i = 1}^n \bm z_i \bm z_i^\top + \bm C \sum_{i = 1}^n \left( \bm \eta_0  v_{1i}^{(k)} \bm z_i^\top +  v_{1i}^{(k)} \bm z_i \bm \eta_0^\top \right) - \bm C \bm \eta_0  \bm \eta_0^\top \sum_{i = 1}^n v_{2i}^{(k)} = \bm 0 \\
   & \Rightarrow (\bm C^\top \bm C)^{-1} = \frac{1}{n} \sum_{i = 1}^n \bm z_i \bm z_i^\top + \bm \eta_0  \bm \eta_0^\top \left( \frac{1}{n}\sum_{i = 1}^n v_{2i}^{(k)}\right) - \frac{1}{n} \sum_{i = 1}^n \left( \bm \eta_0  v_{1i}^{(k)} \bm z_i^\top +  v_{1i}^{(k)} \bm z_i \bm \eta_0^\top \right) .
\end{align*}
}
Therefore, we update
\begin{equation*}
    \bm \Psi^{(k+1)} = \frac{1}{n} \sum_{i = 1}^n \bm z_i \bm z_i^\top + \bm \eta_0  \bm \eta_0^\top \left( \frac{1}{n}\sum_{i = 1}^n v_{2i}^{(k)}\right) - \frac{1}{n} \sum_{i = 1}^n \left( \bm \eta_0  v_{1i}^{(k)} \bm z_i^\top +  v_{1i}^{(k)} \bm z_i \bm \eta_0^\top \right).
\end{equation*}
We stop the algorithm when $ \{\ell(\bm \Psi^{(k+1)})/\ell(\bm \Psi^{(k)}) - 1\}$ is sufficiently close to $0$.

\subsection{Tests Based on the $\mathcal{SNTH}$ Distribution}

It is a well-known fact \citep{2012.M.H.C.L.B} that the Fisher information matrix of the $\mathcal{ASN}$ and the $\mathcal{SN}$ distributions is singular when the skewness parameter, $\bm \alpha$ or $\bm \eta$, is set to zero. As a result, we cannot use the Wald type test or the likelihood ratio test (LRT) for testing the null hypothesis that the skewness parameter is zero based on the $\mathcal{ASN}$ or the $\mathcal{SN}$ distribution. Although the asymptotic distribution of the LRT statistic is $\chi^2_p$ for the univariate $\mathcal{ASN}$ or the univariate $\mathcal{SN}$ distribution, i.e. for $p = 1$, the same is not true for $p>1$; see \cite{2023.S.M.R.B.A.V.M.G.G.SP}. The explanation of why the asymptotic distribution of the LRT statistic is $\chi^2_1$ for the univariate $\mathcal{ASN}$ or the univariate $\mathcal{SN}$ is still an open problem. 

For the skew-$t$ distribution, this singularity of the Fisher information matrix does not occur when the skewness parameter is set to zero. Hence, we can perform the test of the null hypothesis that the skewness parameter is zero based on the skew-$t$ distribution using the Wald type test or the LRT. Next, we show that the Fisher information matrix of the $\mathcal{SNTH}_2$ distribution, when the skewness parameter is set to zero, remains nonsingular.

\begin{proposition}\label{prop:FIM_SNTH_bivariate}
The Fisher information matrix for a bivariate random vector $\bm Y \sim \mathcal{SNTH}_2 (\bm \xi,\bm \omega,\bar{\bm \Psi},\bm \eta,\bm h)$ is nonsingular when $\bm \eta = \bm 0$.
\end{proposition}
\begin{proof}
From Equation \eqref{eq:SNTH_full_llh}, the log-likelihood function for $\bm Y = \bm y = (y_1,y_2)^\top$ is
\begin{align*}
    \ell (\bm \theta) &= - \log(\pi) - \frac{1}{2} \log \{ \det(\bar{\bm \Psi} + \bm \eta \bm \eta^\top)\} - \frac{1}{2} \bm g(\bm y)^\top (\bar{\bm \Psi} + \bm \eta \bm \eta^\top)^{-1} \bm g (\bm y) + \log \left[ \Phi \left\{ \dfrac{\bm \eta^\top \bar{\bm \Psi}^{-1} \bm g (\bm y) }{\sqrt{1+\bm \eta^\top \bar{\bm \Psi}^{-1} \bm \eta}} \right\} \right]\\
    &~~~~~+ \sum_{i=1}^2 \Bigg ( - \log(\omega_{ii}) + \frac{1}{2} W_0 \left(h_i x_i ^2\right)
    - \log \left[ h_i x_i ^2 + \exp\left\{W_0 \left(h_i x_i ^2\right)\right\}\right] \Bigg ),
\end{align*}
where $x_i = \left( \frac{y_i - \xi_i}{\omega_{ii}} \right)$, $i=1,2$. The score functions of all the parameters are obtained by differentiating the log-likelihood with respect to the parameters. Assuming that $\bar{\bm \Psi} = \begin{psmallmatrix} 1 & \rho \\ \rho & 1 \end{psmallmatrix}$,
the score functions of all the parameters, when $\bm \eta = \bm 0$, are listed below for $i=1,j=2$ or $i=2,j=1$:
\small{
\begin{align*}
    &S_{\xi_i} = \dfrac{1}{\omega_{ii}}  \left(  \dfrac{ x_i -  \rho x_j \exp \left \{\frac{1}{2} W_0 (h_i x_i^2)-\frac{1}{2} W_0 (h_j x_j^2) \right \} }{(1-\rho^2) [ h_i x_i^2 + \exp \left\{ W_0 (h_i x_i^2) \right\}]} + \dfrac{ h_ix_i [ h_i x_i^2 + 3  \exp \{ W_0 (h_i x_i^2)\}]}{[h_i x_i^2 + \exp \{ W_0 (h_i x_i^2)\}]^2} \right),\\
    %&S_{\xi_2} = \dfrac{1}{\omega_{22}}  \left(  \dfrac{ x_2 -  \rho x_1 \exp \left \{\frac{1}{2} W_0 (h_2 x_2^2)-\frac{1}{2} W_0 (h_1 x_1^2) \right \} }{(1-\rho^2) [h_2 x_2^2 + \exp \left\{ W_0 (h_2 x_2^2) \right\}]} + \dfrac{ h_2x_2 [ h_2 x_2^2 + 3  \exp \{ W_0 (h_2 x_2^2)\}]}{[h_2 x_2^2 + \exp \{ W_0 (h_2 x_2^2)\}]^2} \right),\\
    &S_{\omega_{ii}}= \dfrac{1}{\omega_{ii}}  \left(  \dfrac{ x_i^2 -  \rho x_i x_j \exp \left \{\frac{1}{2} W_0 (h_i x_i^2)-\frac{1}{2} W_0 (h_j x_j^2) \right \} }{(1-\rho^2)[h_i x_i^2 + \exp \left\{ W_0 (h_i x_i^2) \right\}]} + \dfrac{ \exp \{ W_0 (h_i x_i^2)\} [ h_i x_i^2 -  \exp \{ W_0 (h_i x_i^2)\}]}{[h_i x_i^2 + \exp \{ W_0 (h_i x_i^2)\}]^2} \right), \\
    %&S_{\omega_{22}}= \dfrac{1}{\omega_{22}}  \left( \dfrac{ x_2^2 -  \rho x_1 x_2 \exp \left \{\frac{1}{2} W_0 (h_2 x_2^2)-\frac{1}{2} W_0 (h_1 x_1^2) \right \} }{(1-\rho^2)[h_2 x_2^2 + \exp \left\{ W_0 (h_2 x_2^2) \right\}]} + \dfrac{ \exp \{ W_0 (h_2 x_2^2)\} [ h_2 x_2^2 -  \exp \{ W_0 (h_2 x_2^2)\}]}{[h_2 x_2^2 + \exp \{ W_0 (h_2 x_2^2)\}]^2} \right), \\
    &S_{\eta_i} = \sqrt{\dfrac{2}{\pi}} \left \{ \dfrac{
     g_i(y_i)-\rho g_j(y_j)}{(1-\rho^2)} \right\},\\
    %&S_{\eta_2} = \sqrt{\dfrac{2}{\pi}} \left \{ \dfrac{g_2(y_2)-\rho g_1(y_1)}{(1-\rho^2)} \right\},\\
    &S_{h_i}= \dfrac{1}{2}  \left(  \dfrac{ x_i^4 \exp\left\{-W_0(h_i x_i^2)\right\} -  \rho x_i^3 x_j \exp \left \{-\frac{1}{2} W_0 (h_i x_i^2)-\frac{1}{2} W_0 (h_j x_j^2) \right \} }{(1-\rho^2)[h_i x_i^2 + \exp \left\{ W_0 (h_i x_i^2) \right\}]}- \dfrac{ h_i x_i^4 + 3 x_i^2 \exp \{ W_0 (h_i x_i^2)  }{[h_i x_i^2 + \exp \{ W_0 (h_i x_i^2)\}]^2} \right),\\
    %&S_{h_2}= \dfrac{1}{2}  \left(  \dfrac{ x_2^4 \exp\left\{-W_0(h_2 x_2^2)\right\} -  \rho x_2^3 x_1 \exp \left \{-\frac{1}{2} W_0 (h_2 x_2^2)-\frac{1}{2} W_0 (h_1 x_1^2) \right \} }{(1-\rho^2)[h_2 x_2^2 + \exp \left\{ W_0 (h_2 x_2^2) \right\}]}- \dfrac{ h_2 x_2^4 + 3 x_2^2 \exp \{ W_0 (h_2 x_2^2)  }{[h_2 x_2^2 + \exp \{ W_0 (h_2 x_2^2)\}]^2} \right),\\
    &S_\rho = \dfrac{g_1(y_1)g_2(y_2)}{(1-\rho^2)} - \dfrac{\rho}{(1-\rho^2)^2} \{ g_1^2(y_1) + g_2^2(y_2) - 2 \rho g_1(y_1) g_2(y_2) \} + \dfrac{\rho}{(1-\rho^2)}.
\end{align*}
}
From the form of the score functions we can observe that they are not linearly dependent when $\bm \eta = \bm 0$ and hence the Fisher information matrix, which is the variance-covariance matrix of the score vector, is nonsingular when $\bm \eta = \bm 0$.
\qed
\end{proof}

Proposition \ref{prop:FIM_SNTH_bivariate} demonstrates that the Fisher information matrix of the $\mathcal{SNTH}$ distribution is nonsingular when $\bm \eta = \bm 0$ for $p = 2$. Our conjecture is that this statement remains true for $p>2$. We justify this by plotting, in Figure \ref{fig:LRT}, the histogram of the LRT statistic for testing $H_0: \bm \eta = \bm 0$ vs $H_1: \bm \eta \neq \bm 0$ for $p = 2$, $3$, and $4$, based on samples of size $5000$ and $1000$ replicates. Along with the histograms, we also plot the $\chi_p^2$ pdf. The plots indicate that the asymptotic distribution of the LRT statistic indeed follows $\chi^2_p$, for $p = 2$,  $3$, and $4$. This would not have been the case if the Fisher information matrix was singular for $\bm \eta = \bm 0$.

\begin{figure}[htp!]
\begin{center}
\centering
  \includegraphics[width=\linewidth]{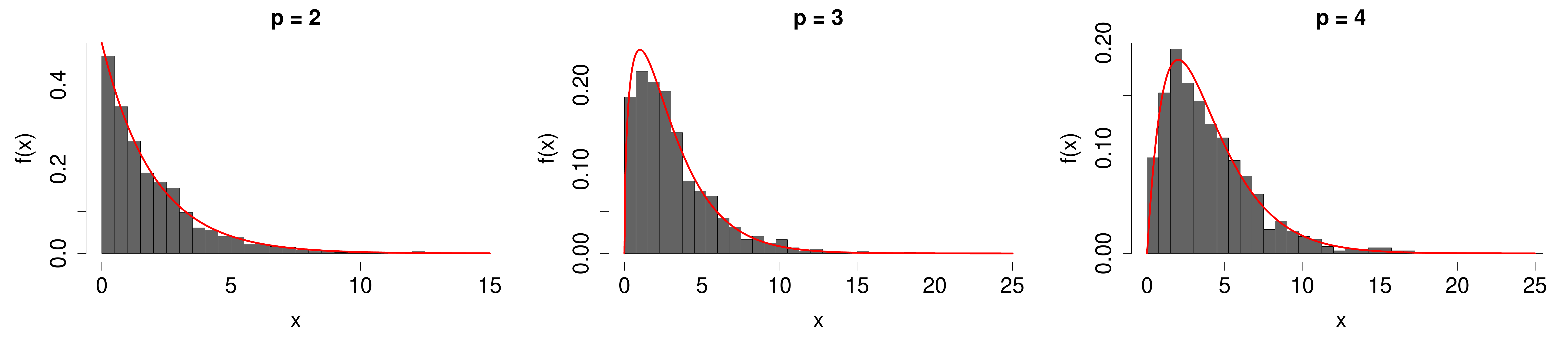}
\caption{Histograms of the LRT statistic for testing $H_0: \bm \eta = \bm 0$ vs $H_1: \bm \eta \neq \bm 0$ for $\mathcal{SNTH}_p$ when $p = 2$, $3$, and $4$ based on samples of size $5000$ and $1000$ replicates. The red curves indicate the pdf of the $\chi^2_p$ distribution.}
\label{fig:LRT}
\end{center}
\end{figure}

Although we have justified the nonsingularity of the Fisher information matrix for the $\mathcal{SNTH}$ distribution when $\bm \eta = \bm 0$, we do not have the mathematical form of the Fisher information matrix. As a result, we cannot use the Wald type test for testing $\bm \eta = \bm 0$. We have to rely on the LRT for that:
\begin{itemize}
    \item \textit{Testing $H_0:\bm \eta = \bm 0$ vs $H_1:\bm \eta \neq \bm 0$, given that $\bm h \neq \bm 0$}:\\ 
    Since the Fisher information matrix of the $\mathcal{SNTH}$ distribution when $\bm \eta = \bm 0$ is nonsingular, given that $\bm h \neq \bm 0$, we use the asymptotic distribution of the LRT statistic for conducting the test.
    \item \textit{Testing $H_0:\bm h = \bm 0$ vs $H_1:\bm h \neq \bm 0$, given that $\bm \eta \neq \bm 0$}:\\
    Under the null hypothesis the $\mathcal{SNTH}$ distribution becomes the $\mathcal{SN}$ distribution. The Fisher information matrix of the $\mathcal{SN}$ distribution is nonsingular when $\bm \eta \neq \bm 0$. Hence, under the null hypothesis we can use the asymptotic distribution of the LRT statistic for conducting the test.
    \item \textit{Testing $H_0:\bm \eta = \bm 0 \text{~and~} \bm h = \bm 0$ vs $H_1:\bm \eta \neq \bm 0 \text{~or~} \bm h \neq \bm 0$} :\\
    Under the null hypothesis, the Fisher information matrix is singular. Hence, we cannot use the LRT anymore for this testing problem. However, since the asymptotic distribution of the LRT statistic for testing $\eta = 0$ vs $\eta \neq 0$ based on the univariate $\mathcal{SN}$ is $\chi_1^2$, we can use the LRT for testing $H_{i0}: \eta_i = 0, h_i = 0$ vs $H_{i1}: \eta_i \neq 0 \text{~or~} h_i \neq 0$, $i = 1,\ldots,p$. We reject $H_0$ if any of the $H_{i0}$ gets rejected. Note here that the rejection region for testing $H_{i0}$ vs $H_{i1}$, $i = 1,\ldots,p$, has to be computed subject to Bonferroni's correction.
\end{itemize}

\section{Simulation Study}

We conduct two simulation studies in this section: one to demonstrate the effectiveness of the parameter estimation method described in Sections 4.1 and 4.2, and another to show in which scenarios the $\mathcal{SNTH}$ distribution is more suitable compared to the skew-$t$ distribution.

\subsection{$\mathcal{SNTH}$ Parameter Estimation}

We test the methodology for $\mathcal{SNTH}$ parameter estimation in a simulation study. We simulate observations of size $n = 50$, $100$ $200$, $500$, and $1000$ from a $\mathcal{SNTH}_3(\bm \xi,\bm \omega,\bar{\bm \Psi},\bm \eta,\bm h)$, with $\bm \xi = (0.8,-0.6,1.3)^\top$, $\bm \omega = \text{diag}(3,5,2)$, $\bar{\bm \Psi} = \begin{psmallmatrix}
1 &-0.5&0.3\\ -0.5&1&-0.2\\0.3&-0.2&1\\
\end{psmallmatrix}$, $\bm \eta = (-1.5,2,0.5)^\top$ and $\bm h = (0.02,0.08,0.03)^\top$. Based on the simulated data, we estimate the parameters by the methodology described in Sections 4.1 and 4.2. We repeat the process $100$ times and summarize the estimated parameter in boxplots in Figure \ref{fig:boxplot_snth}. Alongside the estimates obtained from the methodology described in Section~4.1 (indicated as mMLE (short for marginal MLE) for $\bm \xi$, $\bm \omega$, $\bm \eta$, $\bm h$ and as EM for $\bar{\bm \Psi}$ in Figure \ref{fig:boxplot_snth}) we also report the MLEs of all the parameters as well. The boxplots indicate that the methodology is working reasonably well for estimating the parameters from the $\mathcal{SNTH}$ model. Moreover, as the sample size increases, the variance of the estimates decreases, as it should. Hence, we can say that the parameter estimation methodology described in Sections 4.1 and 4.2 is justified. The boxplots also show that the estimates of the parameters obtained from the EM algorithm are not very different from the MLEs, although they have more variability. The variability difference between the two estimation methods also decreases as the sample size increases. For problems with high dimensions where the computation of the exact MLEs are infeasible, one can use the methodology described in Sections 4.1 and 4.2 as an alternative. Moreover, these estimates are an excellent choice for the starting values of the parameters when optimizing the exact log-likelihood for computing the MLEs.

\begin{figure}[h!]
% \centering
% \begin{subfigure}{0.15\textwidth}
%   \centering
%   \includegraphics[width=\linewidth]{xi.pdf}
%   \vspace{-0.35in}
%   \caption{Estimates of $\bm \xi$}
%   \label{fig:xi_snth}
% \end{subfigure}
% \begin{subfigure}{0.15\textwidth}
%   \centering
%   \includegraphics[width=\linewidth]{omega.pdf}
%   \vspace{-0.35in}
%   \caption{Estimates of $\bm \omega$}
%   \label{fig:omega_snth}
% \end{subfigure}
% \begin{subfigure}{0.15\textwidth}
%   \centering
%   \includegraphics[width=\linewidth]{eta.pdf}
%   \vspace{-0.35in}
%   \caption{Estimates of $\bm \eta$}
%   \label{fig:eta_snth}
% \end{subfigure}
% \begin{subfigure}{0.15\textwidth}
%   \centering
%   \includegraphics[width=\linewidth]{h.pdf}
%   \vspace{-0.35in}
%   \caption{Estimates of $\bm h$}
%   \label{fig:h_snth}
% \end{subfigure}
% \begin{subfigure}{0.15\textwidth}
%   \centering
%   \includegraphics[width=\linewidth]{Psi.pdf}
%   \vspace{-0.35in}
%   \caption{Estimates of the off-diagonals of $\bar{\bm \Psi}$}
%   \label{fig:Psi_snth}
% \end{subfigure}
\centering
\includegraphics[width=0.99\linewidth]{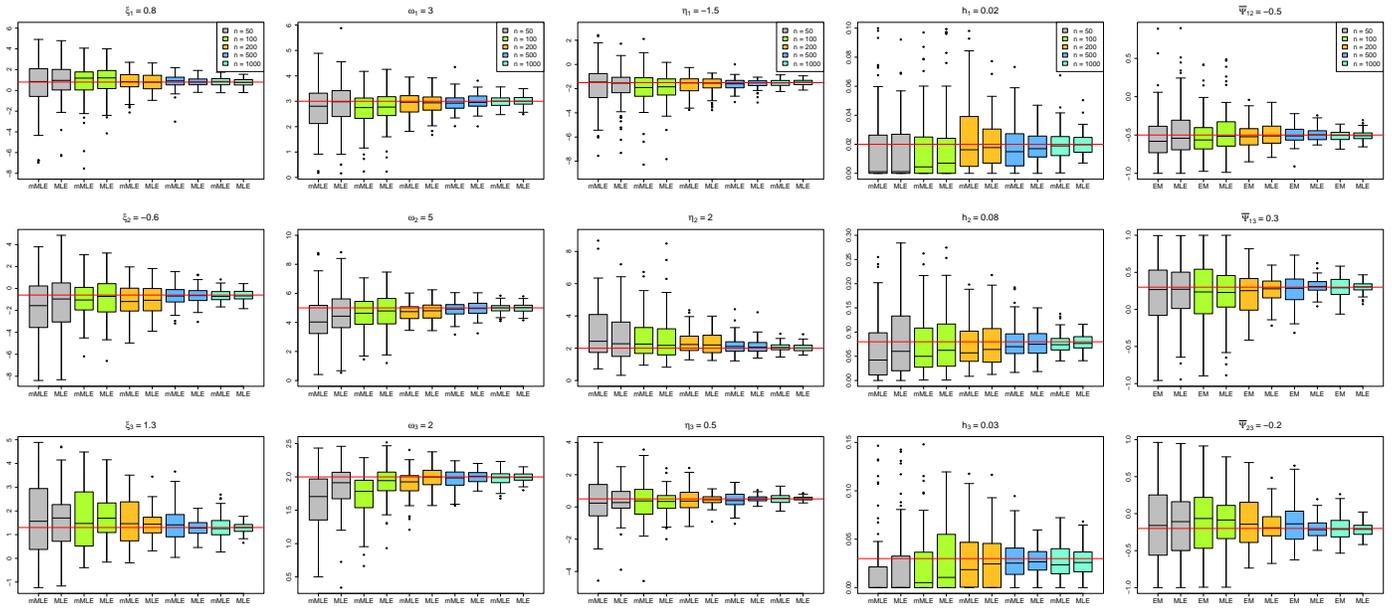}
\caption{Boxplots of the parameter estimates (100 replicates) of a $\mathcal{SNTH}_3$ distribution obtained from the methodology in Sections 4.1 and 4.2 for different sample sizes $n$, given as mMLE (marginal MLE) for $\bm \xi$, $\bm \omega$, $\bm \eta$, $\bm h$ and as EM for $\bar{\bm \Psi}$ along with the MLE boxplots. The red line in each plot indicates the true parameter value.}
\label{fig:boxplot_snth}
\end{figure}

\subsection{Comparison Between the $\mathcal{SNTH}$ and the Skew-$t$ Distributions}

\begin{figure}[t!]
\begin{center}
\centering
\begin{subfigure}{0.25\textwidth}
  \centering
  \includegraphics[width=\linewidth]{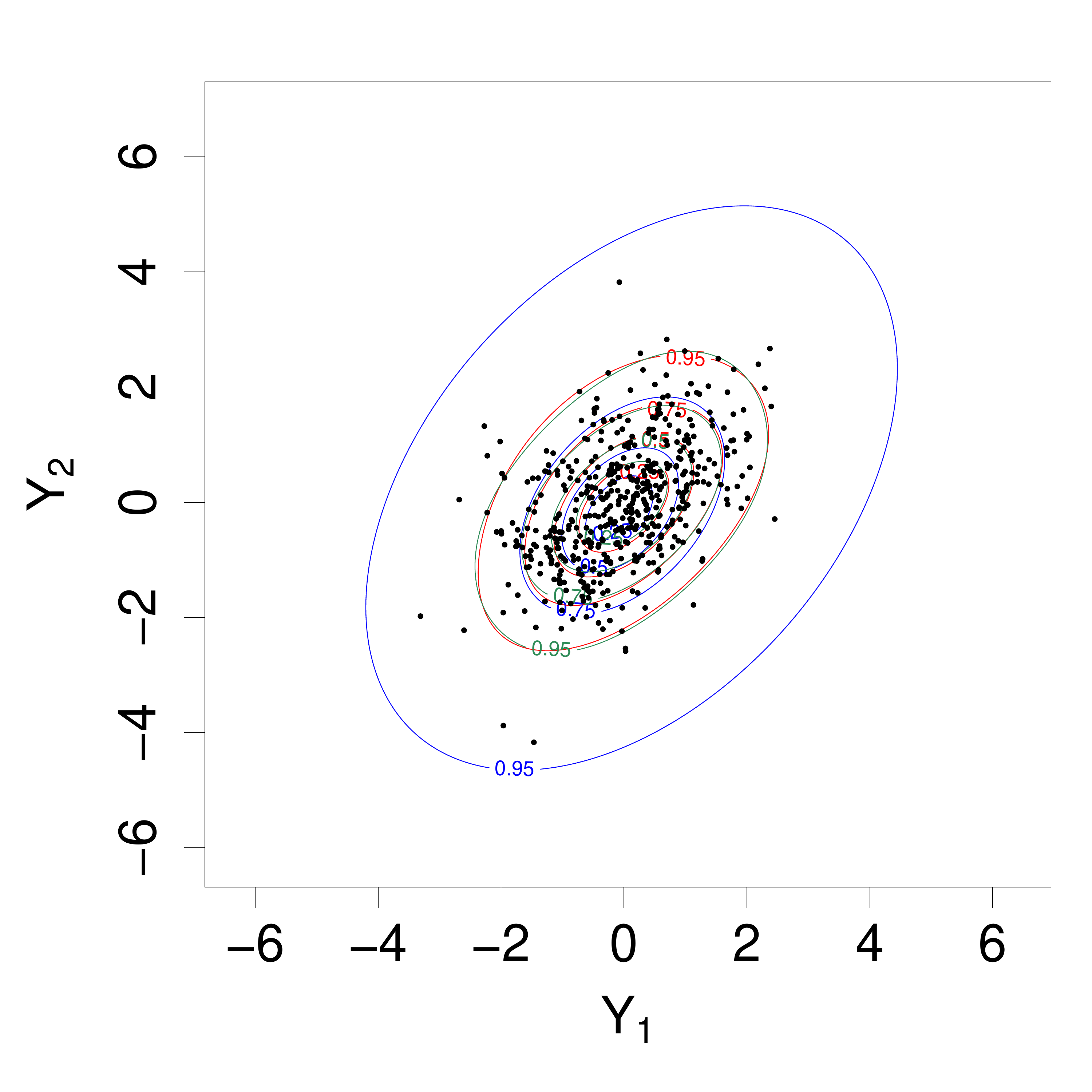}
  %\vspace{-0.35in}
  
 % \label{fig:boxplot_2.1}
\end{subfigure}%
\begin{subfigure}{0.25\textwidth}
  \centering
  \includegraphics[width=\linewidth]{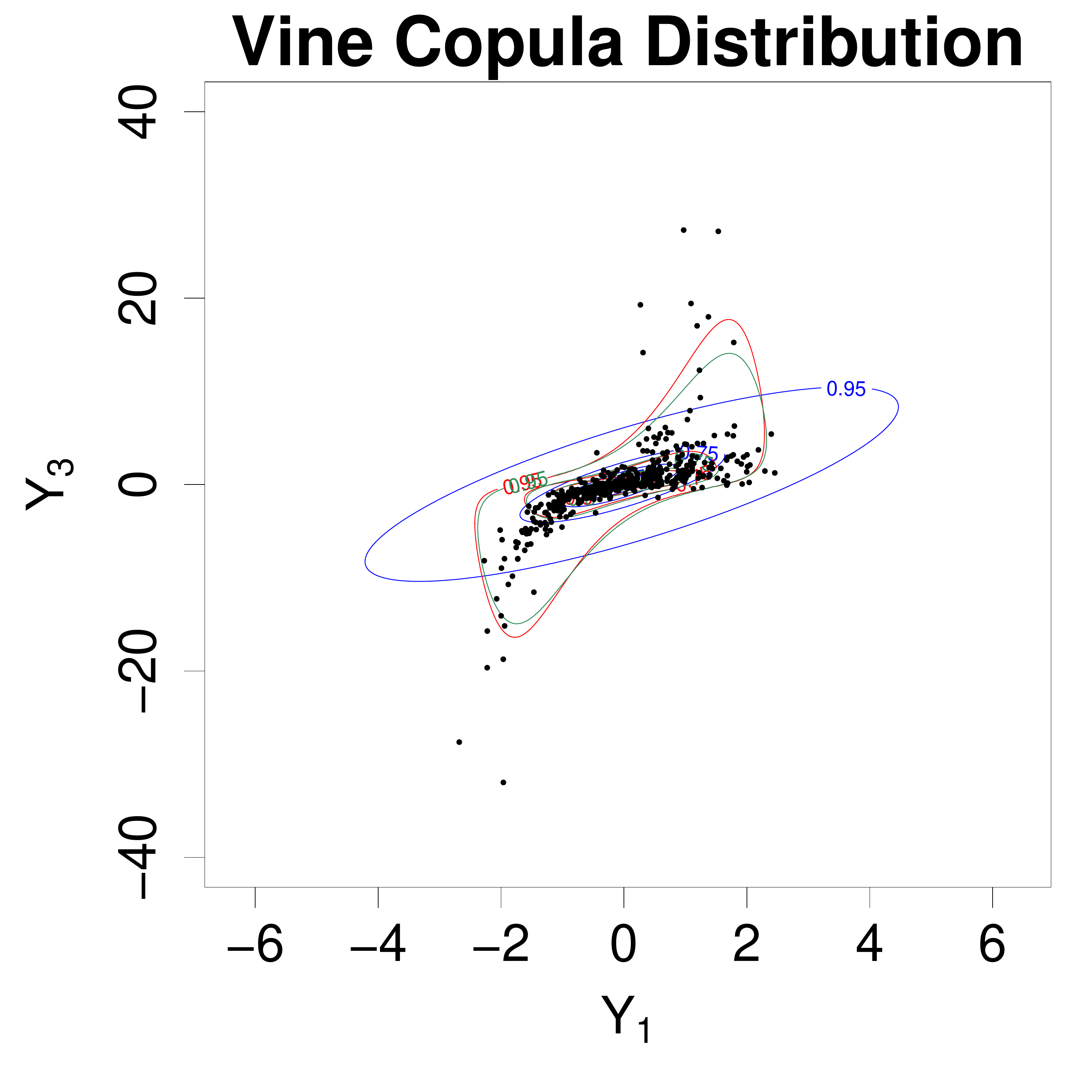}
  %\vspace{-0.35in}
  
 % \label{fig:boxplot_2.2}
\end{subfigure}
\begin{subfigure}{0.25\textwidth}
  \centering
  \includegraphics[width=\linewidth]{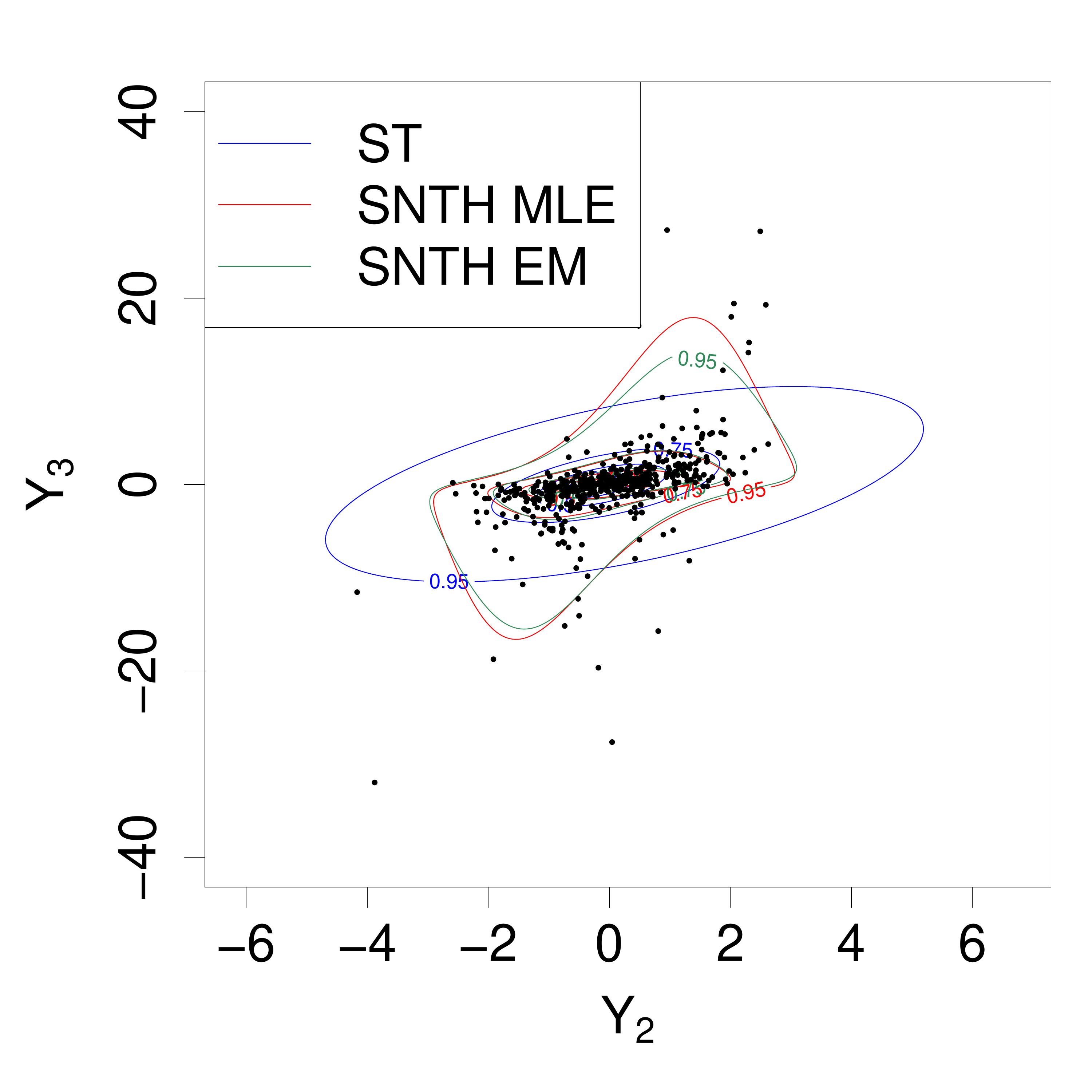}
  %\vspace{-0.35in}
  
 % \label{fig:boxplot_2.2}
\end{subfigure}
\begin{subfigure}{0.25\textwidth}
  \centering
  \includegraphics[width=\linewidth]{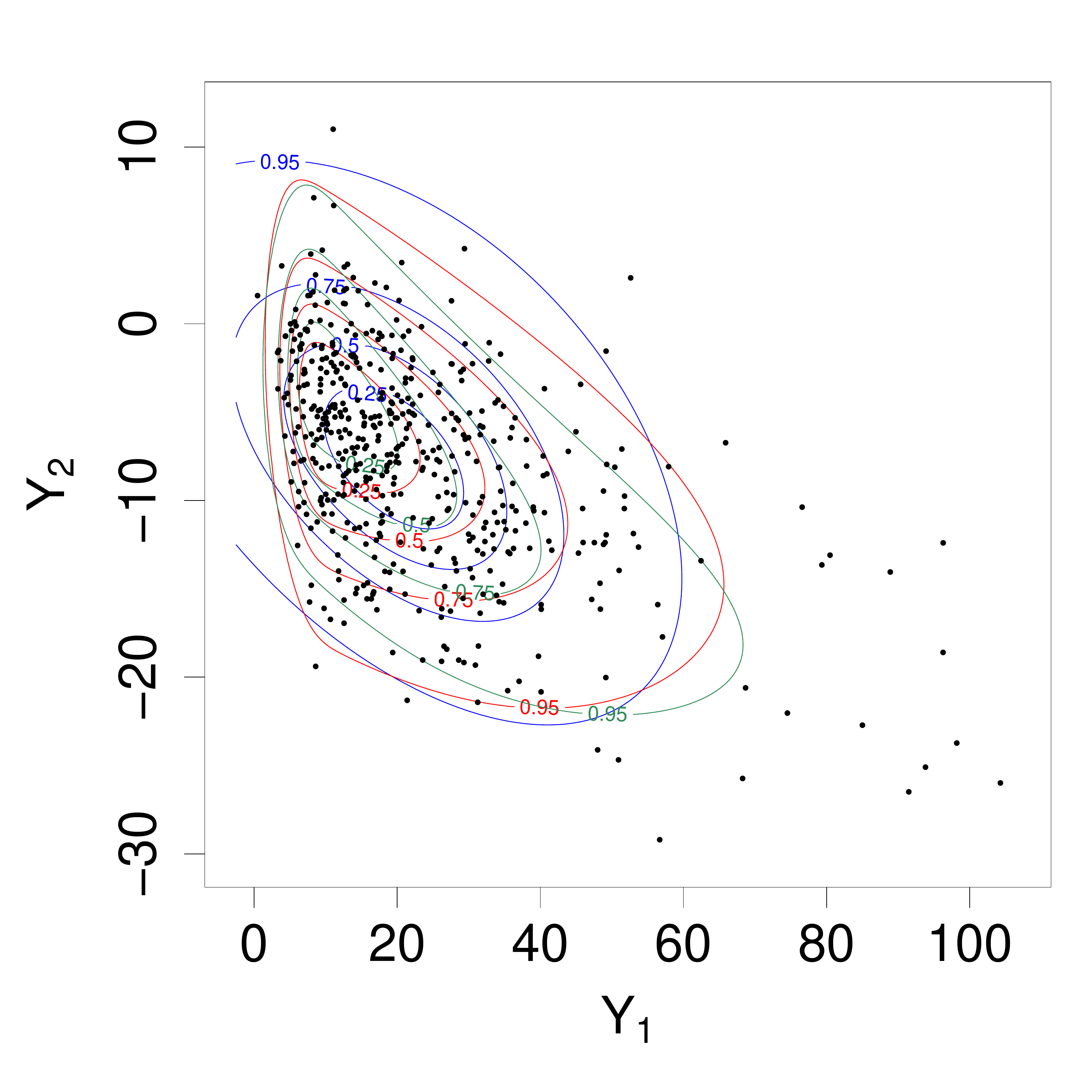}
  % \vspace{-0.35in}
  
 % \label{fig:boxplot_2.1}
\end{subfigure}%
\begin{subfigure}{0.25\textwidth}
  \centering
  \includegraphics[width=\linewidth]{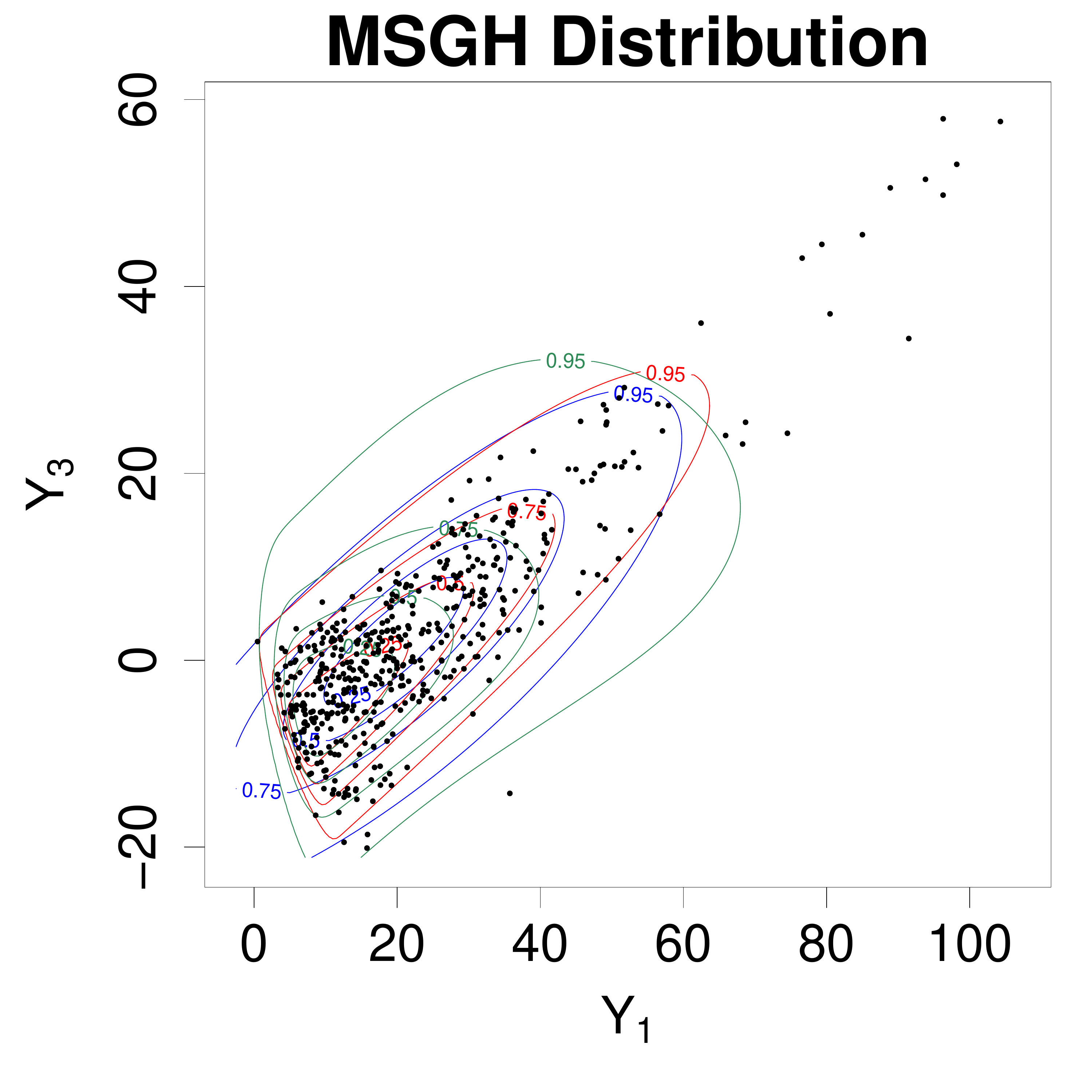}
  % \vspace{-0.35in}
  
 % \label{fig:boxplot_2.2}
\end{subfigure}
\begin{subfigure}{0.25\textwidth}
  \centering
  \includegraphics[width=\linewidth]{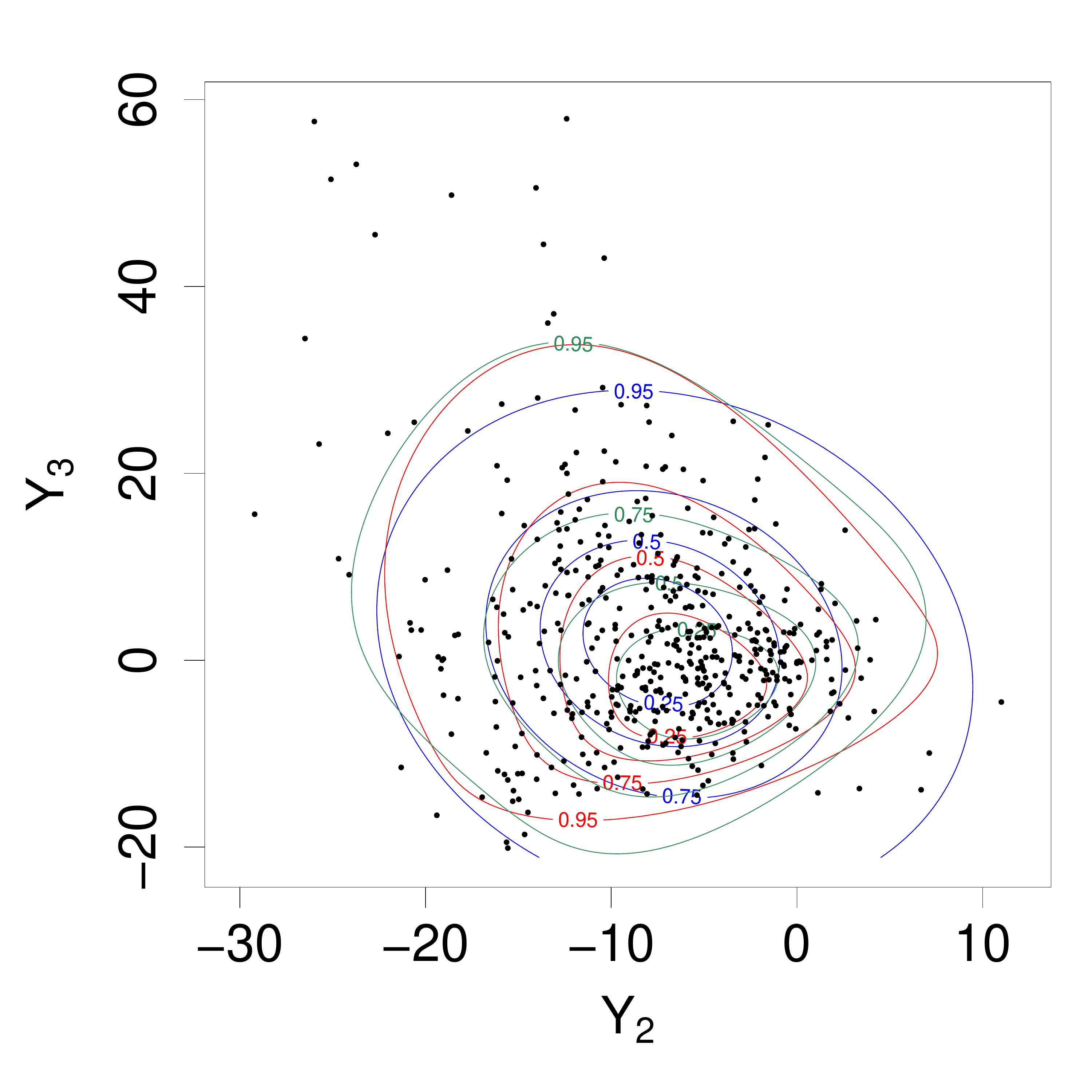}
  % \vspace{-0.35in}
  
 % \label{fig:boxplot_2.2}
\end{subfigure}

\begin{subfigure}{0.25\textwidth}
  \centering
  \includegraphics[width=\linewidth]{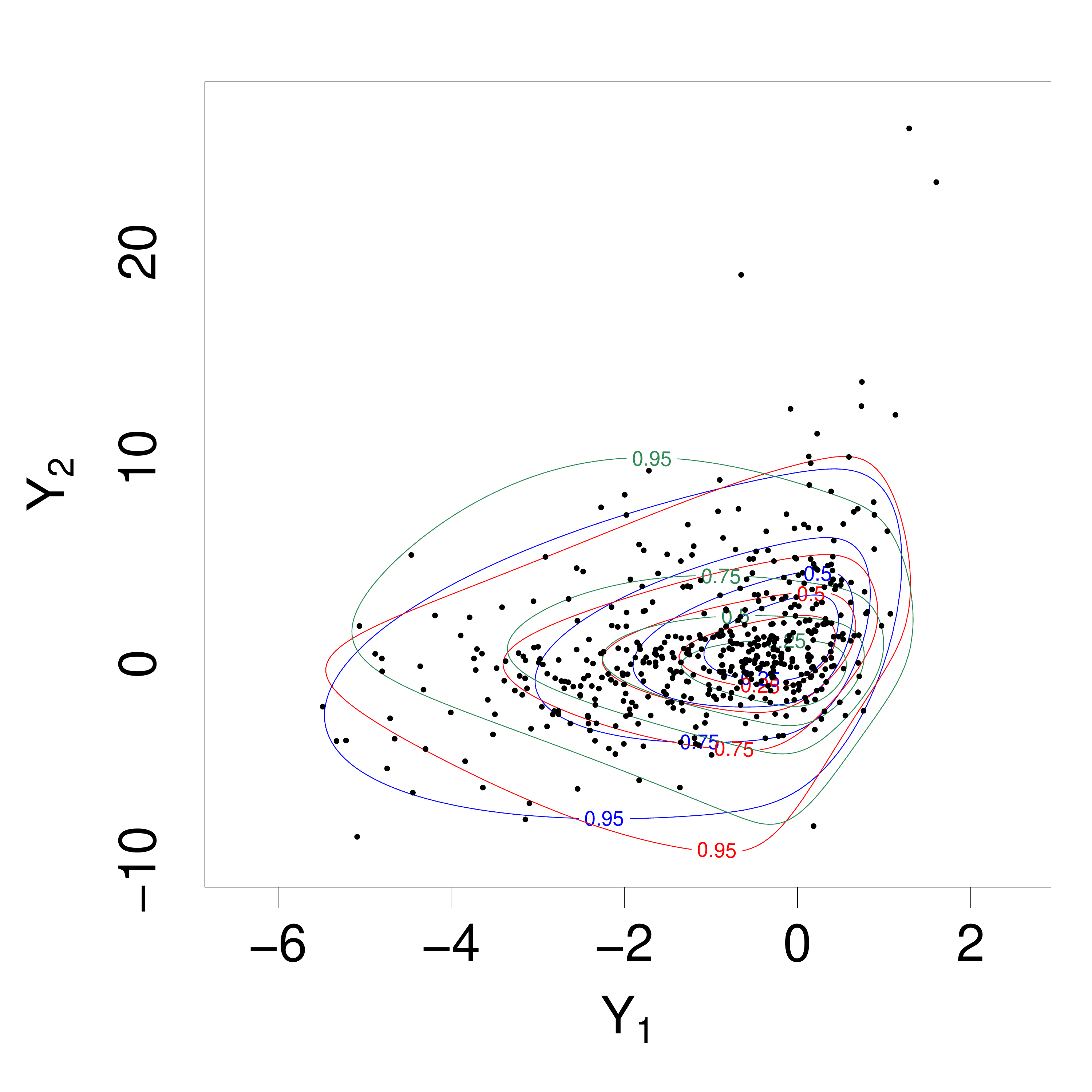}
  \vspace{-0.35in}
  
 % \label{fig:boxplot_2.2}
\end{subfigure}
\begin{subfigure}{0.25\textwidth}
  \centering
  \includegraphics[width=\linewidth]{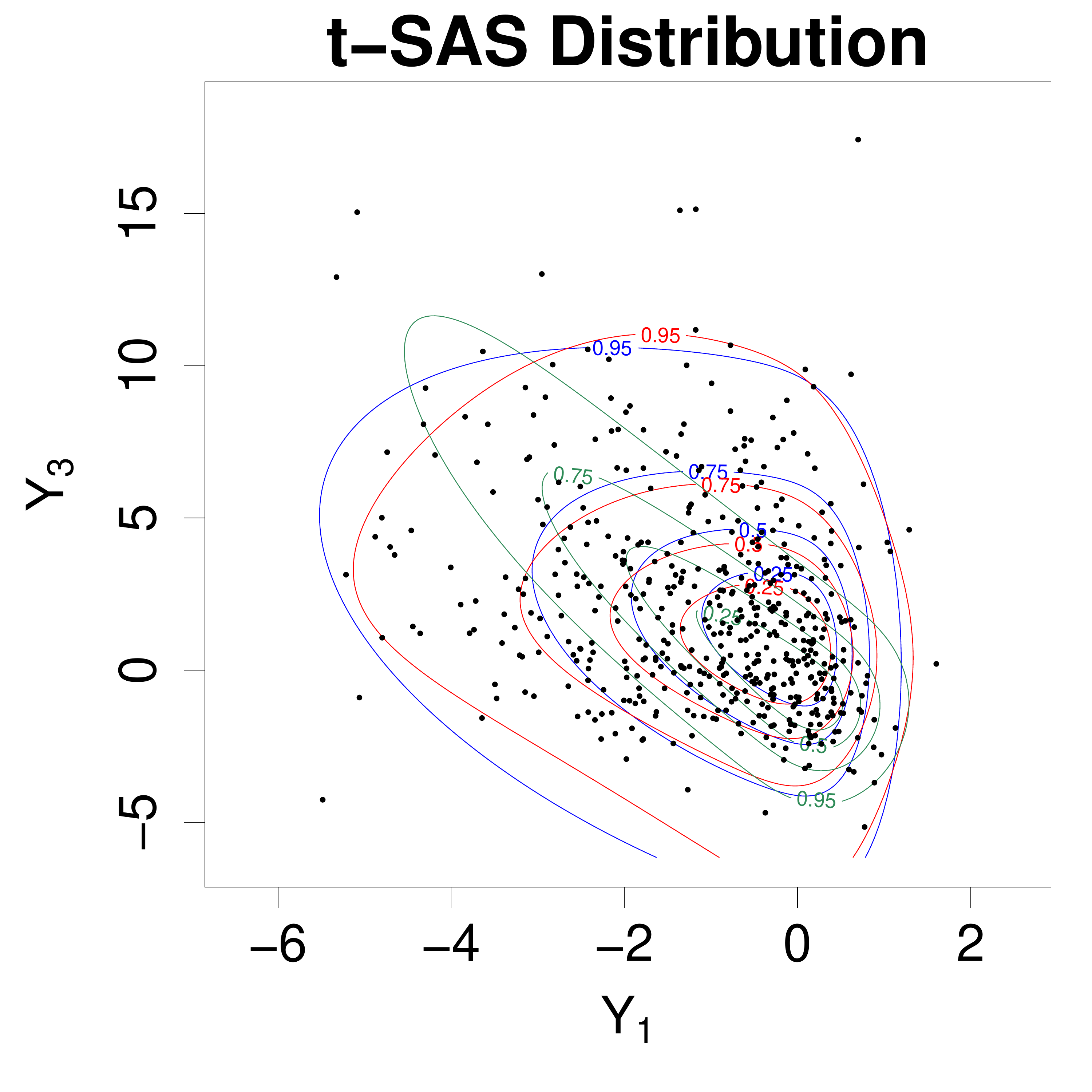}
  \vspace{-0.35in}
  
 % \label{fig:boxplot_2.1}
\end{subfigure}%
\begin{subfigure}{0.25\textwidth}
  \centering
  \includegraphics[width=\linewidth]{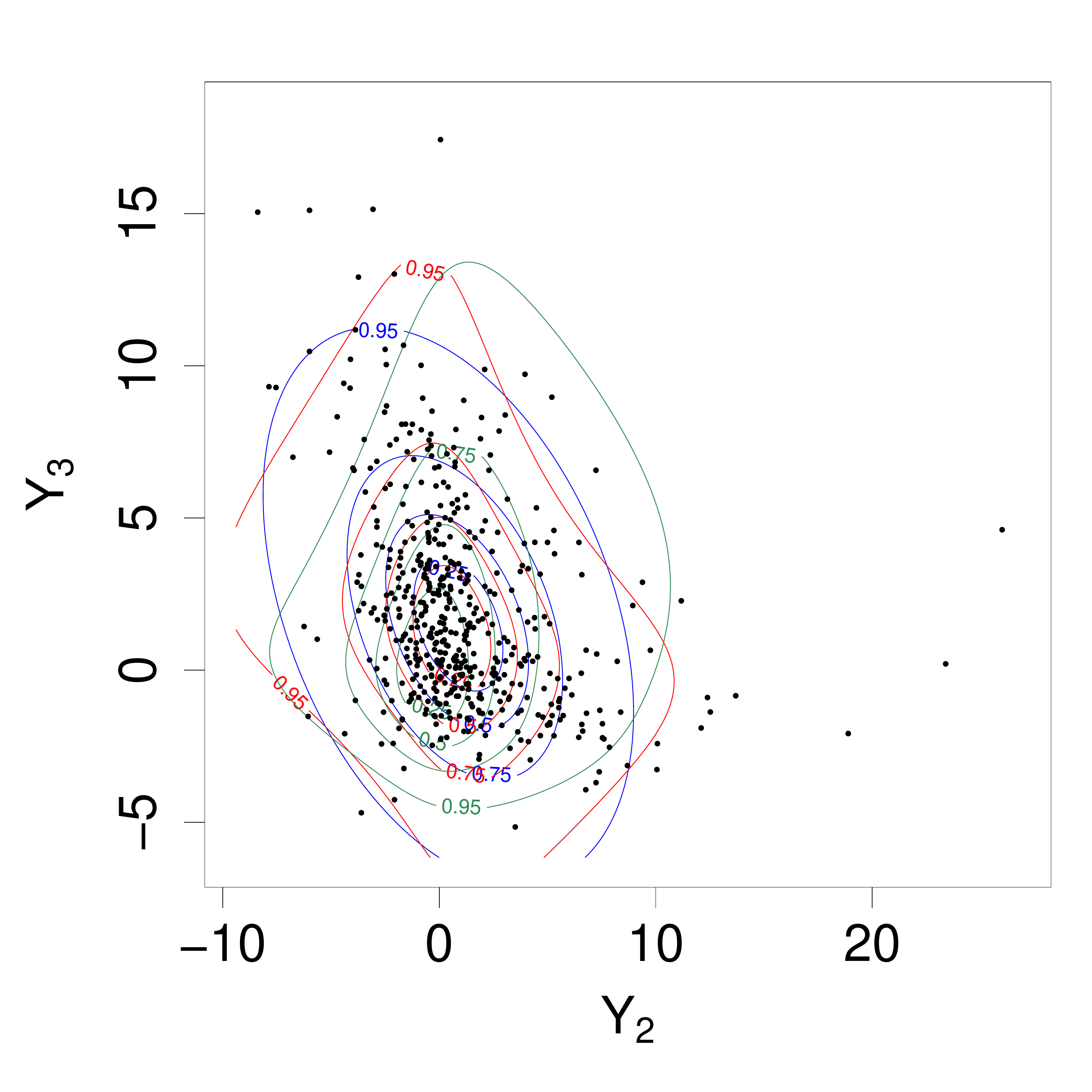}
  \vspace{-0.35in}
  
 % \label{fig:boxplot_2.2}
\end{subfigure}

\caption{Bivariate contours of the marginal bivariate pdfs obtained from the fitted $\mathcal{SNTH}$ using Sections 4.1 and 4.2 methodology (green), from the fitted $\mathcal{SNTH}$ using MLEs (red) and from the fitted skew-$t$ (blue) distributions to trivariate vine copula data (first row), $\mathcal{MSGH}$ data (second row), and $t$-SAS data (third row). The contours correspond  to $0.25$, $0.5$, $0.75$, and $0.95$ approximate probability regions.}
\label{fig:snth_vs_st_contour}
\end{center}
\end{figure}

In this simulation study we show that when there is a great disparity between the marginal kurtosis values in a multivariate dataset, the $\mathcal{SNTH}$ distribution is more appropriate than the skew-$t$ distribution. We generate $500$ random samples from a three-dimensional vine copula to create a trivariate dataset in Uniform$(0,1)$ scale. In this vine copula model, variables $1$ and $2$ are related with a Gaussian copula with $\rho = 0.5$, variables $1$ and $3$ are related with a Clayton copula with parameter $4.8$ and variables $2$ and $3$ given variable $1$ are related with a Gumbel copula with parameter $1.9$. On the trivariate simulated data, we transform the $1^{\text{st}}$ component to the standard normal scale, the $2^{\text{nd}}$ component to the Cauchy $t_1$ scale, and the $3^{\text{rd}}$ component to the Student's $t_{10}$ scale. We fit both the $\mathcal{SNTH}$ and the skew-$t$ distribution to this simulated data. The Akaike information criterion (AIC) computed for the $\mathcal{SNTH}$ and the skew-$t$ are $4393$ and $4848$, respectively, suggesting the $\mathcal{SNTH}$ distribution is more suitable for this simulated dataset, compared to the skew-$t$ distribution. 

We perform similar experiments where we generate $500$ observations from a three-dimensional multiple-scaled generalized hyperbolic $(\mathcal{MSGH})$ distribution \citep{2015.D.W.F.F.CSDA} and from a three-dimensional $t$-SAS distribution \citep{2019.S.B.C.L.D.V.S}. For the $\mathcal{MSGH}$ distribution we use the following parameters: $\bm \mu = (0,0,0)^\top$, $\bm \Sigma = \begin{psmallmatrix}
    1 & 0.3 & -0.2\\ 0.3 & 1 & -0.4\\ -0.2 & -0.4 &1
\end{psmallmatrix}$, $\bm \beta = (3,0.5,-0.2)^\top$, $\bm \lambda = (2,1,4)^\top$, $\bm \gamma = (\sqrt{3},\sqrt{0.2},\sqrt{0.25})^\top$, and $\delta = 1$. For the $t$-SAS distribution, we use a three-dimensional $t$-copula with correlation matrix $\begin{psmallmatrix}
    1 & 0.3 & -0.2\\ 0.3 & 1 & -0.4\\ -0.2 & -0.4 &1
\end{psmallmatrix}$ to generate observations on the uniform scale. For the Sinh-Arcsinh (SAS) transformation, we use $(-0.7,1)$, $(0.2,0.6)$, and $(0.5,0.8)$ as our $(g,h)$ (for skewness and tail-thickness, as used in \cite{2019.S.B.C.L.D.V.S}) parameters for the three marginals, respectively. Finally, we scale the marginals by $1$, $1.2$, and $1.8$, respectively. When the $\mathcal{SNTH}$ and the skew-$t$ model are fitted to the $\mathcal{MSGH}$ dataset the obtained AICs are $9982$ and $10110$, and for the $t$-SAS dataset, the AICs are $6606$ and $6634$. The AICs for both studies suggest that the $\mathcal{SNTH}$ is a better fit to these two datasets compared to the skew-$t$ model. 

We provide the contour plots of the bivariate marginal pdfs of the $\mathcal{SNTH}$ and the skew-$t$ distribution fitted to the three simulated datasets in Figure \ref{fig:snth_vs_st_contour}. The bivariate marginal pdfs for the $\mathcal{SNTH}$ distribution are obtained based on the MLEs and also based on the estimates from the EM algorithm. The contours are plotted for the $0.25$, $0.5$, $0.75$, and $0.95$ approximate probability regions. The plots show that, as expected, the skew-$t$ distribution cannot handle different tail-thickness for different marginals, and instead tries to find the best compromise with a single parameter, $\nu$. In scenarios like this, the $\mathcal{SNTH}$ distribution is more appropriate. Moreover, in the first row of Figure \ref{fig:snth_vs_st_contour} we see from the contour plots that the difference between the bivariate marginal pdfs obtained based on the MLE and the EM algorithm is small for the vine copula dataset. However, in the second and third rows of Figure \ref{fig:snth_vs_st_contour} the dissimilarity between the two $\mathcal{SNTH}$ parameter estimation methods is more prominent, especially for the $(Y_1,Y_3)$ pair. Finally, it is clear from the plots that the marginal bivariate $\mathcal{SNTH}$ pdf contours obtained from the MLEs are more suitable for all three datasets compared to the skew-$t$ counterparts.

\section{Data Applications}
We use two data applications to illustrate the effectiveness of the $\mathcal{SNTH}$ distribution over the skew-$t$ in certain situations. The parameter estimates and standard errors for the two data applications, as well as log-likelihood and AIC values along with computing times, are given in Sections S3 and S4 of the supplementary material.

\subsection{Italian Wine Dataset}

We consider a trivariate dataset consisting of the amount of chloride, glycerol and magnesium in a particular type of wine. The data were obtained from \cite{1986.M.F.C.A.M.C.M.U.V} and originally consist of measurements on $28$ chemicals from $178$ samples of Italian wines. Among these $178$ samples, $48$ originated from the Barbera region, $59$ from the Barolo region, and $71$ from the Grignolino region. Here we use the variables chloride, glycerol and magnesium for the Grignolino region as previously analyzed by \citet{2013.A.A.A.C.CUP} with a skew-$t$ distribution, hence $p = 3$ variables and $n = 71$ observations. 
\begin{figure}[b!]
\centering
\includegraphics[width=0.8\linewidth, height=5cm]{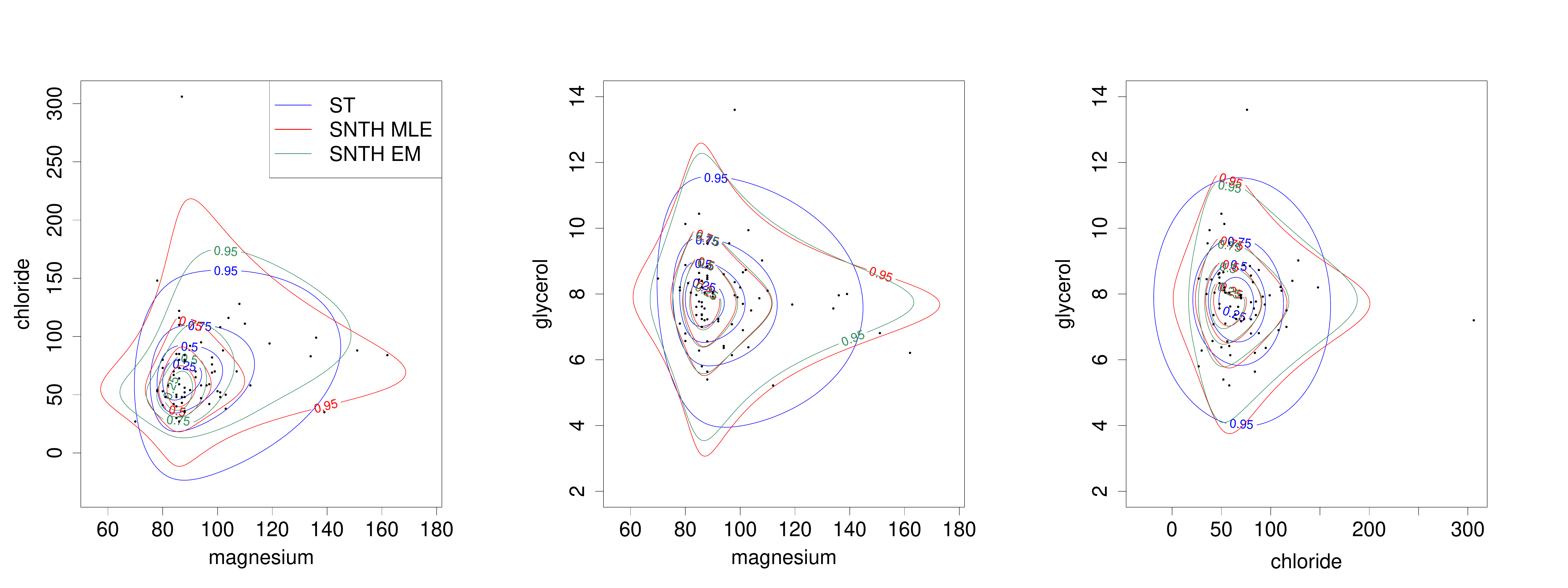}
  \caption{Bivariate contours of the marginal bivariate pdfs obtained from the fitted $\mathcal{SNTH}$ using Section 4.1 methodology (green), from the fitted $\mathcal{SNTH}$ using MLE (red) and the skew-$t$ (blue) distributions to the wine data. The contours correspond to $0.25$, $0.5$, $0.75$, and $0.95$ approximate probability regions.}
  \label{fig:data_application}
\end{figure}

The sample estimate of the marginal Pearson's measure of kurtosis for this dataset are $7.7$, $21.1$, and $7.9$, which suggest that the $\mathcal{SNTH}$ distribution might be more suitable for this dataset compared to the skew-$t$ distribution. We fit both the $\mathcal{SNTH}$ and the skew-$t$ distribution to this dataset. The contour plots of the bivariate marginal pdfs obtained from the two fitted distributions are presented in Figure \ref{fig:data_application}. For the $\mathcal{SNTH}$ model we have produced the contours of the bivariate marginal pdfs using MLEs (in red) and the EM algorithm estimates (in green) along with the skew-$t$ bivariate marginal pdfs (in blue). One can see visually that the $\mathcal{SNTH}$ distribution fits the data better than the skew-$t$. Moreover, the contour plots indicate that there are some discrepancies between the two estimation methodologies based on the $\mathcal{SNTH}$ distribution, specifically for the magnesium-chloride pair, but much less in the other two pairs. The difference is likely due to a relatively small sample size ($n = 71$). The AIC corresponding to the $\mathcal{SNTH}$ distribution and the skew-$t$ distribution are $1474$ and $1492$, respectively. Hence, for this dataset, the $\mathcal{SNTH}$ distribution is a better model than the skew-$t$ distribution. Moreover, assuming that $\bm \eta \neq \bm 0$, the $p$-value for testing $H_0:\bm h = \bm 0$ vs $H_1:\bm h \neq \bm 0$ is $2.53 \times10^{-14}$, using the LRT based on the $\mathcal{SNTH}$ distribution. This suggests that $\bm h \neq \bm 0$ for this dataset. Using the LRT for testing $H_0:\bm \eta = \bm 0$ vs $H_1:\bm \eta \neq \bm 0$ when $\bm h \neq \bm 0$ is $1.6 \times 10^{-5}$, hence confirming the apparent skewness in the data.

\subsection{Saudi Arabian Wind Speed Dataset}

We analyze the dependence structure of the daily average, minimum, and maximum wind speed in the city of Sharurah in southern Saudi Arabia, at $100$ meters in height (a typical hub height for wind turbines), in the year 2015. Understanding the dependence and distribution of these variables is important for setting up wind farms for harvesting wind energy. We remove a  quadratic trend from all three variables and fit an AR($1$) time series model to the detrended data marginally to obtain residuals. A Ljung-Box test shows that there is no significant serial correlation left in all three residuals. Hence, the residuals can be treated as a random sample of size $n=365$ from a trivariate distribution.

The sample estimates of the marginal Pearson's measure of kurtosis for the three variables are $3.0$, $7.5$, and $4.4$, which means that the residuals corresponding to the average windspeed have a Gaussian-like tail and the other two residuals have heavier tails than the Gaussian distribution. This indicates that the $\mathcal{SNTH}$ distribution may be more apt for this dataset compared to the skew-$t$ distribution. We fit both the $\mathcal{SNTH}$ and the skew-$t$ distribution to the residuals. Similar to the previous contour plots, we have produced in Figure \ref{fig:data_application_2} the contours of the bivariate marginal pdfs using MLEs (in red) and the EM algorithm estimates (in green) along with the skew-$t$ bivariate marginal pdfs (in blue). The plots indicate that the $\mathcal{SNTH}$ distribution is more suitable here for capturing different tail-thickness for different marginals, compared to the skew-$t$ distribution. This conclusion is further validated by the AIC which is $3274$ for the $\mathcal{SNTH}$ distribution and is $3432$ for the skew-$t$ distribution. Moreover, the difference between the contours obtained from the MLEs and from the EM algorithm estimates for the $\mathcal{SNTH}$ distribution are very close to each other. Similar to the wine dataset, we can perform the following tests: assuming that $\bm \eta \neq \bm 0$, the $p$-value for testing $H_0:\bm h = \bm 0$ vs $H_1:\bm h \neq \bm 0$ is $9.7 \times10^{-35}$, using the LRT based on the $\mathcal{SNTH}$ distribution, which confirms that the data here are not from a skew-normal distribution; $H_0:\bm \eta = \bm 0$ vs $H_1:\bm \eta \neq \bm 0$ when $\bm h \neq \bm 0$ is $2.56 \times 10^{-13}$, which confirms the presence of skewness in the data.
\begin{figure}[h!]
\centering
\includegraphics[width=0.8\linewidth, height=5cm]{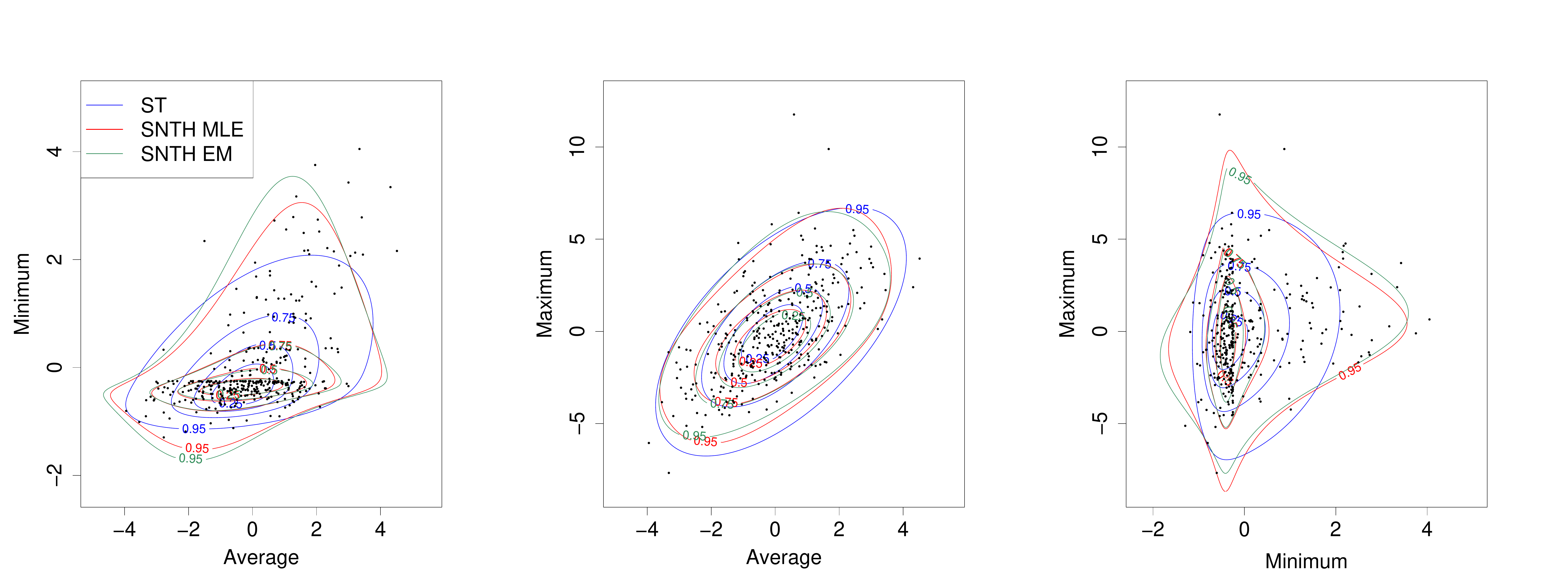}
  \caption{Bivariate contours of the marginal bivariate pdfs obtained from the fitted $\mathcal{SNTH}$ using the EM algorithm (green), from the fitted $\mathcal{SNTH}$ using MLE (red) and the skew-$t$ (blue) distributions to the wind speed residuals. The contours correspond to $0.25$, $0.5$, $0.75$, and $0.95$ approximate probability regions.}
  \label{fig:data_application_2}
\end{figure}

\vspace{-0.7cm}

\section{Discussion}

In this article, we have introduced the multivariate $\mathcal{SNTH}$ distribution, a new extension of the multivariate skew-normal distribution for modeling heavy-tailed data. We have compared our proposed distribution with the skew-$t$ distribution, another extension of the skew-normal distribution for adapting tail-thickness. Unlike the skew-$t$ distribution, our proposal is capable of handling data with different kurtosis for different marginals. As a consequence, the $\mathcal{SNTH}$ model can be used as a robust model, as suggested by \cite{2008.A.A.M.G.ISR} for the skew-$t$, for modeling outliers. Moreover, the $\mathcal{SNTH}$ distribution can capture outliers in some marginals while having Gaussian-like distributions in other marginals. We have discussed various appealing stochastic and inferential properties of the $\mathcal{SNTH}$ distribution in detail. A methodology for parameter estimation of the $\mathcal{SNTH}$ distribution was also provided. 

There are other proposals in the multivariate setup for modeling varying marginal tail-thickness, such as the $\mathcal{MSGH}$ distribution by \cite{2015.D.W.F.F.CSDA} and the $t$-SAS distribution by \cite{2019.S.B.C.L.D.V.S}. However, they lack appealing stochastic properties, such as a tractable conditional distribution and an explicit form of conditional mean and variance, unlike the $\mathcal{SNTH}$ model. How the $\mathcal{SNTH}$ model performs compared to these other multivariate models for modeling varying marginal tail-thickness is left as a future research direction.

The $\mathcal{SNTH}$ distribution can be further generalized by extending the idea of using transformation to induce tail-thickness in the distribution to the extended skew-normal ($\cal ESN$) family and the unified skew-normal ($\cal SUN$) family \citep{2006.R.B.A.V.A.A.SJS}. In Section 3.1, we have discussed how the $\mathcal{SNTH}$ distribution induces tail-thickness in the $\mathcal{SN}$ distribution by stretching the distribution along different axes, and this stretching can be different for different marginals. This idea could be further generalized where the stretching occurs along arbitrary directions. 

The EM algorithm in Section 4.2 discussed how we can estimate the scale matrix $\bm \Psi$ of an $\mathcal{SN}_p (\bm 0,\bm \Psi, \bm \eta_0)$ distribution, given that $\bm \eta_0$ is fixed. However, we need this $\bm \Psi$ to be a correlation matrix, not a covariance matrix. This is achieved by transforming the final estimate of $\bm \Psi$ from covariance to a correlation matrix. The EM algorithm for the scenario when $\bm \Psi$ is a correlation matrix is an open problem. 

The R-codes and real data for Sections 5 and 6 are available on a GitHub repository: \url{https://github.com/sagnikind/Skew-normal-Tukey-h}.

\setstretch{0.4}
\baselineskip 5pt
{\footnotesize
\bibliographystyle{abbrvnat}
\bibliography{SNTH_manuscript}
}

\pagebreak
\section*{Supplementary Material}
\subsection*{S1~~~ Proof of Proposition 8}
\begin{proof}
We have:
\begin{align*}
    \mathbb{E}(Y_{0_i}^2) &= \int_{\mathbb{R}} x^2 \exp(h_i x^2) \dfrac{1}{\Phi(\bar{\tau}_{1.2})} \phi(x;\xi_{1.2_i}+\bar{\tau}_{1.2} {\bar{ \eta}_{1.2_i}}, \bar{ \Psi}_{11.2_{ii}}+{\bar{ \eta}_{1.2_i}}^2) \Phi \left\{ \dfrac{\bar{\tau}_{1.2} + {\bar{ \eta}_{1.2_i}} (x - \xi_{1.2_i})/\bar{ \Psi}_{11.2_{ii}}}{\sqrt{1+{\bar{ \eta}_{1.2_i}}^2/\bar{ \Psi}_{11.2_{ii}}}} \right\} \text{d} x\\
    &= \exp \left \{ \dfrac{(\xi_{1.2_i}+\bar{\tau}_{1.2} {\bar{ \eta}_{1.2_i}})^2 h_i}{1-2(\bar{ \Psi}_{11.2_{ii}}+{\bar{ \eta}_{1.2_i}}^2) h_i} \right \} \dfrac{1}{\Phi(\bar{\tau}_{1.2})}   \dfrac{1}{\sqrt{2\pi} \sqrt{\bar{ \Psi}_{11.2_{ii}}+{\bar{ \eta}_{1.2_i}}^2}} \\
    &~~~~~\times \int_{\mathbb{R}} x^2 \exp \left [ -\dfrac{1}{2} \dfrac{\left\{x - \frac{\xi_{1.2_i} +\bar{\tau}_{1.2} {\bar{ \eta}_{1.2_i}}}{1-2(\bar{ \Psi}_{11.2_{ii}} +{\bar{ \eta}_{1.2_i}}^2) h_i} \right\} ^2}{\frac{\bar{ \Psi}_{11.2_{ii}}+{\bar{ \eta}_{1.2_i}}^2}{1-2(\bar{ \Psi}_{11.2_{ii}}+{\bar{ \eta}_{1.2_i}}^2)h_i}}  \right ] 
    \Phi \left\{ \dfrac{\bar{\tau}_{1.2} + {\bar{ \eta}_{1.2_i}} (x - \xi_{1.2_i})/\bar{ \Psi}_{11.2_{ii}}}{\sqrt{1+{\bar{ \eta}_{1.2_i}}^2/\bar{ \Psi}_{11.2_{ii}}}} \right\} \text{d}x\\
    &=\exp \left \{ \dfrac{(\xi_{1.2_i}+\bar{\tau}_{1.2} {\bar{ \eta}_{1.2_i}})^2 h_i}{1-2(\bar{ \Psi}_{11.2_{ii}}+{\bar{ \eta}_{1.2_i}}^2)h_i} \right \} \dfrac{1}{\sqrt{1-2(\bar{ \Psi}_{11.2_{ii}} +{\bar{ \eta}_{1.2_i}}^2)h_i}} \dfrac{\Phi(\Tilde{\tau}_i)}{\Phi(\bar{\tau}_{1.2})}\\
    &~~~~~\times \int_{\mathbb{R}} x^2 \dfrac{1}{\Phi(\Tilde{\tau}_i)} \phi(x;\Tilde{\xi}_i,\Tilde{\omega}_i^2) \Phi \{\Tilde{\alpha}_{0_i} + \Tilde{\alpha}_i \Tilde{\omega}_i^{-1} (x -\Tilde{\xi}_i) \} \text{d}x\\
    &= \dfrac{1}{\sqrt{1-2(\bar{ \Psi}_{11.2_{ii}} +{\bar{ \eta}_{1.2_i}}^2)h_i}} \exp \left \{ \dfrac{(\xi_{1.2_i}+\tau {\bar{ \eta}_{1.2_i}})^2 h_i}{1-2(\bar{ \Psi}_{11.2_{ii}}+{\bar{ \eta}_{1.2_i}}^2)h_i} \right \} \dfrac{\Phi(\Tilde{\tau}_i)}{\Phi(\bar{\tau}_{1.2})}\\
    &~~~~~\times \left \{ \Tilde{\xi}^2_i + \Tilde{\omega}^2_i - \Tilde{\tau}_i \frac{\phi(\Tilde{\tau}_i)}{\Phi(\Tilde{\tau}_i)} \Tilde{\omega}_i^2 \Tilde{\delta}_i^2 + 2 \frac{\phi(\Tilde{\tau}_i)}{\Phi(\Tilde{\tau}_i)} \Tilde{\xi}_i \Tilde{\omega} _i\Tilde{\delta}_i \right \},
\end{align*}
The last step is obtained from the moments of the extended skew-normal distribution from \cite{2013.A.A.A.C.CUP} (see Section 5.3.4). \qed
\end{proof}

\subsection*{S2~~~ Proof of Proposition 9}
\begin{proof}
We have:
\begin{align*}
    \mathbb{E} (Y_i Y_j) &= \int_{\mathbb{R}^2} x_i x_j \exp(h_i x_i^2/2) \exp(h_j x_j^2/2)  \dfrac{1}{\Phi(\bar{\tau}_{1.2})} \phi_2 (\bm x_{i,j}; \bm \xi_{i,j} + \bar{\tau}_{1.2} \bm \eta_{i,j}, \bm \Psi_{i,j} +\bm \eta_{i,j} \bm \eta_{i,j}^\top)\\&~~~~~\times \Phi \left \{ \dfrac{\bar{\tau}_{1.2} + \bm \eta_{i,j}^\top \bm \Psi_{i,j} ^{-1} (\bm x_{i,j} -\bm \xi_{i,j})}{\sqrt{1+\bm \eta_{i,j} ^\top \bm \Psi_{i,j} ^{-1} \bm \eta_{i,j}}} \right \} \text{d} \bm x_{i,j}, \quad \bm x_{i,j} = (x_i,x_j)^\top\\
    &= \dfrac{\sqrt{\det \{ (\bm \Omega_{i,j}^{-1}- \bm H_{i,j})^{-1} \}}}{\sqrt{\det (\bm \Omega_{i,j})}} \exp \left[ -\frac{1}{2}  \{ \tilde{{\bm \mu}}_{i,j} ^\top \bm \Omega_{i,j}^{-1} \tilde{{\bm \mu}}_{i,j} - \tilde{{\bm \mu}}_{i,j}^\top (\bm \Omega_{i,j} - \bm \Omega_{i,j} \bm H_{i,j} \bm \Omega_{i,j})^{-1} \tilde{{\bm \mu}}_{i,j} \} \right ]  \\
    &~~~~~\times \dfrac{1}{\Phi(\bar{\tau}_{1.2})} \int_{\mathbb{R}^2} x_i x_j  \phi_2 \Bigg ( \bm x_{i,j};  (\textbf{I}_2 - \bm \Omega_{i,j} \bm H_{i,j} )^{-1} \tilde{{\bm \mu}}_{i,j}, (\bm \Omega_{i,j}^{-1} -\bm H_{i,j})^{-1} \Bigg )\\
    &~~~~~~~~~~~~~~~~~~~~~~~~~~~~~~~~~~\times \Phi \left \{ \dfrac{\bar{\tau}_{1.2} + \bm \eta_{i,j}^\top \bm \Psi_{i,j} ^{-1} (\bm x_{i,j} -\bm \xi_{i,j})}{\sqrt{1+\bm \eta_{i,j} ^\top \bm \Psi_{i,j} ^{-1} \bm \eta_{i,j}}} \right \} \text{d} \bm x_{i,j}\\
    &=\dfrac{\sqrt{\det \{ (\bm \Omega_{i,j}^{-1}- \bm H_{i,j})^{-1} \}}}{\sqrt{\det (\bm \Omega_{i,j})}} \exp \left[ -\frac{1}{2}  \{ \tilde{{\bm \mu}}_{i,j} ^\top \bm \Omega_{i,j}^{-1} \tilde{{\bm \mu}}_{i,j} - \tilde{{\bm \mu}}_{i,j}^\top (\bm \Omega_{i,j} - \bm \Omega_{i,j} \bm H_{i,j} \bm \Omega_{i,j})^{-1} \tilde{{\bm \mu}}_{i,j} \} \right ]  \\
    &~~~~~\times \dfrac{1}{\Phi(\bar{\tau}_{1.2})} \int_{\mathbb{R}^2} x_i x_j \phi_2 (\bm x_{i,j}; \Tilde{\bm \xi}_{i,j} , \Tilde{\bm \Omega}_{i,j} ) \Phi \{ \Tilde{\alpha}_{0_{i,j}} + \Tilde{\bm \alpha}_{i,j}^\top  \Tilde{\bm \omega}_{i,j}^{-1} (\bm x_{i,j} - \Tilde{\bm \xi}_{i,j}) \} \text{d} \bm x_{i,j} \\
    &= \dfrac{\sqrt{\det \{ (\bm \Omega_{i,j}^{-1}- \bm H_{i,j})^{-1} \}}}{\sqrt{\det (\bm \Omega_{i,j})}} \exp \left[ -\frac{1}{2}  \{ \tilde{{\bm \mu}}_{i,j} ^\top \bm \Omega_{i,j}^{-1} \tilde{{\bm \mu}}_{i,j} - \tilde{{\bm \mu}}_{i,j}^\top (\bm \Omega_{i,j} - \bm \Omega_{i,j} \bm H_{i,j} \bm \Omega_{i,j})^{-1} \tilde{{\bm \mu}}_{i,j} \} \right ]\\
    &~~~~~\times \dfrac{\Phi(\Tilde{\tau}_{i,j})}{\Phi(\bar{\tau}_{1.2})}  \Bigg \{ (\Tilde{\bm \Omega}_{{i,j}})_{12} - \Tilde{\tau}_{i,j} \frac{\phi(\Tilde{\tau}_{i,j})}{\Phi(\Tilde{\tau}_{i,j})} (\Tilde{\bm \omega}_{{i,j}})_{11} (\Tilde{\bm \omega}_{{i,j}})_{22} (\Tilde{\bm \delta}_{{i,j}})_{1} (\Tilde{\bm \delta}_{{i,j}})_{2} + \xi_{1.2_i} \xi_{1.2_j} \\&~~~~~~~~~~~~~~~~~~~~~~~~~~+ \frac{\phi(\Tilde{\tau}_{i,j})}{\Phi(\Tilde{\tau}_{i,j})}\xi_{1.2_i} (\Tilde{\bm \omega}_{{i,j}})_{22} (\Tilde{\bm \delta}_{{i,j}})_{2} +  \frac{\phi(\Tilde{\tau}_{i,j})}{\Phi(\Tilde{\tau}_{i,j})} \xi_{1.2_j} (\Tilde{\bm \omega}_{{i,j}})_{11} (\Tilde{\bm \delta}_{{i,j}})_{1} \Bigg\}.
\end{align*}
The last step is obtained from the moments of the extended skew-normal distribution from \cite{2013.A.A.A.C.CUP} (see Section 5.3.4). \qed
\end{proof}

\subsection*{S3~~~ Wine Dataset Parameter Estimates}
\definecolor{MineShaft}{rgb}{0.2,0.2,0.2}
\begin{table}[h!]
\caption{Estimates of the $\mathcal{SNTH}$ distribution and the skew-$t$ distribution fitted to the wine dataset, along with the number of parameters in the model, maximized log-likelihood, the model AIC, and the model fitting time. The $\mathcal{SNTH}$ MLEs are estimated using the $\mathcal{SNTH}$ EM estimates as the optimization initial parameters. The standard error of each parameter is reported in parentheses. The computer used is a MacBook Pro (Retina, 16-inch), Processor 2.3 GHz 8-Core Intel Core i9, Memory 16 GB 2667 MHz DDR4.}
\centering
\begin{adjustbox}{max width=\textwidth}
\begin{tblr}{
  cells = {c},
  row{3} = {fg=MineShaft},
  row{4} = {fg=MineShaft},
  row{6} = {fg=MineShaft},
  row{7} = {fg=MineShaft},
  cell{2}{1} = {r=3}{},
  cell{2}{3} = {r=3}{},
  cell{2}{4} = {r=3}{},
  cell{2}{5} = {r=3}{},
  cell{2}{6} = {r=3}{},
  cell{5}{1} = {r=3}{},
  cell{5}{3} = {r=3}{},
  cell{5}{4} = {r=3}{},
  cell{5}{5} = {r=3}{},
  cell{5}{6} = {r=3}{},
  cell{8}{1} = {r=3}{},
  cell{8}{3} = {r=3}{},
  cell{8}{4} = {r=3}{},
  cell{8}{5} = {r=3}{},
  cell{8}{6} = {r=3}{},
  vlines,
  hline{1-2,5,8,11} = {-}{},
}
Model                & Parameter estimates                                                                                                                                                     & \# param & log-likelihood & AIC    & Time        \\
$\mathcal{SNTH}$ EM  & $\widehat{\bm \xi} = (83.2(2.27), 38.0(4.74), 7.6( 0.43))^\top$, $\widehat{\bm \omega} = \text{diag}(4.8(1.05),9.4(3.54), 0.9(0.13))$,~                                 & $15$         & $-731.6$       & $1493$ & $0.63$ s \\
                     & $\widehat{\Bar{\bm \Psi}} = \begin{psmallmatrix}1 & -0.36(0.44 ) & -0.09(0.20)\\ -0.36 & 1 & -0.09(0.21)\\ -0.09 & -0.09 & 1 \end{psmallmatrix}$,~                      &              &                &        &             \\
                     & $\widehat{\bm \eta} = (1.6(0.94), 3.7(2.07), 0.1(0.51))^\top$, $\widehat{\bm h} = (0.13(0.16),0.01(0.01),0.16(0.10))^\top$                                              &              &                &        &             \\
$\mathcal{SNTH}$ MLE & $\widehat{\bm \xi} = (84.4(1.67), 47.3(7.21), 7.8(0.28))^\top$, $\widehat{\bm \omega} = \text{diag}(5.2(1.12), 15.4(4.49), 0.9(0.13))$,                                 & $15$         & $-721.8$       & $1474$ & $2.12$ s \\
                     & $\widehat{\Bar{\bm \Psi}} = \begin{psmallmatrix}1 & -0.12(0.25 ) & -0.09(0.20)\\ -0.12 & 1 & -0.05(0.20)\\ -0.09 & -0.05 & 1\end{psmallmatrix}$,~                       &              &                &        &             \\
                     & $\widehat{\bm \eta} = (1.1(0.54), 1.3(0.89), -0.01( 0.33))^\top$, $\widehat{\bm h} = (0.26(0.19),0.11 (0.11),0.17(0.10))^\top$                                          &              &                &        &             \\
Skew-$t$             & $\widehat{\bm \xi} = (79.7(1.40), 60.4(4.95), 7.8(0.22 ))^\top$, ~                                                                                                      & $13$         & $-733.1$       & $1492$ & $0.20$ s  \\
                     & $\widehat{\bm \Omega} = \begin{psmallmatrix}\\ 237.0(74.54 ) &110.9(72.97)& -0.7(2.89)\\ 110.9& 522.3(130.74)& -0.8(3.05)\\ -0.7& -0.8 & 0.9(0.21) \end{psmallmatrix}$, &              &                &        &             \\
                     & $\widehat{\bm \alpha} = (4.31(1.67), 0.05(0.36), 0.18(0.38))^\top$,~$\widehat{\nu} = 3.4(0.95)$                                                                         &              &                &        &             
\end{tblr}
\end{adjustbox}
\end{table}

\subsection*{S4~~~ Windspeed Dataset Parameter Estimates}
\definecolor{MineShaft}{rgb}{0.2,0.2,0.2}
\begin{table}[htp!]
\caption{Estimates of the $\mathcal{SNTH}$ distribution and the skew-$t$ distribution fitted to the wind speed dataset, along with the number of parameters in the model, maximized log-likelihood, the model AIC, and the model fitting time. The $\mathcal{SNTH}$ MLEs are estimated using the $\mathcal{SNTH}$ EM estimates as the optimization initial parameters. The standard error of each parameter is reported in parentheses. The computer used is the same as above.}
\centering
\begin{adjustbox}{max width=\textwidth}
\begin{tblr}{
  cells = {c},
  row{3} = {fg=MineShaft},
  row{4} = {fg=MineShaft},
  row{6} = {fg=MineShaft},
  row{7} = {fg=MineShaft},
  cell{2}{1} = {r=3}{},
  cell{2}{3} = {r=3}{},
  cell{2}{4} = {r=3}{},
  cell{2}{5} = {r=3}{},
  cell{2}{6} = {r=3}{},
  cell{5}{1} = {r=3}{},
  cell{5}{3} = {r=3}{},
  cell{5}{4} = {r=3}{},
  cell{5}{5} = {r=3}{},
  cell{5}{6} = {r=3}{},
  cell{8}{1} = {r=3}{},
  cell{8}{3} = {r=3}{},
  cell{8}{4} = {r=3}{},
  cell{8}{5} = {r=3}{},
  cell{8}{6} = {r=3}{},
  vlines,
  hline{1-2,5,8,11} = {-}{},
}
Model                & Parameter estimates                                                                                                                                              & \# param & log-likelihood & AIC    & Time        \\
$\mathcal{SNTH}$ EM  & $\widehat{\bm \xi} = (0.12(2.70), -0.42(0.02), -1.77(0.51))^\top$, $\widehat{\bm \omega} = \text{diag}(1.44(0.18), 0.19(0.02), 1.85(0.20))$,~                    & $15$     & $-1633.4$      & $3297$ & $0.84$ s \\
                     & $\widehat{\Bar{\bm \Psi}} = \begin{psmallmatrix}1 & 0.55(0.04) & 0.74(0.03)\\ 0.55 & 1 & 0.05(0.07)\\ 0.74 & 0.05 & 1 \end{psmallmatrix}$,~                      &          &                &        &             \\
                     & $\widehat{\bm \eta} = (-0.10(2.37), 0.78(0.15), 1.11(0.46))^\top$, $\widehat{\bm h} = (0.00(0.02), 0.53(0.11), 0.03(0.03))^\top$                                 &          &                &        &             \\
$\mathcal{SNTH}$ MLE & $\widehat{\bm \xi} = (-0.99(0.22), -0.43(0.03), -1.14(0.33))^\top$, $\widehat{\bm \omega} = \text{diag}(1.25(0.10), 0.20(0.02), 2.00(0.14))$,                    & $15$     & $-1622.3$      & $3274$ & $4.31$ s \\
                     & $\widehat{\Bar{\bm \Psi}} = \begin{psmallmatrix}1 & 0.35(0.06) & 0.54(0.05)\\ 0.35 & 1 & 0.03(0.07)\\ 0.54 & 0.03 & 1\end{psmallmatrix}$,~                       &          &                &        &             \\
                     & $\widehat{\bm \eta} = (0.99(0.29), 0.89(0.18), 0.59(0.21))^\top$, $\widehat{\bm h} = (0.00(9.52\times10^-9), 0.42(0.10), 0.07(0.02))^\top$                       &          &                &        &             \\
Skew-$t$             & $\widehat{\bm \xi} = (-0.56(0.11), -0.74(0.03), -0.55(0.19))^\top$, ~                                                                                            & $13$     & $-1702.9$      & $3432$ & $0.42$ s \\
                     & $\widehat{\bm \Omega} = \begin{psmallmatrix}\\ 1.83(0.21) &0.66(0.13)& 1.80(0.24)\\ 0.66 & 0.78(0.12)& 0.46(0.16)\\ 1.80& 0.46 & 4.77(0.47) \end{psmallmatrix}$, &          &                &        &             \\
                     & $\widehat{\bm \alpha} = (-0.61(0.27) , 5.46(0.81) , 0.34(0.24))^\top$,~$\widehat{\nu} = 6.76(1.50)$                                                              &          &                &        &             
\end{tblr}
\end{adjustbox}
\end{table}

\end{document}